\documentclass[12pt]{article}
\usepackage{geometry}
\usepackage{amssymb}
\usepackage{amstext,enumerate}
\usepackage{bm}
\usepackage{graphics}
\usepackage{graphicx}
\usepackage{amsthm}
\usepackage{amsmath}
\usepackage{latexsym}
\usepackage{rotate}
\usepackage{lscape}
\usepackage{color}
\usepackage{epsfig,epsf,psfrag}
\usepackage{sectsty}
\usepackage{graphics}
\usepackage{epsfig}
\usepackage{multirow}
\usepackage{float}
\usepackage{graphicx}
\usepackage{amsmath}
\usepackage{url}
\usepackage{amsthm}
\usepackage{amssymb}
\usepackage{sectsty}
\usepackage{epsfig,natbib}
\usepackage{rotate}
\usepackage{lscape}
\usepackage{setspace}
\usepackage{subcaption}
\usepackage{float}
\usepackage{graphicx}
\usepackage{mathtools}
\usepackage{amstext,amssymb,amsfonts,amscd,color}
\usepackage{graphicx}
\usepackage{epsfig, natbib}
\usepackage{threeparttable}
\usepackage{cases,color}
\usepackage{multirow}
\usepackage{verbatim}
\usepackage{booktabs}
\usepackage{multicol}
\usepackage{bbm}
\usepackage{bm}
\usepackage{subcaption}
\usepackage[ruled,vlined]{algorithm2e}
\usepackage{comment}
\usepackage{multirow}
\usepackage{mathrsfs}
\usepackage{float}

\usepackage{xparse}
\usepackage{tikz}
\usetikzlibrary{positioning}
\usetikzlibrary{plotmarks}
\usetikzlibrary{matrix}
\usepackage{pgfplots}
\usetikzlibrary{matrix,backgrounds}
\pgfdeclarelayer{myback}
\tikzset{mycolor/.style = {dashed,rounded corners,line width=1bp,color=#1}}%
\tikzset{myfillcolor/.style = {draw,fill=#1}}%
\tikzset{
	declare function={
		normcdf(\x,\m,\s)=1/(1 + exp(-0.07056*((\x-\m)/\s)^3 - 1.5976*(\x-\m)/\s));
	}
}
\sectionfont{\centering\bf\sf\normalsize}
\subsectionfont{\sf\normalsize}

\renewcommand{\baselinestretch}{1.6} %{2} %{1.4}
\newcommand{\single}{\renewcommand{\baselinestretch}{1.2}\normalsize}
\newcommand{\double}{\renewcommand{\baselinestretch}{1.63}\normalsize}

  \renewenvironment{thebibliography}[1]{
    \begin{oldthebibliography}{#1}
      \setlength{\parskip}{0ex}
      \setlength{\itemsep}{0ex}
  }
  {
    \end{oldthebibliography}
  }

\newcommand{\bea}{\begin{eqnarray*}}
\newcommand{\eea}{\end{eqnarray*}}
\newcommand{\be}{\begin{eqnarray}}
\newcommand{\ee}{\end{eqnarray}}
\newcommand{\ed}{\end{document}}

\newcommand{\btab}{\begin{tabular}}
\newcommand{\etab}{\end{tabular}}

\newcommand{\bi}{\begin{itemize}}
\newcommand{\ei}{\end{itemize}}
\newcommand{\bfi}{\begin{figure}}
\newcommand{\efi}{\end{figure}}
\newcommand{\ben}{\begin{enumerate}}
\newcommand{\een}{\end{enumerate}}
\newcommand{\bay}{\begin{array}}
\newcommand{\eay}{\end{array}}

\def\vs{\vspace{.5cm}}

\definecolor{DarkBlue}{rgb}{0,.08,.45}
\definecolor{DarkRed}{rgb}{.7,0,.4}

\def\hg #1 {\texcolor{cyan}{{\it Hans:}   #1}}

\def\bco{\iffalse}

\def\cp{\citep}

\newcommand{\no}{\noindent}
\newcommand{\bc}{\begin{center}}
\newcommand{\ec}{\end{center}}

\newcommand{\bsp}{\begin{split}}
\newcommand{\esp}{\end{split}}
\newcommand{\bdes}{\begin{description}}
\newcommand{\edes}{\end{description}}
\newcommand{\bass}{\begin{assumption}}
\newcommand{\eass}{\end{assumption}}
\newcommand{\bthm}{\begin{theorem}}
\newcommand{\ethm}{\end{theorem}}
\newcommand{\blem}{\begin{lemma}}
\newcommand{\elem}{\end{lemma}}

\def\bco{\iffalse}

\def\cp{\citep}

\DeclareMathOperator*{\argmin}{argmin}

\bibpunct{(}{)}{;}{a}{}{,}
\newtheorem{theorem}{Theorem}
\newtheorem{lemma}{Lemma}

\begin{document}
\thispagestyle{empty} \single \bc {\bf \sc \Large Autoregressive optimal transport models}
%$^{\dagger *}$}
\vspace{0.15in}\\

Changbo Zhu \\
Department of Applied and Computational Mathematics and Statistics \\ University of Notre Dame \\
Notre Dame, IN 46556 USA \vspace{0.1in} \\
Hans-Georg M\"uller \\
Department of Statistics \\ University of California, Davis \\
Davis, CA 95616 USA \ec \centerline{30 March 2022}

%Changbo Zhu and  Hans-Georg M\"uller \\
%Department of Statistics, University of California, Davis \\
%Davis, CA 95616 USA \ec \centerline{30 March 2022}

%\vspace{0.1in} \thispagestyle{empty}
%\bc{\bf \sf ABSTRACT} \ec \vspace{-.1in} \no 
%\setstretch{1} 

\begin{abstract}
Series of univariate distributions indexed by equally spaced time points are ubiquitous in applications and their analysis constitutes one of the challenges of the emerging field of distributional data analysis. To quantify such distributional time series, we propose a class of intrinsic autoregressive models that operate in the space of optimal transport maps. The autoregressive transport models that we introduce here are based on regressing optimal transport maps on each other, where predictors can be
 transport maps from an overall barycenter to a current distribution or transport maps between past consecutive distributions of the  distributional time series.  Autoregressive transport models and their associated distributional regression models specify the link between predictor  and response  transport maps by moving along geodesics in Wasserstein space. These models emerge as natural extensions of the classical autoregressive models in Euclidean space. Unique stationary solutions of  autoregressive transport models are  shown to exist under a geometric moment contraction condition of  \cite{wu2004limit}, using properties of  iterated random functions.  We also discuss an extension to a varying coefficient model for first order autoregressive transport models. In addition to simulations, the proposed models are illustrated  with distributional time series of house prices across U.S. counties  and annual summer temperature distributions.\\   % of stock returns across the S\&P 500 stock index.\\
\end{abstract}

\no {KEY WORDS:\quad Distributional Data Analysis, Distributional Regression, Distributional Time Series, Iterated Random Function, Optimal Transport, Wasserstein space}.
\thispagestyle{empty} \vfill
\noindent \vspace{-.2cm}\rule{\textwidth}{0.5pt}\\
{\small Research supported in part by  NSF DMS-2014626 and NIH Echo UH3OD023313} 

\newpage
\pagenumbering{arabic} \setcounter{page}{1} \double

%\bc {\bf \sf 1.\quad INTRODUCTION}\sm \ec \rs

\section{INTRODUCTION}

\no Distributional data analysis (DDA) deals with data that include random distributions as data elements. While such data  are prevalent in many applied problems \citep{mena:18,mata:21}, this area is still in its early development.   An important instance where one encounters distributional  data arises for sequences of dependent distributions 
that are indexed by discrete time.  Such distributional time series are ubiquitous. For instance, the distribution of the log returns of the stocks  included in a stock  index is expected to contain more information than the  index itself, which only conveys  the mean of the distribution but not any further information inherent in the distribution such as quantiles.  Elucidating the nature of such financial time series is for example of interest for risk management \citep{bekierman2016mixed, kokoszka2019forecasting}. We will illustrate the proposed methods with the time series of distributions of %stocks in the S\&P index and also with 
 house prices that are formed from U.S. county house price data  and may inform economic policy \citep{oikarinen2018us, bogin2019local} and also with time series of annual distributions of temperatures aggregated over the summer, where a rise in night time temperatures and more frequent extremes have been related to global warming.

Other pertinent examples include the analysis of  sequences of age-at-death distributions over calendar years, which is instrumental for  the study of human longevity \citep{mazzuco2015fitting, shang2017grouped, ouellette2011changes} and also the study of the distributions of correlations between pairs of voxels within brain regions that can be derived from fMRI Bold signals 
\citep{petersen2016functional}, where such distributions may be observed repeatedly for the same subject in longitudinal studies that include fMRI brain imaging and where measurements are taken at regular time intervals.  %Age-at-death distributional data are widely available and are featured to motivate the proposed methods,  in addition to  patterns of house price distributional time series that may inform economic policies \citep{oikarinen2018us, bogin2019local} 

Distributions can be equivalently represented as either density, quantile or cumulative distribution functions, assuming that all of these exist. Each of these representations comes with certain constraints (for example, density functions are nonnegative and integrate to 1). An important observation is that  the spaces where these objects live are nonlinear. As a consequence, common statistical tools that are available in linear function spaces  such as the Hilbert space $L^2$ that is utilized in functional time series analysis \citep{bosq:00}  
are inadequate and there is a need for the development of adequate statistical methodology. It is the goal of this paper to contribute to the development of autoregressive models for one-dimensional distributions, given that autoregressive models are  popular in time series analysis and have been also considered for distributional time series  in recent work based on mapping to tangent spaces in the Wasserstein manifold  \citep{mull:20:7,zhang2020wasserstein}.

Existing approaches for distributional regression are based on various  transformation approaches that include mapping the distributions into a Hilbert space as implemented in the log quantile distribution approach \citep{kokoszka2019forecasting, petersen2016functional} or through logarithmic maps in the Wasserstein manifold \citep{mull:20:7}, where one 
uses the Wasserstein metric in the distribution space and maps the distributions to a tangent space that is a $L^2$ space, anchored at a suitable distribution, often chosen as a barycenter. One then can implement functional regression models in the 
ensuing $L^2$ space, and analyze these models by employing parallel transport. While the log quantile distribution transformation approach to distributional regression can lead to large metric deformations, the tangent bundle approach is extrinsic and there are some difficulties with the required inverse exponential maps that are caused  by the injectivity requirement that one needs to numerically address 
in finite sample situations. % only for a convex set around the barycenter in the tangent space the inverse of the log map is well defined. 
Various projection methods have been devised to tackle  this problem  \citep{bigo:17,mull:20:7,pego:21}, while in other recent work on extrinsic modeling it has  been ignored \citep{zhang2020wasserstein}, which can lead to inferior performance.

Since the autoregressive transport models we propose here are intrinsic, they bypass the construction of a tangent space and the ensuing problems with mapping and projection. 
%and not all elements in the $L^2$ tangent space can be mapped back to the Wasserstein manifold so that somewhat artificial modifications are needed to stay within a convex set that
%is defined by the injectivity requirement  
%Just as we finalized the write-up of this paper a simplified  distributional regression models was proposed for the independent case  \citep{ghod:21}. 
For the case where only the responses are distributional but predictors are vectors, one can apply Fr\'echet regression \citep{mull:19:3}. % demonstrating that this is an area in rapid development. 
Concurrently with this paper, a distributional regression model with one predictor was proposed for the independent case  \citep{ghod:21}, where the goal is to learn a single best transport map that maps the predictor distribution to the response distribution, so the model parameter is the transport map learned from the data. This is akin to fitting a linear regression model where the slope parameter is fixed at 1. A nice feature of this simple model is that finding the best transport map has been shown to be equivalent to an isotonic regression problem, which can be solved by standard optimization techniques.

In this paper, we propose a novel class of intrinsic distributional regression models for the autoregressive modeling of distributional time series. The proposed models are based on transports of the probability measures. The most popular notion of transport of distributions is optimal transport, which commonly refers  to moving distributions along geodesics in the Wasserstein  space, i.e., the space of distributions equipped with the Wasserstein metric.  The key innovation  in the proposed regression model is that both predictors and responses are taken to be transports of distributions, rather than distributions themselves, in contrast to the currently available distributional regression models. Our  focus is on univariate distributions with bounded support on the real line, which is the most relevant case in statistical data analysis. % not only for theoretical conveniences but also for practical purposes. 
Moreover, in data applications the  distributions that are part of the data sample  are not known a priori and in practice need  to be estimated from data they generate by nonparametric methods. Such methods  include  kernel density estimation and related approaches,  and for practical implementations a  bounded interval that defines the domain needs  to be fixed beforehand. For the relatively uncommon applications that require the distributions to be supported on the entire real line it is common practice to truncate the target distribution at a large enough interval and to target the truncated distribution, with negligible error. % if the interval is chosen large enough. 

Typical examples for predictor or response transports are the transports defined by pushing distributional barycenters (Fr\'echet means) forward to individual distributions, and  in the distributional time series framework also the transports pushing   the distribution at time $(j-1)$ to that at  time $j$,  which may serve as predictors for the transports pushing the distribution at time $j$ forward to that at  time $j+1$. The idea of considering transports rather than distributions as predictors or responses, especially transports from barycenters,  is motivated by the classical simple linear regression model for scalar predictors and responses.  This model  can be written in transport form as $E(Y-\mu_Y|X-\mu_X)=\beta(X-\mu_X)$, where $\mu_Y=EY,\, \mu_X=EX$ and 
$\beta$ is the slope parameter, where  both responses $Y-\mu_Y$ and predictors $X-\mu_X$ can be interpreted as transports pushing the barycenters $\mu_X, \mu_Y$ forward to the individual data $X,Y$. As we show here, this transport interpretation of linear regression provides a natural approach to extend classical regression to distributional regression modeling by regressing transports on each other.

We focus here on  autoregressive transport models (ATM) that permit an inherent  geometrical interpretation by relating geodesics in transport space to each other, where  a first order ATM (or ATM(1)) connects transports related to time $(j-1)$ to transports related to time $j$.  As in the independent case, geometric transport interpretations can also be applied to the case of  scalar or vector time series in Euclidean space, motivating the extension to distributional time series where transports
are very natural.
 One of our main results is the existence and uniqueness of a  stationary solution for ATM(1) processes, for which we utilize  the geometric-moment contraction condition \citep{wu2004limit} for iterated random maps. While the proposed models generally involve scalar coefficients and 
 are well interpretable, we also consider an extension for  ATM(1)  processes, where the ATM features a  functional rather than scalar coefficient. We show that this functional coefficient can  also be estimated consistently from samples. The definition  of ATMs of order $p$  (ATM($p$)) is obtained as a  straightforward extension;  these models possess a  multi-layer structure. We demonstrate that  ATMs  are useful to capture the dynamic evolution  of distributions for both real and synthetic data.

The rest of the paper is organized as follows. In Section 2, we provide some preliminary discussion on basic concepts such as Wasserstein space, optimal transport maps and geodesics. We also introduce addition and scalar multiplication operations for the space of  transport maps. Section 3 includes methodology and theoretical results for  ATM(1) models.  Extensions to ATM($p$) models and versions of ATM(1) models with functional coefficients are the topics  of 
Sections 4 and 5.  Numerical considerations and applications to simulated and real data can be found in Section 6.  Conclusions are in Section 7, while the Appendix contains proofs and technical details.

%\bc {\bf \sf 2.\quad THE SPACE OF TRANSPORT MAPS}\sm \ec \rs
\section{THE SPACE OF TRANSPORT MAPS}

\no Defining  $\mathcal{W}$ to be the set of probability distributions on $(\mathcal{S}, \mathscr{B}(\mathcal{S}))$ with finite second moments, where $\mathcal{S} = [s_1,s_2]$ is a bounded closed interval in $\mathbb{R}$ and  $\mathscr{B}(\mathcal{S})$ is the Borel $\sigma$-algebra on $\mathcal{S}$, let
$\mathcal{W}$ 
%\begin{align*}
%\mathcal{W} = \left\lbrace \mu \in \mathcal{P}(\mathcal{S}) : \int_{\mathcal{S}} x^2 d \mu(x) < \infty \right\rbrace .
%\end{align*}
be the set of probability measures on $\mathcal{S}$. We assume there is an underlying probability space $(\Omega, \mathcal{A}, P)$ 
of $\mathcal{W}$-valued random variables that induces a probability measure on the space $\mathcal{W}$ with respect to which we can calculate moments for random variables taking values in $\mathcal{W}$.

For any measurable function $T:\mathcal{S} \rightarrow \mathcal{S}$ and $\mu \in \mathcal{W}$, let $T_{\#}\mu$ denote the pushforward measure of $\mu$, i.e. for any $B \in \mathscr{B}(\mathcal{S})$, $T_{\#}\mu (B) = \mu( \{ x : T(x) \in B \} )$. If $\mu_1$ is absolutely continuous with respect to the Lebesgue measure, then the 2-Wasserstein metric ($d_{\mathcal{W}}$) on $\mathcal{W}$ can be written using the Monge formulation \citep{vill:03}
\begin{multline} \label{eq:monge}
d_{\mathcal{W}}(\mu_1, \mu_2)  = \inf_{T:T_{\#}\mu_1 = \mu_2} \left\{ \int_{\mathcal{S}} ( T(x) - x )^2 d\mu_1(x) \right\}^{1/2} \\
 =\left\{ \int_{\mathcal{S}} ( T_{12}(x) - x )^2 d\mu_1(x) \right\}^{1/2} 
 = \left\{ \int_{0}^1 ( F_{2}^{-1}(u) - F_1^{-1}(u) )^2 d u \right\}^{1/2}.
\end{multline}
Here  $\mu_1, \mu_2 \in \mathcal{W}$, $ F_1=F(\mu_1)$ and  $F_2=F(\mu_2) $ are the cumulative distribution functions (cdf)  of $ \mu_1 ,\mu_2 $ respectively, and 
 \begin{align*}
F_{1}^{-1}(u) := \inf \{ x \in \mathcal{S}: F_1(x) \geq u \}, \;
F_{2}^{-1}(u) := \inf \{ x \in \mathcal{S}: F_2(x) \geq u \}
 \end{align*}
are the corresponding quantile functions, defined as left-continuous inverses of the cdf.  
A map $T$ that satisfies $T_{\#}\mu_1 = \mu_2$ is a transport map from $\mu_1$ to $\mu_2$ and $T_{12} = F_2^{-1} \circ F_1$ is referred to as the optimal transport map that pushes the probability measure $\mu_1$ forward to the measure $\mu_2$.

For a nonempty interval $I \subset \mathbb{R}$, the length of a given curve $ \gamma:I \rightarrow \mathcal{W}$ is $L(\gamma) := \sup \sum_{i=1}^k d_{\mathcal{W}}(\gamma(t_{i-1}), \gamma(t_i))$, where the supremum is taken  over all $k \in \mathbb{N}$ and $t_0 \leq t_1 \leq \dots \leq t_k $ in $I$. For absolutely continuous  $\mu_1$,  McCann's interpolant \citep{mcca:97} is the curve $\gamma: [0,1] \rightarrow \mathcal{W}$ given by 
\begin{align*}
\gamma(a) =\left( id + a ( T_{12} - id ) \right)_{\#} \mu_1, 
\end{align*}
where $a \in [0,1]$ and $id$ is the identity map. McCann's interpolant is the geodesic in $\mathcal{W}$ that corresponds to the optimal transport from $\mu_1$ to $\mu_2$, where we do not distinguish  between this geodesic and the transport map  $T_{12}$; we note  that  $L(\gamma) = d_{\mathcal{W}}(\mu_1, \mu_2)$ and $\gamma$ has constant speed $d_{\mathcal{W}}(\gamma(a_1), \gamma(a_2)) = (a_2-a_1) d_{\mathcal{W}}(\mu_1, \mu_2)$ for any $ 0 \leq a_1 \leq a_2 \leq 1 $. 

Our focus is on a time series of distributions $ \{  \mu_{i} \}_{i=1,2, \dots, n} \subset \mathcal{W}$,  which is assumed to possess some stationarity properties, including stationarity of the mean. This  means that there exists a  common Fr\'{e}chet mean or barycenter $\mu_{\mathcal{F}}$, given by 
\begin{align*}
\mu_{\mathcal{F}} := \argmin_{\nu \in \mathcal{W}} \mathbb{E} d_{\mathcal{W}}^2(\nu, \mu_i) \text{ for all } i =1,2, \dots, n, 
\end{align*}
where existence and uniqueness are assured by the fact that the Wasserstein space for one-dimensional distributions is a Hadamard space  \citep{kloe:10}.

We now consider the space of all Lebesgue integrable functions on $\mathcal{S}$, $L^p(\mathcal{S}) = \{ f:\mathcal{S} \rightarrow \mathbb{R}\; | \; \| f \|_{\mathcal{L}^p} < \infty \} $, where $1 \leq p < \infty$, $\lambda$ is the Lebesgue measure and $ \| f \|_{\mathcal{L}^p} := (\int_{\mathcal{S}} |f|^p d\lambda)^{1/p} $ is the usual $L^p$-norm. Define the set $\mathcal{T}$ as
\begin{align} \label{eq:inc}
\mathcal{T} = \left\lbrace T:\mathcal{S} \rightarrow \mathcal{S} \; | \; T(s_1)=s_1, T(s_2) =s_2, \,\, T \text{ is non-decreasing} \right\rbrace.
\end{align}
Since $\mathcal{T}$ is a closed subset of $L^p(\mathcal{S})$, it is a complete metric space with respect to the $L^p$-norm, i.e. the limit of every Cauchy sequence of points in $\mathcal{T}$ is still in $\mathcal{T}$.  In addition, $\mathcal{T} \subset L^p(\mathcal{S})$ can be equivalently identified as  
$
\mathcal{T}= \left\{ T: \mathcal{S} \rightarrow \mathcal{S} \; \left| \; T := F_{1}^{-1} \circ F_{2} \right. \right\}, 
$ 
where, as above,  $F_1, F_2$ are the cdfs of probability measures  $\mu_1, \mu_2 \in \mathcal{W}$. % $F^{-1}_1, F_{2}^{-1}$ are the corresponding quantile functions defined by the left-continuous inverse. 
Here, $F_1, F_2$ may not be continuous and are not necessarily strictly increasing. For any $T \in \mathcal{T}$, the representation $T=F_1^{-1} \circ F_2$ is not unique and one may choose $F_2$ to be the cdf of a uniform distribution, in which case
$T$ is represented by $F_1^{-1}$ only, which then is unique.  This not only induces a metric on $\mathcal{T}$ but also shows that   $\mathcal{W}$ and $\mathcal{T}$ are isometric with this induced metric. This isometry  induces a probability measure on   $\mathcal{T}$ that is inherited from the corresponding measure on $\mathcal{W}$.
Furthermore,  for every $T \in \mathcal{T}$, there exists a uniquely defined inverse transport map $T^{-1}\in \mathcal{T}$; for any given representation $T=F_1^{-1} \circ F_2$,  $T^{-1}=F_2^{-1} \circ F_1$.
%note that a function with countable discontinuities is Riemann integrable.

To build an autoregressive model for elements in $\mathcal{T}$, we introduce addition and scalar multiplication operations   in the transport space   $\mathcal{T}$
as follows. 
\begin{itemize}
	\item Addition: $T_1 \oplus T_2 := T_2 \circ T_1$, where $T_1, T_2 \in \mathcal{T}$.
	\item Scalar multiplication: For any $x \in \mathcal{S}$ and $T \in \mathcal{T}$, for any  $\alpha \in \mathbb{R}$ with $-1 \leq \alpha  \leq 1$, let 
		\begin{align*}
		\alpha \odot T (x) := \left\lbrace  \begin{array}{cc}
		x + \alpha (T(x)-x), & 0 < \alpha \leq 1 \\
		x, & \alpha=0 \\
		x + \alpha (x-T^{-1}(x)), & -1 \leq \alpha <0
		\end{array} \right. .
		\end{align*}
		For any  $|\alpha| > 1$, let $b = \lfloor |\alpha| \rfloor$, the integer part of $\alpha$, and set $a = |\alpha| - b$. We then define  a scalar multiplication in transport space by
		\begin{align*}
		\alpha \odot T (x) := \left\lbrace  \begin{array}{cl}
		(a \odot T) \circ \underbrace{ T\circ T \circ \dots \circ T}_{b \text{ compositions of } T} (x) , &  \alpha > 1  \\
		(a \odot T^{-1}) \circ \underbrace{T^{-1} \circ T^{-1} \circ \dots \circ T^{-1} }_{b \text{ compositions of } T^{-1} } (x) , &  \alpha < -1
		\end{array} \right.,
		\end{align*}
%		where  $T^{-1}$ is the inverse of $T.$ 
\end{itemize}

These operations are motivated as follows.  Addition of transports is defined as their simple concatenation, which is a straightforward  extension from the case of transports in $\mathbb{R}^p$,  where transports correspond to vectors $V$ that are added to a vector argument $c$, so that $T_V(c)=V+c.$  Consecutively applying two transport maps $T_{V_1}$ and then $T_{V_2}$ then means 
 adding the sum of the two vectors $V_1 + V_2$ to the argument vector $c$, so that
$T_{V_2}\circ T_{V_1}(c)=V_1 + V_2 +  c.$ 
 %then means adding the sum of the two transport vectors to the argument vector. %For two vectors $c, V_1 \in \mathbb{R}^p$, $c+V_1$ can be interpreted as the result of transporting $c$ with $V_1$. Here, $V_1$ can be viewed as the optimal transport map between $c$ and $c+V_1$. 
For scalar multiplication, given  $0 < \alpha <1$, a transport vector  $\alpha V$ defines the transport  $T_{\alpha V}(c)=\alpha V + c$ and therefore transports an argument vector $c$ to a point on the straight line (geodesic) between $c$ and $c + V$.  So if $T_1$ is the optimal transport that pushes $\mu_{\mathcal{F}}$ to $(T_1)_{\#} \mu_{\mathcal{F}}$, it is natural to define  $\alpha \odot T_1$  such that it pushes $\mu_{\mathcal{F}}$ to a distribution lying  on the geodesic from $ \mu_{\mathcal{F}} $ to $ (T_1)_{\#} \mu_{\mathcal{F}}$ where its location on the geodesic is characterized by a  fraction of length $\alpha$ when measuring length from the starting point $ \mu_{\mathcal{F}} $.
   When $\alpha$ is negative,  $ c+ \alpha V = c + |\alpha| (-V)$, where $-V$ can be   interpreted as the transport  map that pushes $c+V$ to $c$ and thus is the inverse transport $T_V^{-1}$  of the transport $T_V$. The obvious extension to optimal transport maps in distribution spaces then leads to the above definition of scalar multiplication, which is further illustrated in % The counterpart of vector $( -V_1)$ in $\mathcal{T}$ would be $T_1^{-1}$, which pushes $ (T_1)_{\#} \mu_{\mathcal{F}} $ to $\mu_{\mathcal{F}}$. See
 Figure \ref{fig:operationillustration}.
%For the addition of transport maps. Let $T_1, T_2 \in \mathcal{T}$, the transportation of $\mu_{\mathcal{F}}$ using $T_1 \oplus T_2$ can be written as $(T_1 \oplus T_2)_{\#} \mu_{\mathcal{F}} = (T_2)_{\#} ( (T_1)_{\#} \mu_{\mathcal{F}})$, which is equivalently as pushing distribution $\mu_{\mathcal{F}}$ by $T_1$ and $T_2$ sequentially. Similar situation happens when transporting $c \in \mathbb{R}^p$ using two vectors $V_1, V_2 \in \mathbb{R}^p$, i.e. $c+V_1+V_2$ can be interpreted as first transporting $c$ along $V_1$, then along $V_2$. 
A  distinction from  the vector space case is  that the addition $\oplus$ for optimal transport maps is not commutative. For scalar multiplication with factors $\alpha$ that are such that $|\alpha|>1$, if $\alpha$ is an integer we decompose the map $T^\alpha$ into an iterative sum of maps $T$,  and if $\alpha$ is not an integer we apply the integer part of $\alpha$ first and after this apply an additional transport map that is a scalar multiplication of $T$ with the left-over fractional part of $\alpha$.
\begin{figure}[t]
	\centering
	\begin{tikzpicture}
	%\draw[help lines,step=1] (0,0) grid (14,6);
	%\foreach \x in {0,1,2,3,4,5,6,7, 8, 9, 10, 11, 12, 13, 14}
	%\node[anchor=north] at (\x,0) {\x};
	%\foreach \y in {0,1,2,3,4, 5, 6}
	%\node[anchor=east] at (0,\y) {\y};
	
	\draw[very thick,->] (0,0) -- (6,0) ;
	\draw[very thick, ->] (0,0) -- (0,6) ;
	\draw[very thick,rounded corners=3mm] (7,3)--(7.5,5)--(8.4,5.3)--(8.6,5.8)--(11,6)--(12.5,5.7)--(14.1,4.6)--(14.7,3.3)--(14.9,2.5)--(14.6,0.9)--(13,0.2)--(11,0)--(9,0.2)--(8.3,0.3)--(7.3,1)--cycle;
	\filldraw (1,1) node[anchor=south east] {$c$};
    \filldraw [gray] (1,1) circle (2pt);
    \filldraw [gray] (3,5) circle (2pt);
	\filldraw (8,1) node[anchor=south east] {$\mu_{\mathcal{F}}$};
	\draw (3,6) node[anchor = south] {$\mathbb{R}^2$};
	\draw (10.5,6) node[anchor = south] {$\mathcal{W}$};
	\filldraw [gray] (8,1) circle (2pt);
	\draw[arrows={->[scale=1.5]}, dashed] (2.5,4) -- (3,5) node[anchor=south]  {$c + V_{1}$};
	\filldraw [gray] (11,5) circle (2pt);
	\draw (11,5) node[anchor=south] {$ (T_1)_{\#}\mu_{\mathcal{F}}$};
	\filldraw [gray] (10.1,4) circle (2pt);
	%\draw[dashed] (8,1) to[out=50,in=-100] (10.1,4);
	\draw[dashed,arrows={->[scale=1.5]}] (10.1,4) to[out=80,in=-150] (11,5);
	\draw[thick, arrows={->[scale=1.5]}, dashed, red] (8,1) to[out=50,in=-100,] (10.1,4);
	\filldraw (10.1,4) node[anchor=east] {$(\alpha \odot T_{1})_{\#} \mu_{\mathcal{F}}$};
	\draw[thick, arrows={->[scale=1.5]}, dashed, red] (1,1) -- (2.5,4) ;
	
	\draw[thick, arrows={->[scale=1.5]}, dashed, blue] (3,5) --(4.5, 1.5) ;
	\filldraw [gray] (2.5,4) circle (2pt);
	\draw (2.5,4) node[anchor=east] {$c+ \alpha V_{1} $};
	\draw (4.5,1.5) node[anchor=north west] {$c+ V_{1} + V_2$};
	\filldraw [gray] (4.5,1.5) circle (2pt);
	\draw[thick, arrows={->[scale=1.5]}, dashed, green] (1,1) -- (4.5, 1.5) ;
	
	\draw[thick, arrows={->[scale=1.5]}, dashed, blue] (11,5) to[out = -50, in =170] (12, 2.5);
	\draw (12,2.5) node[anchor=west] {$(T_1 \oplus T_2)_{\#}\mu_{\mathcal{F}}$};
	\filldraw [gray] (12,2.5) circle (2pt);
	\draw[thick, arrows={->[scale=1.5]}, dashed, green] (8,1) to[out=20, in=-120] (12, 2.5) ;
	\end{tikzpicture}
	\caption{Motivating the definition of the  addition $\oplus$ and scalar multiplication $\odot$ operations for $0 < \alpha <1$ in the Wasserstein optimal transport space for transports $T_1, T_2$ (right), while in 
	$\mathbb{R}^2$ optimal transports are defined by vectors  $V_1, V_2$ (left).}
	\label{fig:operationillustration}
\end{figure}

Observe that($\mathcal{T}$,$\oplus$) is a (non-Abelian) group with the identity map as identity.  For any $T \in \mathcal{T}$, the inverse is $T^{-1}$. By the definition of $\oplus$, we have
\begin{align*}
(T_1 \oplus T_2) \oplus T_3 = T_3 \circ (T_2 \circ T_1) = (T_3 \circ T_2) \circ T_1 = T_1 \oplus ( T_2 \oplus T_3 ),
\end{align*}
which entails the associativity of $\oplus$. Regarding the relation between $\odot$ and $\oplus$, distributive laws do not hold, i.e. there exists $\alpha, \beta \in \mathbb{R}$ and $T_1, T_2 \in \mathcal{T}$ such that
\begin{align*}
     \alpha \odot ( T_1 \oplus T_2 ) \neq (\alpha \odot T_1 ) \oplus (\alpha \odot T_2), \quad\quad
     (\alpha + \beta) \odot T_1 \neq (\alpha \odot T_1) \oplus (\beta \odot T_1).
\end{align*}

A simple example is as follows. Set  $\mathcal{S} = [0,1]$, $ T_1(x) = x^2 $, $T_2 (x) = (x+x^2)/2$ and $\alpha = 0.6$, $\beta =0.7$. Simple algebra shows that  $ (\alpha + \beta) \odot T_1 = 0.7x^2 + 0.3x^4 \neq 0.3(0.4x+0.6x^2) + 0.7 (0.4x+ 0.6x^2)^2 = (\alpha \odot T_1) \oplus (\beta \odot T_1).$ In addition, the coefficient of $x^4$ in the $4$th order polynomial (with respect to $x$) $\alpha \odot ( T_1 \oplus T_2 )$ is 0.3, while $x^4$ has coefficient $0.6^3$ in $ (\alpha \odot T_1 ) \oplus (\alpha \odot T_2) $, which indicates that $\alpha \odot ( T_1 \oplus T_2 ) \neq   (\alpha \odot T_1 ) \oplus (\alpha \odot T_2)$.

%\bc {\bf \sf 3. \quad AUTOREGRESSIVE TRANSPORT MODELS OF ORDER 1} \sm \ec \rs
\section{AUTOREGRESSIVE TRANSPORT MODELS OF ORDER 1}

%\noindent { \sf 3.1 \quad Model and Stationary Solution}
\subsection{Model and Stationary Solution}

\no We first consider a time series $\{ X_i \}_{i=1,2, \dots, n} \subset \mathbb{R}^p$ with constant mean $E[X_i] = \mu \in \mathbb{R}^p$. The vector autoregressive model of order 1 ($\text{VAR}$(1)) with scalar coefficient is
\begin{align}
\label{eq:ar1}
X_i - \mu = \beta (X_{i-1} - \mu) + \epsilon_i,
\end{align}
where $\beta \in \mathbb{R}$ and $\{ \epsilon_i \}_{i=1,2, \dots, n} \subset \mathbb{R}^p$ are the i.i.d innovations with mean 0. In this Euclidean time series model, the  vector $X_i - \mu$ can be interpreted as the optimal transport map pushing  $\mu$ to $X_i$, which provides the inspiration for the proposed ATM.  

In general metric spaces, differences cannot be formed and thus a direct extension of model \eqref{eq:ar1} is not feasible.  However, in transport spaces with uniquely defined optimal transports along geodesics we can reinterpret differences of elements in terms of such optimal transports. 
%Specifically, in Wasserstein space, the optimal transport map that pushes $\mu_{\mathcal{F}}$ to $\mu_i$ is given by 
%\begin{align*}
%$T_i = F_i^{-1} \circ F_{\mathcal{F}},$
%\end{align*}
%where $ F_{i}=F(\mu_i)$, $F_{\mathcal{F}}=F(\mu_{\mathcal{F}})$ are the cdfs of measures $\mu_i$, $\mu_{\mathcal{F}}$, respectively.
Specifically, in Wasserstein space, we define the difference between two distributions $\mu_2$ and $\mu_1$ to be the optimal transport map that pushes $\mu_{1}$ to $\mu_2$, i.e., 
\begin{align} \label{diff} 
\mu_2 \ominus \mu_1 = F_2^{-1} \circ F_{1},
\end{align}
where in \eqref{diff} $ F_{1}=F(\mu_1)$, $F_{2}=F(\mu_{2})$ are the cdfs of measures $\mu_1$, $\mu_{2}$, respectively. We also require appropriate generalizations for  the random innovations $\epsilon_i$ that now become random transports. %For this we note  that in $\mathbb{R}^p$, $c + \epsilon_i$ can be viewed as transporting $c$ along a random vector $\epsilon_i$. 
Extending the notion of  additive noise for Euclidean data, we model noise in transport space  as random transport maps  in $\mathcal{T}$ constrained in such a way that  their Fr\'echet mean (barycenter) is the identity transport. A noise contaminated version of a transport 
map $T \in \mathcal{T}$ is thus $T \oplus \epsilon$,  where $E(\epsilon)=id.$

Motivated by model \eqref{eq:ar1}, the autoregressive transport model of order 1 (ATM(1)) is  %then takes the following form,
\begin{align}
\label{eq:atmm}
T_i =  \alpha \odot T_{i-1} \oplus \varepsilon_{i}, \text{ where } T_i  = \mu_i \ominus \mu_{\mathcal{F}},
\end{align}
%Given a probability space $(\Omega, \mathcal{F}, \mathcal{P})$ the  
where $\alpha \in \mathbb{R}$ is the model parameter and  the  $\varepsilon_i$ are 
%$\{ \varepsilon_i : \Omega \rightarrow \mathcal{T} \}_{i \in \mathbb{Z}} $ are 
i.i.d random distortion transport maps with mean $E (\varepsilon_i) = id$. % and taking values in $\mathcal{T}$. To get some insights, 
The proposed ATM approximates the optimal transport map at time $t=i$ with the scaled transport map $\alpha \odot T_{i-1}$, in analogy to the VAR(1) model $ X_i -c = \beta (X_{i-1} - c) + \epsilon_i$, which can  be interpreted as approximating the optimal transport map $X_i-c$ with the scaled  transport map $ \beta (X_{i-1} - c) $; see  Figure \ref{fig:atmillustration} for an illustration. While  \eqref{eq:ar1}  provides the usual formulation of the VAR(1) model, another way to view the model is by relating past differences to current differences, i.e., model \eqref{eq:ar1} gives rise to the alternative model 
\begin{align}
\label{eq:ar1diff}
    X_i - X_{i-1} = \beta (X_{i-1} - X_{i-2}) + \epsilon_i. 
\end{align}

The difference $ X_i - X_{i-1} $ can be interpreted as the optimal transport map between $X_{i-1}$ and $X_i$. In Wasserstein space, autoregressive transport models of order 1 (ATM(1)) can analogously be built with optimal transports between adjacent distributions,
\begin{align}
\label{eq:atmd}
T_i =  \alpha \odot T_{i-1} \oplus \varepsilon_{i},  \text{ where } T_i = \mu_{i+1} \ominus \mu_i, 
\end{align}
where  the $\varepsilon_i$  % : \Omega \rightarrow \mathcal{T} \}_{i \in \mathbb{Z}} \subset \mathcal{T} $ 
are again i.i.d random distortion maps with $E (\varepsilon_i) = id$.

Next we show the existence of stationary solutions for models \eqref{eq:atmm} and \eqref{eq:atmd}. For any $S, T \in \mathcal{T}$, $1 \leq q  < \infty$  and random distortion map $ \varepsilon $, we utilize the distances $d_q(S,T) = \| S-T \|_{\mathcal{L}^q}$ on $\mathcal{T}$ and define  $\phi_{\varepsilon}, \widetilde{\phi}_{i,m} : \mathcal{T} \rightarrow \mathcal{T}$ by 
\begin{align*}
    \phi_{\varepsilon}( S ) =  \alpha \odot S \oplus \varepsilon, \qquad 
    \widetilde{\phi}_{i,m}(S)  = \phi_{\varepsilon_i} \circ \phi_{\varepsilon_{i-1}} \circ \dots \circ \phi_{\varepsilon_{i-m+1}}(S).
\end{align*}
Then under a suitable contraction condition, stationary solutions exist.

\begin{theorem}\label{thm1}
Suppose there exists $ \eta >0$, $ S_0 \in \mathcal{T}$, $C >0$ and $r \in (0,1)$ such that 
\begin{align}
\label{eq:moment}
E \left[ d_q^{\eta} \left( \widetilde{ \phi}_{i, m}(S_0), \widetilde{\phi}_{i, m}(T) \right) \right] \leq C r^{m} d_q^{\eta} ( S_0, T)
\end{align}
holds for a given  $1\leq q < \infty$ and all $m \in \mathbb{N}$ and  all $ T\in \mathcal{T}$. Then, for all $S \in \mathcal{T}$, $\widetilde{T}_i := \lim\limits_{m \rightarrow \infty} \widetilde{ \phi}_{i, m}(S) \in \mathcal{T}$ exists almost surely and does not depend on $S$. In addition, $\widetilde{T}_i$ is a stationary solution to the following system of stochastic transport equations 
\begin{align}
\label{eq:arr1}
T_i =  \alpha \odot T_{i-1} \oplus \varepsilon_{i}, \; i \in \mathbb{Z}
\end{align}
and is unique almost surely.
\end{theorem}

The proof utilizes the theory of iterated random function systems \citep{diaconis1999iterated}, where a crucial element is 
the 
geometric-moment contraction condition \eqref{eq:moment} of \cite{wu2004limit}. Regarding  sufficient conditions for  \eqref{eq:moment} when $q=1$, easy algebra shows that $ d_1( \alpha \odot S, \alpha \odot T ) = \alpha d_1(  S,  T ) $ for a positive $\alpha$. From the corresponding result on the $L_1$ distance of cdfs \cite[see, e.g.,][]{shor:09},
one immediately finds   % it follows similarly from Proposition 1.5.2 \citep{panaretos2020invitation}
\begin{align*}
d_1(S,T)=\int_{\mathcal{S}} |S(x) - T(x)| dx = \int_{\mathcal{S}} |S^{-1}(x) - T^{-1}(x)| dx,
\end{align*}
which then entails that $d_1( \alpha \odot S, \alpha \odot T ) = -\alpha d_1(  S,  T )$ when $\alpha <0$. Suppose for any $S, T \in \mathcal{T}$,
%\begin{align*}
$E[ d_1( \varepsilon_i\circ S, \varepsilon_i \circ T )] \leq L d_1( S, T ),$
%\end{align*}
where $L$ is some positive constant such that $\alpha L \in (0,1) $, then \eqref{eq:moment} is seen to hold with $\eta =1$ and $r = \alpha L$ by iterating the argument. Moreover, $E[ d_1( \varepsilon_i\circ S, \varepsilon_i \circ T )] \leq L d_1( S, T )$ holds if the  $\{ \varepsilon_i \}$ satisfy  $E[ | \varepsilon_i(x) - \varepsilon_i(y) | ] \leq L |x-y|  $.

%In literature, the idea of iterated random function system is proposed to draw fractals and derive steady state distributions of Markov chains. 

\begin{figure}[t]
	\centering
	\begin{tikzpicture}
	%\draw[help lines,step=1] (0,0) grid (14,6);
	%\foreach \x in {0,1,2,3,4,5,6,7, 8, 9, 10, 11, 12, 13, 14}
	%\node[anchor=north] at (\x,0) {\x};
	%\foreach \y in {0,1,2,3,4, 5, 6}
	%\node[anchor=east] at (0,\y) {\y};
	
	\draw[very thick,->] (0,0) -- (6,0) ;
	\draw[very thick, ->] (0,0) -- (0,6) ;
	\draw[very thick,rounded corners=3mm] (7,3)--(7.5,5)--(8.4,5.3)--(8.6,5.8)--(11,6)--(12.5,5.7)--(14,4.6)--(13.6,3.3)--(13.675,0.9)--(11,0)--(9,0.2)--(8.3,0.3)--(7.3,1)--cycle;
	\filldraw (1,1) node[anchor=south east] {$c$};
    \filldraw [gray] (1,1) circle (2pt);
    \filldraw [gray] (3,5) circle (2pt);
	\filldraw (8.5,1) node[anchor=south east] {$\mu_{\mathcal{F}}$};
	\draw (3,6) node[anchor = south] {$\mathbb{R}^2$};
	\draw (10.5,6) node[anchor = south] {$\mathcal{W}$};
	\filldraw [gray] (8.5,1) circle (2pt);
	\draw[arrows={->[scale=1.5]}, dashed] (1,1) -- (3,5) node[anchor=south]  {$X_{i-1}$};
	\filldraw [gray] (10.5,5) circle (2pt);
	\draw (10.5,5) node[anchor=south] {$\mu_{i-1}$};
	\filldraw (2.5,3) node[anchor=east] {$\beta(X_{i-1}-c)$};
	\draw[dashed] (8.5,1) to[out=50,in=-100,] (9.6,4);
	\draw[dashed, arrows={->[scale=1.5]}] (9.6,4) to[out=80,in=-150,] (10.5,5);
	\draw[thick, arrows={->[scale=1.5]}, dashed, red] (8.5,1) to[out=50,in=-100,] (9.6,4);
	\filldraw (9.3,3) node[anchor=east] {$\alpha \odot T_{i-1}$};
	\draw[thick, arrows={->[scale=1.5]}, dashed, red] (1,1) -- (2.5,4) ;
	
	\draw[thick, arrows={->[scale=1.5]}, dashed, blue] (2.5,4) --(4, 3.5) ;
	\draw (3.4,3.6) node[anchor=south] {$\epsilon_i$};
	\draw (4,3.5) node[anchor=west] {$X_{i}$};
	\filldraw [gray] (4,3.5) circle (2pt);
	\draw[thick, arrows={->[scale=1.5]}, dashed, green] (1,1) -- (4, 3.5) ;
	\draw (2.5,2) node[anchor=west] {$\beta(X_{i-1}-c) +\epsilon_i$};
	
	\draw[thick, arrows={->[scale=1.5]}, dashed, blue] (9.6,4) to[out = -50, in =170] (11.5, 3);
	\draw (10.6,3.2) node[anchor=south] {$\varepsilon_i$};
	\draw (11.5,3) node[anchor=west] {$\mu_{i}$};
	\filldraw [gray] (11.5,3) circle (2pt);
	\draw[thick, dashed, green,arrows={->[scale=1.5]} ] (8.5,1) to[out=20, in=-120] (11.5, 3) ;
	\draw (10.1,1.5) node[anchor=west] {$\alpha \odot T_{i-1} \oplus \varepsilon_i$};
	\end{tikzpicture}
	\caption{Illustration of the VAR(1) model $X_i -c = \beta (X_{i-1}-c) + \epsilon_i  $ in $\mathbb{R}^2$ (left) and the  $\text{ATM}(1)$ model $ T_{i} = \alpha \odot T_{i-1} \oplus \varepsilon_i, T_i = F_i^{-1} \circ F_{\mathcal{F}} $ in $\mathcal{W}$ (right). The colored dashed lines are geodesics and correspond to the respective optimal transport maps.}
	\label{fig:atmillustration}
\end{figure} \vs

%\noindent { \sf 3.2 \quad Estimation} %of Model Parameter $\alpha$}
\subsection{Estimation}

 \no Assuming that the true model parameter is $-1 <\alpha <1$,  a  consistent estimator for $\alpha$ is obtained as follows. In distributional data analysis and distributional time series the distributions that serve as data atoms are usually not known but one rather has available i.i.d. samples of real-valued data that have been generated by these distributions and this needs to be taken into account in the analysis. We denote the available estimates of  transport maps $T_i$ by $ \widehat{T}_i, \, i=1,\dots,n.$ %with details about these estimates to be provided below.  %is not directly given, we assume that its estimate $ \widehat{T}_i$ exists. Details on how to construct $\widehat{T}_i$ are provided later. 
 Depending on whether $\alpha$ is positive or negative, $\widehat{T}_i$  or $\widehat{T}_{i}^{-1}$ is used accordingly in the proposed method. If $ \{T_1, \dots, T_n\} $ satisfies model \eqref{eq:arr1}, then it holds that
\begin{align*}
 \alpha = \left\lbrace \begin{array}{cc}
  \frac{\int_{\mathcal{S}} E[ (T_{i+1}(x)-x)(T_i(x)-x) ] \,dx}{ \int_{\mathcal{S}} E[(T_{i}(x)-x)^2] \,dx}, & \text{ if } \alpha \geq 0, \\
  \frac{\int_{\mathcal{S}} E[ (T_{i+1}(x)-x)(x-T_i^{-1}(x)) ] \,dx}{ \int_{\mathcal{S}} E[(x-T_{i}^{-1}(x))^2] \,dx}, & \text{ if } \alpha < 0. 
\end{array} \right.
\end{align*} 
This motivates the following least squares type estimators of $\alpha$,  
\begin{align*}
\widehat{\alpha} =  \left\lbrace
\begin{array}{cc}
\widehat{\alpha}_{+} & \quad\text{if} \quad  l_{+}( \widehat{\alpha}_{+} ) \leq  l_{-}( \widehat{\alpha}_{-} ),   \\
\widehat{\alpha}_{-} & \quad\text{if} \quad l_{+}( \widehat{\alpha}_{+} ) > l_{-}( \widehat{\alpha}_{-} ).
\end{array} \right.
\end{align*} 
where $ \widehat{\alpha}_{+}  = \argmin_{\alpha} l_{+}(\alpha) $, $ \widehat{\alpha}_{-}  = \argmin_{\alpha} l_{-}(\alpha) $ and 
\begin{align*}
  l_{+}(\alpha)  & = \sum_{i=2}^n \int_{\mathcal{S}} \left( \widehat{T}_i(x) - x - \alpha ( \widehat{T}_{i-1}(x) - x  ) \right)^2 \, d x, \\
  l_{-}(\alpha) & = \sum_{i=2}^n \int_{\mathcal{S}} \left( \widehat{T}_i(x) - x - \alpha ( x- \widehat{T}_{i-1}^{-1}(x)   ) \right)^2 \, d x.
\end{align*} 

\begin{theorem}\label{thm2}
    Suppose $T_0 \sim^{i.i.d}\widetilde{T}_{0}$ and $\{ T_i \}_{i=1}^n$ are strictly increasing, continuous and generated from equation \eqref{eq:arr1} with $-1 < \alpha < 1$ and $T_0$ as the initial transport. Under the assumptions of Theorem 1 with $q=1$, if 
$
\int_{\mathcal{S}}E[ ( T_{1}(x)-x)^2 ] \,dx > 0,
$
\begin{align*}
|\widehat{\alpha} - \alpha| = O_p\left(\tau + \frac{1}{ \sqrt{n} } \right),
\end{align*}
where $\tau  = \sup_i E[ d_{1}(\widehat{T}_i, T_i) ]$. 
\end{theorem}

Intuitively, the condition $\int_{\mathcal{S}}E[ ( T_{1}(x)-x)^2 ] \,dx > 0 $ ensures that the sequence of transport maps deviates from a sequence of identity maps. This is required to arrive at a consistent estimator, since if $T_i = id$ almost surely, equation \eqref{eq:arr1} would hold for any $\alpha \in \mathbb{R}$ and it is then not possible to estimate $\alpha$ consistently. More specifically,  if $\int_{\mathcal{S}}E[ ( T_{1}(x)-x)^2 ] \,dx > 0$, then the following  application of the Cauchy-Schwarz  inequality excludes the case of equality and therefore gives rise  to the  strict inequality
\begin{multline*}
c' = \left(\int_{\mathcal{S}} E[(T_{1}(x)-x)^2] \,dx\right) \left(\int_{\mathcal{S}} E[(x - T_{1}^{-1}(x))^2] \,dx\right) \\  -  \left(\int_{\mathcal{S}} E[(T_{1}(x)-x)(x - T_{1}^{-1}(x))] \,dx\right)^2 >0, 
\end{multline*}
where  $1/c'$ is an implicit  constant in the $O_p$ for the rate of convergence  result $O_p (\tau + 1/\sqrt{n})$.

For practical applications it needs to be taken into account that  the underlying distributions are almost always unknown. Accordingly, a realistic starting point is that one has  available samples  of independent realizations $\{ X_{i,l} \}_{l=1}^{N_i}$ that are obtained for each of  the distributions $\mu_i$. There are then two independent random mechanisms that generate the data. The first of these generates  random distributions $\{ \mu_i \}_{i=1}^n$; the second generates  randomly drawn samples $\{ X_{i,l} \}$ from each $\mu_i.$ 
Based on the $\{ X_{i,l} \}_{l=1}^{N_i}$, cdfs  $\{ F_i\} $ or quantile functions $\{ Q_i \}$ can be estimated with available methodology  \citep{falk1983relative, leblanc2012estimating}. Denoting the estimated cdfs by $\widehat{F}_i$, the corresponding quantile function estimates are  $ \widehat{Q}_i(a) = \inf\{ x \in \mathcal{S} \;| \; \widehat{F}_i \geq a \}, a \in [0,1] $. Alternatively, one can directly estimate quantile functions  \citep{cheng1997unified} %and derive empirical cdfs as $ \widehat{F}_i(x) = \inf\{ a \in [0,1] \;| \; \widehat{Q}_i(a) \leq x \}$. Alternatively, one cam estimate the density functions
or start with density estimates  and convert these to cdfs  using numerical integration, obtaining rates such as  % The distribution estimator of \cite{panaretos2016amplitude} satisfies 
$\sup_{\mu \in \mathcal{W}} E[d_{\mathcal{W}}^2(\widehat{\mu}, \mu)] = O(1/\sqrt{N}),  $ where $N  = \min \{ N_i: i =1,2, \dots, n \}$ \citep{panaretos2016amplitude} under suitable assumptions or alternatively $\sup_{\mu \in \mathcal{W}^{ac}_{R}} E[d_{\mathcal{W}}^2(\widehat{\mu}, \mu)] = O(N^{-2/3})$ on the set of absolutely continuous distributions  \citep{petersen2016functional}.
With estimates for quantile functions and cdf in hand, 
one then obtains optimal transport map estimates  $\widehat{T}_i = \widehat{Q}_i \circ \widehat{F}_{\mathcal{F}} \text{ or } \widehat{T}_i = \widehat{Q}_{i+1} \circ \widehat{F}_{i},$ where $ \widehat{F}_{\mathcal{F}} = \widehat{Q}^{-1} _{\mathcal{F}} $ and $ \widehat{Q}_{\mathcal{F}} = \sum_{i=1}^{n} \widehat{Q}_i/n $, implying 
\begin{align*}
\tau \lesssim \sup_i E[d_1(\widehat{T}_i, T_i)] \lesssim \max\{  \sup_i (E[d_{\mathcal{W}}^2(\widehat{\mu}_i, \mu_i)])^{1/2}, (E[d_{\mathcal{W}}^2(\widehat{\mu}_{\mathcal{F}}, \mu_{\mathcal{F}})])^{1/2} \} % = O(1/\sqrt{N}),
\end{align*}
for the rate $\tau$ in Theorem 2, 
where  $a \lesssim b$ means that there exists a constant $C>0$ such that $a \leq Cb$. Depending on assumptions and estimation procedures as mentioned above, one then obtains convergence rates ranging from $\tau \sim N^{-1/4}$ to  
$\tau \sim N^{-1/3}$. %\vspace{2cm}

\bco

when quantile functions $\{Q_i\}$ as well their estimates $\{ \widehat{Q}_i \}$ are Lipschitz continuous with the same Lipschitz constant $L$  \citep{panaretos2016amplitude}, and 
\begin{align*}
\max\{  E[ d_1(Q_i \circ F, Q_i \circ F' ) ] , E[ d_1(\widehat{Q}_i \circ F, \widehat{Q}_i \circ F' ) ] \} \leq L d_1 ( F, F' )
\end{align*}
restricting to the set of absolutely continuous distribution, we have $ \tau =  O(N^{-2/3})$ using method of \cite{petersen2016functional}. Furthermore, the uniform Lipschitz continuous condition on $ \{ Q_i, \widehat{Q}_i \} $ can be relaxed to 
for any fixed cumulative distribution functions $F, F'$ supported on $\mathcal{S}$.

\fi

%\bc {\bf \sf 4.\quad AUTOREGRESSIVE TRANSPORT MODELS OF ORDER $p$}\sm \ec \rs\rs 

\section{AUTOREGRESSIVE TRANSPORT MODELS OF ORDER $p$}

%\noindent { \sf 4.1 \quad Stationary Solution}
\subsection{Stationary Solution}

\no Autoregressive transport models of order $p$ (ATM($p$)) are defined as
\begin{align}
\label{eq:orderp}
T_i = \alpha_p \odot T_{i-p} \oplus  \alpha_{p-1} \odot T_{i-p+1} \oplus \dots \oplus  \alpha_1 \odot T_{i-1} \oplus \varepsilon_{i},
\end{align}
where $\alpha_1, \dots, \alpha_p \in \mathbb{R}$ are model parameters  and $\varepsilon_i$ are i.i.d.  random distortion maps with $E(\varepsilon_i) = id$. To show the existence of stationary solutions, we construct a chain of functions and again  apply the geometric-moment contraction condition \citep{wu2004limit}. Let $ \mathcal{T}^{p} = \mathcal{T} \times \dots \times \mathcal{T}$ be the product space, $ \mathbf{S} = (S_1, S_2, \dots, S_{p}), \mathbf{R} = (R_1, R_2, \dots, R_{p} ) \in \mathcal{T}^p$ and define the random functions $ \Upsilon_{\varepsilon}, \widetilde{\Upsilon}_{i,m} : \mathcal{T}^{p} \rightarrow \mathcal{T}^{p} $ as
\begin{align*}
    \Upsilon_{\varepsilon} (\mathbf{S})  &= ( S_{2}, \dots, S_{p},  \alpha_p \odot S_{1} \oplus \dots \oplus  \alpha_1 \odot S_{p} \oplus \varepsilon), \\
     \widetilde{\Upsilon}_{i,m}(\mathbf{S})   &= \Upsilon_{\varepsilon_i} \circ \Upsilon_{\varepsilon_{i-1}} \circ \dots \circ \Upsilon_{\varepsilon_{i-m+1}}(\mathbf{S}),
\end{align*}
where $\varepsilon, \varepsilon_i$  are random distortion transports. We employ the product $L^q$-metric on $ \mathcal{T}^{p}$ given by 
%\begin{align*}
$d_{q} (\mathbf{S}, \mathbf{R}) = \left\lbrace \sum_{i=1}^{p}  d_{q}^{2}(S_i, R_i)   \right\rbrace^{1/2},$ where $q \ge 1$ is a fixed constant in the following. 

%\end{align*}
\begin{theorem}\label{thm3}
    Suppose there exists $ \eta >0$, $ \mathbf{S}_0 \in \mathcal{T}^p$, $C >0$ and $r \in (0,1)$ such that 
\begin{align}
\label{eq:AR}
    E \left[ d_{q}^{\eta} \left(\widetilde{\Upsilon}_{i,m}( \mathbf{S}_0), \widetilde{\Upsilon}_{i,m}( \mathbf{R}) \right) \right] \leq C r^m  d_{q}^{\eta}  ( \mathbf{S}_0, \mathbf{R} )
\end{align}
holds for all $ \mathbf{R} \in \mathcal{T}^{p} $ and $m \in \mathbb{N}$.  Then, for all $\mathbf{S} \in \mathcal{T}^p$, 
$$
(\widetilde{T}_{i-p+1},\widetilde{T}_{i-p+2}, \dots, \widetilde{T}_i ) := \lim\limits_{m \rightarrow \infty} \widetilde{\Upsilon}_{i, m}(\mathbf{S}) \in \mathcal{T}^p
$$ 
exists almost surely and does not depend on $\mathbf{S}$. In addition, $(\widetilde{T}_{i-p+1},\widetilde{T}_{i-p+2}, \dots, \widetilde{T}_i )$ is a stationary solution of the following system of stochastic equations 
\begin{align*}
T_i = \alpha_p \odot T_{i-p} \oplus  \alpha_{p-1} \odot T_{i-p+1} \oplus \dots \oplus  \alpha_1 \odot T_{i-1} \oplus \varepsilon_{i}, \; i \in \mathbb{Z}
\end{align*}
and is unique almost surely.
\end{theorem}

For  motivation of   $\Upsilon_{\varepsilon}$ and condition \eqref{eq:AR}, consider the classical AR($p$) model in $\mathbb{R}$, i.e. 
$Y_{i} = \sum_{j=1}^{p} \beta_j Y_{i-j} + \epsilon_i \in \mathbb{R}$, which can be represented as a vector autoregressive model of order 1 (VAR(1)) in the form $ \mathbf{Y}_i = B \mathbf{Y}_{i-1} + \bm{\epsilon}_i$, where $\mathbf{Y}_i = (Y_i, \dots, Y_{i-p+1})^T$, $\bm{\epsilon}_i = (\epsilon_i, 0, \dots, 0)^T \in \mathbb{R}^p$ and 
\begin{align*}
B = \left(
\begin{array}{ccccc}
\beta_1 & \beta_2& \dots & \beta_{p-1} & \beta_p  \\
  1   & 0  & \dots  & 0 &0 \\
0     & 1  & \dots  & 0 &0 \\ 
0  &0  & \dots  & 1  & 0 
\end{array}
\right).
\end{align*}
With  (nonrandom) starting points  $\mathbf{Y}_0$ and $\mathbf{Y}_0'$, running the VAR(1) model recursively $m$ times, one obtains   $ \mathbf{Y}_{m} = B^m\mathbf{Y}_0 + \sum_{j=1}^m B^{m-j}\bm{\epsilon}_j $ and $ \mathbf{Y}_{m}' = B^m \mathbf{Y}_0' + \sum_{j=1}^m B^{m-j}\bm{\epsilon}_j$. With $\| \cdot \|_2$ denoting the Euclidean norm, condition \eqref{eq:AR} for this model becomes $ E\left[ \| \mathbf{Y}_m - \mathbf{Y}_m' \|_2 \right] \lesssim r^m \| \mathbf{Y}_0 - \mathbf{Y}_0' \|_2$ for some $0<r<1$. With
 a slight abuse of notation,  denoting the spectral norm of $B$  as $\| B \|_2$, 
\begin{align*}
E\left[ \| \mathbf{Y}_m - \mathbf{Y}_m' \|_2 \right] = \| B^m(\mathbf{Y}_0 - \mathbf{Y}_0') \|_2^{\eta}  \leq \| B^m \|_{2}  \| \mathbf{Y}_0 - \mathbf{Y}_0' \|_2. 
\end{align*}

Now if the absolute values of the eigenvalues of $B$ are bounded above by a constant $0<r <1$, i.e. they are inside the  unit circle, then $ \| B^m \|_2 \lesssim r^m $, and this is equivalent to the fact that the  roots of $\phi(z) = 1 - \sum_{j=1}^p \beta_j z^j$ all lie outside the unit circle. The latter is a standard assumption for the existence of stationary solutions of AR($p$) processes in Euclidean space. In  linear spaces   the terms containing the innovation errors  in $ \mathbf{Y}_{m}$ and $ \mathbf{Y}_{m}' $ cancel, which for this case simplifies the verification of  Condition \eqref{eq:AR}. 

To select the order of the ATM, we propose an approach based on rolling-window validation and refer to \cite{zivot2007modeling} for more details on rolling-window analysis for time series. To train the ATM($p$) on a given sequence $\{\mu_{t}, \mu_{t+1}, \dots, \mu_{t+m-1}\}$ of length $m$ with starting time $t$, we assume that there exists a pre-sample of length $k$, i.e., $\{\mu_{t-k} , \dots,  \mu_{t-1} \}$. For each fixed $p$ in a candidate set, the sample $\{ \mu_{t-k}, \mu_{t-k+1}, \dots, \mu_{t-k+m-1} \}$ is used as  training set to predict the distribution at time $t-k+m$. Denoting this predicted distribution as  $\widehat{\mu}_{t-k+m}$, the prediction accuracy can be measured by Wasserstein distance $d_{\mathcal{W}}(\mu_{t-k+m}, \widehat{\mu}_{t-k+m})$. Then  roll the window one step forward and use $\{ \mu_{t-k+1}, \mu_{t-k+1}, \dots, \mu_{t-k+m} \}$ as training set to make a prediction at time $t-k+m+1$ and compute the error $d_{\mathcal{W}}(\mu_{t-k+m+1}, \widehat{\mu}_{t-k+m+1})$. Rolling the training window forward repeatedly until the last window covering  time $t-1$ to $t+m-2$ is reached and computing the error $d_{\mathcal{W}}(\mu_{t+m-1}, \widehat{\mu}_{t+m-1})$ then leads to the selection of the  autoregressive order $p$   as the minimizer of  $\sum_{i=t+m-k}^{t+m-1} d_{\mathcal{W}}(\mu_i, \widehat{\mu}_i)$ over a candidate set of orders. 

%\noindent { \sf 4.2 \quad Estimation of Model Parameters} %$(\alpha_1, \alpha_2, \dots, \alpha_p)$}
\subsection{Estimation of Model Parameters}

\no Hereafter, we denote the true model parameters as $(\alpha_1^{*}, \dots, \alpha_p^{*})$ to avoid confusion. Obvious estimates of the  ATM($p$) parameters $\alpha_1^{*}, \alpha_2^{*}, \dots, \alpha_p^{*} $ are  obtained as minimizers of %square $L^2$-distances
\begin{align*}
L_n(\alpha_1, \alpha_2, \dots, \alpha_p) = \frac{1}{n-p} \sum_{i = p+1}^n \int_{\mathcal{S}} \left( T_{i}(x) - \alpha_p \odot T_{i-p} \oplus  \dots \oplus  \alpha_1 \odot T_{i-1}(x) \right)^2 dx.
\end{align*}
%If the transport maps $\{T_i\}$ need to be estimated first, $T_i$ can be replaced by its estimate $\widehat{T}_i$ conveniently in the subsequent computation.  
When $p>1$, the minimization of $L_n(\alpha_1, \dots, \alpha_p)$ is challenging,  as the functional $L_n$ in general is  not convex. We propose a back propagation-type algorithm to address this minimization problem.  
The partial derivatives of $ \alpha \odot T_i (x) $ with respect to $x$ are 
\begin{align*}
\frac{\partial}{\partial x} \alpha \odot T_{i} (x)  = \left\lbrace  \begin{array}{cc}
  (1 + a (g_b(x, T_i) - 1) ) \times \left( \prod\limits_{l=0}^{b-1} g_l(x, T_i) \right) , & \text{ if }\alpha > 0, \\
1, &  \text{ if } \alpha =0, \\
 (1 + a ( 1 - g_b(x, T_i^{-1}) ) ) \times \left( \prod\limits_{l=0}^{b-1} g_l(x, T_i^{-1}) \right) , & \text{ if } \alpha <0,
\end{array} \right.
\end{align*}
where $b = \lfloor |\alpha| \rfloor$, $a = |\alpha| - b$, $T', (T^{-1})'$ are the derivatives of $T, T^{-1}$ respectively, $ \prod_{l=0}^{b-1} g_l(x, T)$ is defined to be 1 if $ b-1 <0 $ and 
\begin{align*}
g_l(x, T) =  \left\lbrace  \begin{array}{cl}
T'( x ), & \text{ if } l = 0,\\
T'( \underbrace{T\circ T \circ \dots \circ T}_{l \text{ compositions of } T} (x) ) & \text{ if } l = 1,2, \dots \\
\end{array} \right.
\end{align*}
The partial derivative with respect to $\alpha$ when $a>0$ is 
\begin{align*}
\frac{\partial}{\partial \alpha} \alpha \odot T_{i} (x) = \left\lbrace  \begin{array}{cc}
T_i(h(x,T_i)) - h(x, T_i), & \text{ if }\alpha > 0 \\
h(x, T_i^{-1}) - T_i^{-1}(h(x, T_i^{-1})), & \text{ if } \alpha <0,
\end{array} \right. 
\end{align*}
where  
\begin{align*}
h(x, T) = \left\lbrace \begin{array}{cc}
x     & \text{ if } b=0, \\
 \underbrace{T \circ T \circ \dots \circ T}_{b \text{ compositions of } T} (x)    & \text{ if } b >0.
\end{array} \right.
\end{align*}

Since  $ \alpha \odot T_{i} (x) $ is not differentiable w.r.t $\alpha$ if $\alpha \in \mathbb{Z}$, we use its subdifferential (subgradient). When $\alpha = 0$, we set $ \partial \alpha \odot T_{i} (x)/\partial \alpha $ at $\alpha=0$ to be any value in the closed interval between $ T_i(x) - x $ and $ x - T_i^{-1}(x) $. In our simulations, $ \partial \alpha \odot T_{i} (x)/\partial \alpha $ at $\alpha=0$ is selected uniformly from $T_i(x) - x $ and $ x - T_i^{-1}(x)$. When $0 \neq \alpha \in \mathbb{Z} $, $ \partial \alpha \odot T_{i} (x)/\partial \alpha $ is set to be the partial derivative of $ \alpha \odot T_{i} (x) $ at a point $ \alpha' $ such that $\alpha'$ has the same sign as $\alpha$ and $ |\alpha| < |\alpha'| < (|\alpha|+1) $. For more details on the  back-propagation type algorithm for ATM of order $p$ see the  display for Algorithm \ref{alg2}. We employ  gradient clipping,  a common technique used in deep neural networks to prevent exploding gradients.

Next, we  establish consistency for the minimizer of  $L_n(\alpha_1, \dots, \alpha_p)$,  i.e.   
\begin{align*}
    \widetilde{\bm{\alpha}} := (\widetilde{\alpha}_1, \widetilde{\alpha}_2,  \dots, \widetilde{\alpha}_p)^T \in \argmin_{-c \leq \alpha_1, \dots, \alpha_p \leq c} L_n(\alpha_1, \alpha_2, \dots, \alpha_p),
\end{align*}
where $c$ is the same constant as in Theorem \ref{thm4} below, which demonstrates that $ (\widetilde{\alpha}_1, \widetilde{\alpha}_2,  \dots, \widetilde{\alpha}_p) $ converges to the true model parameters in probability with respect to the discrepancy
$$
\Delta (\widetilde{\bm{\alpha}}, \bm{\alpha}^{*}) := \int_{\mathcal{S}} E\left[ (\widetilde{\alpha}_p \odot T_{1} \oplus  \dots \oplus  \widetilde{\alpha}_1 \odot T_{p}(x) - \alpha_p^{*} \odot T_{1} \oplus  \dots \oplus \alpha_1^{*} \odot T_{p}(x)   )^2 \right] dx, 
$$
where $\bm{\alpha}^{*} = (\alpha_1^{*}, \dots, \alpha_p^{*})^T$ are the true model parameters. The key step, where the constant $c$ is used, is to show that $\sup_{-c \leq \alpha_1, \dots, \alpha_p \leq c} | L_n(\alpha_1, \dots, \alpha_p) - E[L_n(\alpha_1, \dots, \alpha_p)]|=o_p(1)$ based on Corollary 3.1 of \cite{whitney1991}. In practice, we simply set $c$ to be a large enough number.
\begin{theorem}\label{thm4}
Under the assumptions of Theorem 3 with $q=1$, 
%suppose the minimizer of $L(\alpha_1, \dots, \alpha_p)  = E \left[ \int_{\mathcal{S}} \left( T_{i}(x) - \alpha_p \odot T_{i-p} \oplus  \dots \oplus  \alpha_1 \odot T_{i-1}(x) \right)^2 dx \right]$ exists and is unique. Let $(\alpha_1^{*}, \dots, \alpha_p^{*}) = \argmin_{\alpha_1, \dots, \alpha_p} L(\alpha_1, \alpha_2, \dots, \alpha_p)$, 
if $T_0 \sim^{i.i.d}\widetilde{T}_{0}$ and $\{ T_i \}_{i=1}^n$ are strictly increasing, differentiable,  bi-Lipschitz continuous with Lipschitz constant $K$ and generated from equation \eqref{eq:orderp} with $T_0$ as the initial transport and $(\alpha_1, \dots, \alpha_p) = (\alpha_1^{*}, \dots, \alpha_p^{*})$ where $ -c \leq \alpha_1^{*}, \dots, \alpha_p^{*} \leq c $ for some constant $c >0$, then
\begin{align*}
    \Delta (\widetilde{\bm{\alpha}}, \bm{\alpha}^{*}) \overset{p}{\rightarrow} 0\text{ as } n \rightarrow \infty.
\end{align*}
%\begin{align*}
%(\widetilde{\alpha}_1, \widetilde{\alpha}_2,  \dots, \widetilde{\alpha}_p)^T \overset{p}{\rightarrow}  (\alpha_1^{*}, \alpha_2^{*}, \dots, \alpha_p^{*})^T \text{ as } n \rightarrow \infty.
%\end{align*}
\end{theorem}

\begin{algorithm}[htbp]
	\caption{Back Propagation  Algorithm for Fitting ATM$(p), p>1.$}
	\label{alg2}
	\SetKwInOut{In}{input}\SetKwInOut{Out}{output}
	Select a grid $ s_1 < x_1 < x_2 < \dots < x_m < s_2 $. \\
	Select step size $\eta$. \\
	Initialize  $\alpha_k^0 = 0$ for $k=2,3, \dots, p$ and $$\alpha_1^{0} = \argmin_{\alpha} \frac{1}{n-p} \sum_{i=p+1}^n \sum_{j=1}^m \left( T_i(x_j) -  \alpha \odot T_{i-1} (x_j) \right)^2.$$

	\For{$t =1,2, \dots  $}{ 
	
	\vspace{0.5cm}
	
	    \CommentSty{ \hspace{5cm} Forward Pass} \\
	    $\text{For all } i = p+1, \dots, n, j = 1, \dots, m $, compute $R_{1,ji}^t = \alpha_p^{t-1} \odot T_{i-p}(x_j).$
	    
	    \For{$k=2, 3, \dots, p$}{
	    $\text{For all } j, i $, compute 
	    $$  
	    R_{k,ji}^t =   \alpha_{p+1-k}^{t-1} \odot T_{i-(p+1-k)} (R_{k-1,ji}^t) . 
	    $$
	    }
	   $\text{For all } j, i $, compute $ L^t_{ji} =2 \left( T_i(x_j) -  R_{p,ji}^t \right). $ 
	   
	   \vspace{0.5cm}
	   
    	\CommentSty{ \hspace{5cm} Backward Pass} 
    	
    	For all $j,i$, set $ D_{0, ji}^t = 1 $ . 
    	
    	\For{$k=1,2,\dots, p-1$}{
    	For all $j,i$, compute
    	\begin{align*}
    	D_{k, ji}^t = ( D_{k-1, ji}^t ) \times \left(  \left.  \frac{\partial}{\partial x} \alpha_k^{t-1} \odot T_{i-k}(x) \right|_{x=R_{p-k,ji}^t}  \right) \text{ for all }j,i, 
    	\end{align*}
    	Update $\alpha_k$ as
    	\begin{align*}
    	 \alpha_k^t = \alpha_k^{t-1} +  \frac{\eta}{n-p}\sum_{i=p+1}^n \sum_{j=1}^m \left( L_{ji}^t  D_{k-1, ji}^t  \left. \frac{\partial}{\partial \alpha} \alpha \odot T_{i-k}( R_{p-k,ji}^t ) \right|_{\alpha = \alpha_k^{t-1}} \right).
    	\end{align*}
    	}
    	Compute
    	$
    	\alpha_p^t = \alpha_p^{t-1} + \frac{\eta}{n-p} \sum_{i=p+1}^n \sum_{j=1}^m \left( L_{ji}^t ( D_{p-1, ji}^t ) \left.  \frac{\partial}{\partial \alpha} \alpha \odot T_{i-p}( x_j ) \right|_{\alpha = \alpha_p^{t-1}} \right).
     $
    	
	\If{ stopping conditions hold }{\Return{ $(\alpha^{t}_1, \alpha_2^t, \dots, \alpha^t_p)$ }}
	}   
\end{algorithm}

%\bc {\bf \sf 5.\quad CONCURRENT AUTOREGRESSIVE TRANSPORT MODEL}\sm \ec \rs
\section{CONCURRENT AUTOREGRESSIVE TRANSPORT MODEL}

\no A promising extension of ATMs of order 1 is to consider model coefficients that vary with $x \in \mathcal{S}$. For a function  $\beta : \mathcal{S} \rightarrow [-1,1]$, define   the operation   
\begin{align*}
		\beta \circledcirc T (x) := \left\lbrace  \begin{array}{cc}
		x + \beta (x) (T(x)-x), & 0 < \beta (x) \leq 1 \\
		x, & \beta (x) = 0 \\
		x + \beta (x) (x-T^{-1}(x)), & -1 \leq \beta (x) <0
		\end{array} \right. .
\end{align*}
This leads to the following concurrent autoregressive transport model (CAT), 
\begin{align}
\label{eq:vary}
T_i =  \beta \circledcirc T_{i-1} \oplus \varepsilon_{i},
\end{align}
with i.i.d. random distortion transports  $\varepsilon_i$  satisfying $E(\varepsilon_i)= id$. 

To ensure monotonicity that is required for the transports  to be well defined, given the true function $\beta$,  we consider a subset of transports $\tilde{\mathcal{T}}\subset \mathcal{T}$ 
\bco=\tilde{\mathcal{T}}(\beta)$  of $\mathcal{T}$ \fi
such that 
$ \beta \circledcirc \tilde{\mathcal{T}} := \{ \beta \circledcirc T : T \in \tilde{\mathcal{T}}\} \subseteq \tilde{\mathcal{T}}$ and assume that $P(\varepsilon_i \circ \tilde{\mathcal{T}}\subseteq \tilde{\mathcal{T}})=1$ where $\varepsilon_i \circ \tilde{\mathcal{T}}:= \{ \varepsilon_i \circ T : T \in \tilde{\mathcal{T}}\}$; this obviously  holds if the function $\beta$ does not vary, i.e. is constant, whence  $\tilde{\mathcal{T}}= \mathcal{T}$ and $P(\varepsilon_i \circ \tilde{\mathcal{T}}\subseteq \tilde{\mathcal{T}})=1$. Whenever $ \beta \circledcirc \tilde{\mathcal{T}}  \subseteq \tilde{\mathcal{T}}$ and $P(\varepsilon_i \circ \tilde{\mathcal{T}}\subseteq \tilde{\mathcal{T}})=1$, the  random functions 
\begin{align*}
    \varphi_{\varepsilon}( S ) =  \beta \circledcirc S \oplus \varepsilon, \qquad
    \widetilde{\varphi}_{i,m}(S)  = \varphi_{\varepsilon_i} \circ \varphi_{\varepsilon_{i-1}} \circ \dots \circ \varphi_{\varepsilon_{i-m+1}}(S), \qquad \varphi_{\varepsilon}, \widetilde{\varphi}_{i,m} : \tilde{\mathcal{T}}\rightarrow \tilde{\mathcal{T}}\end{align*}
    are well-defined for any $S \in \tilde{\mathcal{T}}$.
    
An example for this concurrent autoregressive transport model (CAT)  is as follows. Let $s_1 = t_1 < \dots < t_k = s_2$ be a grid over $\mathcal{S}$ and $\beta:\mathcal{S} \rightarrow [0,1]$ be such that  $\beta$ is positive and is either  increasing or  decreasing on each grid interval $ [t_i, t_{i+1}] $. Here $\tilde{\mathcal{T}}$  is selected as a  set of transports such that for any $T \in \tilde{\mathcal{T}}$, $T(t_i) = t_i$, $T(x) \geq x$ if $ \beta (x) $ is increasing and  otherwise $T(x) < x$.  The   properties required for the CAT model are satisfied as  $\tilde{\mathcal{T}}$ is complete and $\{\varepsilon_i\}$ can be defined as random distortion maps taking values in $\tilde{\mathcal{T}}$.
To state our next result, we  equip $\tilde{\mathcal{T}}$ with the sup-metric $d_{\infty}(f,g) = \sup_{x \in \mathcal{S}} |f(x) - g(x)|$.

\begin{theorem}\label{thm5}
    Suppose that $\tilde{\mathcal{T}}$ is a complete metric space,  $P(\varepsilon_i \circ \tilde{\mathcal{T}}\subseteq \tilde{\mathcal{T}})=1$ and there exists $ \eta >0$, $ S_0 \in \tilde{\mathcal{T}}$, $C >0$ and $r \in (0,1)$ such that 
\begin{align}
\label{eq:momentinfty}
E \left[ d_{\infty}^{\eta} \left( \widetilde{ \varphi}_{i, m}(S_0), \widetilde{\varphi}_{i, m}(T) \right) \right] \leq C r^{m} d_{\infty}^{\eta} ( S_0, T)
\end{align}
holds for all $ T\in \tilde{\mathcal{T}}$ and $m \in \mathbb{N}$. Then, for all $S \in \tilde{\mathcal{T}}$, $\widetilde{T}_i := \lim\limits_{m \rightarrow \infty} \widetilde{ \varphi}_{i, m}(S) \in \tilde{\mathcal{T}}$ exists almost surely and does not depend on $S$. In addition, $\widetilde{T}_i$ is a stationary solution of the system of stochastic equations 
\begin{align}
\label{eq:model}
T_i =  \beta \circledcirc T_{i-1} \oplus \varepsilon_{i}, \; i \in \mathbb{Z}
\end{align}
and is unique almost surely.
\end{theorem}

 %$\beta$ as well as $\tilde{\mathcal{T}}$ that satisfies $ \beta \circledcirc \tilde{\mathcal{T}}\subseteq  \tilde{\mathcal{T}}$. 

%The sup metric in \eqref{eq:momentinfty} is used to ensure that 
%Here, we formulate sufficient conditions for \eqref{eq:momentinfty}. %Suppose that there exists contants $L, L'>0$ such that  $ %|\varepsilon_i(x) - \varepsilon_i(y)| \leq L |x-y| $, $ \sup_x | S^{-1}(x) - T^{-1}(x) | \leq \sup_{x} | S(x) - T(x)| $ for any $S, T \in \tilde{\mathcal{T}}$ and $ 0< \sup_x |\beta(x)| L L' < 1 $. Then, Condition \eqref{eq:momentinfty} is satisfied with $\eta=1$ and $r = \sup_x |\beta(x)| L L'$.

%\noindent { \sf 5.1 \quad Estimation of $\beta$}

 The estimation of the CAT model function $\beta$  proceeds  similarly to the  estimation  of the scalar coefficient in the  ATM(1). If $ \{T_1, \dots, T_n\} $ satisfy model \eqref{eq:vary}, then for all $ x\in \mathcal{S} $
\begin{align*}
\beta(x) = \left\lbrace
\begin{array}{cc}
   \frac{E[ (T_{i+1}(x)-x)(T_i(x)-x) ]}{ E[(T_{i}(x)-x)^2]},  & \text{ if } \beta(x) \geq 0, \\
  \frac{E[ (T_{i+1}(x)-x)(x-T_i^{-1}(x)) ]}{ E[(x-T_{i}^{-1}(x))^2]},   & \text{ if } \beta(x) < 0. 
\end{array} \right.
\end{align*}
This suggests estimates  $\widehat{\beta}(x)$ for $\beta(x)$ given by 
\begin{align*}
\widehat{\beta}(x) =  \left\lbrace
\begin{array}{cc}
\widehat{\beta}_{+}(x), & \text{ if }  l_{+}( \widehat{\beta}_{+}(x)|x ) \leq  l_{-}( \widehat{\beta}_{-}(x)|x ),   \\
\widehat{\beta}_{-}(x), & \text{ if }  l_{+}( \widehat{\beta}_{+}(x)|x ) > l_{-}( \widehat{\beta}_{-}(x)|x ),
\end{array} \right.
\end{align*} 
where $ \widehat{\beta}_{+}(x)  = \argmin_{\beta} l_{+}(\beta|x) $, $ \widehat{\beta}_{-}(x)  = \argmin_{\beta} l_{-}(\beta |x) $ and 
\begin{align*}
  l_{+}(\beta|x)   &= \sum_{i=2}^n  \left( \widehat{T}_i(x) - x - \beta ( \widehat{T}_{i-1}(x) - x  ) \right)^2 , \\
  l_{-}(\beta | x)  &= \sum_{i=2}^n  \left( \widehat{T}_i(x) - x - \beta ( x- \widehat{T}_{i-1}^{-1}(x)   ) \right)^2 .
\end{align*} 
Then we obtain  pointwise convergence of $ \widehat{\beta}(x)  $ to $\beta(x)$ in probability.

\begin{theorem}\label{thm6}
    Suppose $T_0 \sim^{i.i.d}\widetilde{T}_{0}$ and $\{ T_i \}_{i=1}^n$ are strictly increasing, continuous and generated from equation \eqref{eq:model} with $\beta$ such that $-1< \beta(x) <1$ for all $x\in \mathcal{S}$ and with $T_0$ as the initial transport.  Under the assumptions of Theorem \ref{thm5}, if $ E[ ( T_{1}(x)-x)^2 ] > 0 $,
\begin{align*}
|\widehat{\beta}(x) - \beta(x) | = O_p\left(\tau(x) + \frac{1}{ \sqrt{n} } \right),
\end{align*}
where $\tau(x)  = \sup_i E[ |\widehat{T}_i(x)- T_i(x)| ]$.
\end{theorem}

\begin{comment}
{ \color{red}
When $p>1$, let
\begin{align*}
    \bm{T}_{i}^{\bm{\alpha}} : = \alpha_p \odot T_{i-p} \oplus  \dots \oplus  \alpha_1 \odot T_{i-1}(x).
\end{align*}
Notice that 
\begin{align*}
L(\alpha_1, \dots, \alpha_p) = & E_{\varepsilon}\left[  \left( T_{i}(x) - \alpha_p \odot T_{i-p} \oplus  \dots \oplus  \alpha_1 \odot T_{i-1}(x) \right)^2 \right] \\
= &  E_{\varepsilon} \left[  \left(  \varepsilon_i (\bm{T}^{\bm{\alpha}*}_i (x)) -\bm{T}^{\bm{\alpha}}_i (x)  \right)^2 \right] \\
= & (\bm{T}^{\bm{\alpha}}_i (x))^2 - 2 \bm{T}^{\bm{\alpha}*}_i (x) \bm{T}^{\bm{\alpha}}_i (x) +  E_{\varepsilon}\left[ \varepsilon_i^2 (\bm{T}^{\bm{\alpha}*}_i (x)) \right]
\end{align*}
Notice that 
\begin{align*}
   \left(  T_i (x) -\bm{T}^{\bm{\alpha}}_i (x)  \right)^2 - \left(  \widehat{T}_i (x) -\widehat{\bm{T}}^{\bm{\alpha}}_i (x)  \right)^2
\end{align*}
}
\end{comment}

%\bc {\bf \sf 6.\quad NUMERICAL STUDIES} \sm \ec \rs
\section{NUMERICAL STUDIES} \label{sec:6}

\no
%Depending on different choices of $p$, $\alpha$ and $T_i$, 
In the following,  $\text{ATM}_m$ and $\text{CAT}_m$ indicate  models that are based on optimal transport maps $\{T_i\}$  from the Fr\'{e}chet mean to  individual distributions, while  $\text{ATM}_d$ and  $\text{CAT}_d$ indicate  models  based  on optimal transport maps between adjacent distributions. Specifically,  model \eqref{eq:atmm} is  denoted as $\text{ATM}_m$($p$),  model \eqref{eq:atmd}  as $\text{ATM}_d$($p$), model \eqref{eq:vary} with  $T_i = F_{i}^{-1} \circ F_{\mathcal{F}}$  as  $\text{CAT}_m$ and model \eqref{eq:vary} with $T_i = F_{i+1}^{-1} \circ F_{i}$ as $\text{CAT}_d$.

To examine the performance of  these ATMs, we compare them in simulations with a recently proposed autoregressive model for distributional time series  \citep{mull:20:7}, which we refer to as WR (Wasserstein Regression). This approach is based on using manifold logarithmic  maps in the Wasserstein manifold to map  distributions  to a tangent space anchored by the overall barycenter.  Since the tangent space is a subspace of a $L^2$-space, functional linear regression techniques can  be applied in this space, followed by a projection on the convex injectivity set and an application of the exponential  map to get back to  the Wasserstein manifold.   Due to the local linearization this is an extrinsic approach, while the proposed  ATMs  are intrinsic to the Wasserstein manifold. 

We also include comparisons with 
 the log quantile (LQD) approach, which ignores the manifold structure of the distribution space, providing a direct 1:1 mapping of distributions to a Hilbert space by the invertible log quantile transformation or other transformations   \citep{petersen2016functional}. % does not use linearization and thus does not require tedious finite-sample projection onto a convex injectivity region as the other extrinsic methods using tangent bundles do.  
 After applying  the LQD  transformation,  standard autoregressive models for functional time series can be employed in the ensuing Hilbert space  \cp{bosq:00}, followed by mapping back into distribution space by the inverse LQD map.   For autoregressive modeling of functional time series we used the R package ``ftsa".\vspace{.4cm}

 %\vspace{.4cm}

%\no { \sf 6.1 \quad Interpretation of ATMs} \sm  

\subsection{Interpretation of ATMs}

\no We illustrate the process of transporting distributions by ATMs with a simple example. Let  $\mu_1, \mu_2, \mu_3$ be three normal distributions $ N(-1, 2.25),  N(2.5, 1.44)$ and $ N(-2, 0.81)$ (here all distributions are truncated to the interval $[-10,10]$, where the miniscule mass left outside of the truncation interval is ignored). We apply  ATM$_m(2)$ and ATM$_d(2)$  models  to produce the distribution $\mu_4$ at time $t=4$. 
Figure \ref{fig:atmmcoef} illustrates  the densities of $\mu_2, \mu_3$ as well as the density of $\mu_4$ generated by $\text{ATM}_m(2)$,  where the Fr\'echet mean is chosen as the standard normal distribution $\mu_0$.  For the optimal transport maps  $T_2, T_3$ that map $\mu_0$ to $\mu_2$ and $\mu_3$,  respectively, we observe that $T_2$ shifts the density of $\mu_0$ to the right and increases its variance, while $T_3$ shifts $\mu_0$   to the left and decreases its variance.   

Note that $\text{ATM}_m(2)$ transforms $\mu_0$  by first applying transport map $\alpha_1 \odot T_3$,  followed by an application of transport map  $\alpha_2 \odot T_2$. For example, from the first row of the figure, when $\alpha_1 = 0.5, \alpha_2=0.5$, the resulting overall transport is close to the identity, whereas for  $\alpha_1 = 0.5, \alpha_2=0$, the resulting density of $\mu_4$ is on the Wasserstein geodesic connecting $\mu_0$ and  $\mu_3$. When $\alpha_1=0$, it can be seen  from the middle row  of Figure \ref{fig:atmmcoef} that the density of $\mu_4$ is  moving  to the left with decreasing variance when  $\alpha_2$ moves from $0.5$ to $-0.5$. This illustrates  the effect of the changing value of $\alpha_2$ for the transport $\alpha_2 \odot T_2$. Similar effects can be seen in the third row of the figure. 

The  densities of $\mu_1, \mu_2, \mu_3$ and the  density of $\mu_4$ that is obtained by applying the difference based models  $\text{ATM}_d(1)$ and $\text{ATM}_d(2)$ are depicted in  Figure \ref{fig:atmdcoef}.  Denoting by  $T_1$ the optimal transport map that maps  $\mu_1$ to $\mu_2$ and by $T_2$ the  transport map that maps $\mu_2$ to $\mu_3$, one finds  that $T_1$ represents a shift to the right with a simultaneous decrease in variance, while  $T_2$ represents a shift to the left, also  accompanied by a decrease in  variance.  Applying  $\text{ATM}_d$ with $T_1, T_2$ as predictors, i.e. model (\ref{eq:orderp}) with $T_3 = \alpha_1 \odot T_2 \oplus \alpha_2 \odot T_1$,  leads to the transport map $T_3$, which is then  applied to  $\mu_3$, resulting in $\mu_4$. To illustrate the effect of  the coefficients, all panels show that decreasing $\alpha_2$ enhances a shift to the left, while decreasing $\alpha_1$ is associated with a shift to the right. 
\vspace{.4cm} 

\begin{figure}
    \centering
    \begin{tikzpicture}
\matrix (m) [row sep = -1em, column sep = - 1.2em]{ 
	 \node (p11) {\includegraphics[scale=0.35]{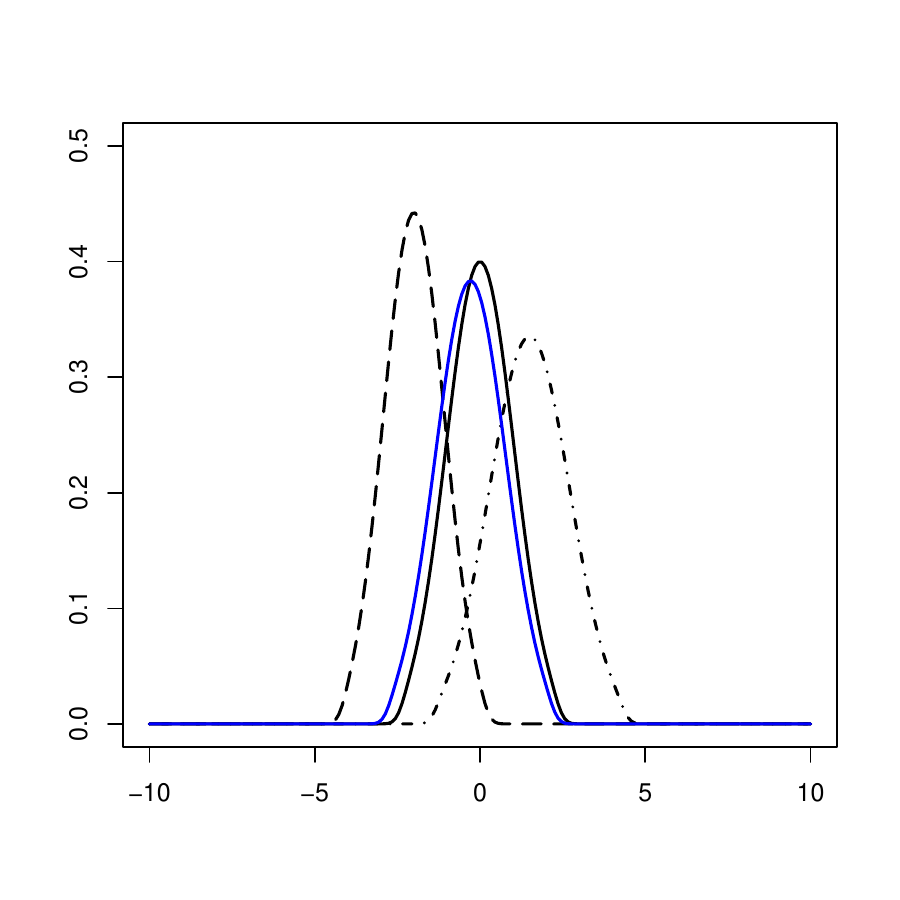}}; 
	 \node[above = -0.8cm of p11] (t11) {$\alpha_1=0.5, \alpha_2=0.5$}; &
	 \node (p12) {\includegraphics[scale=0.35]{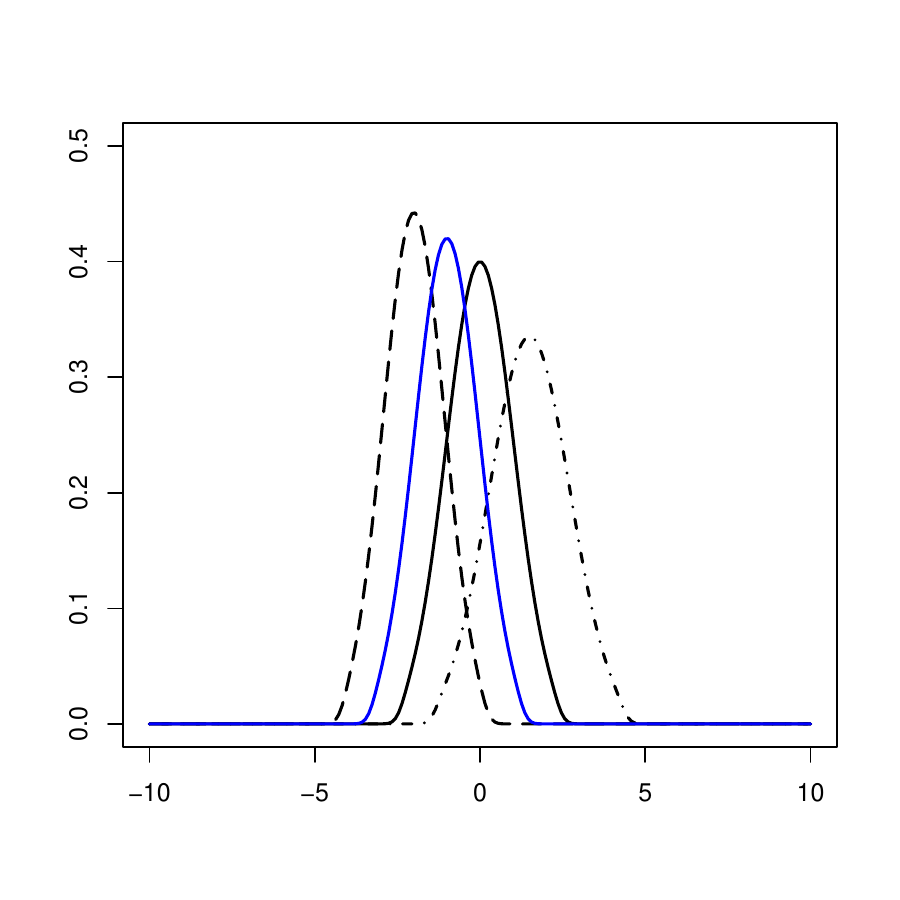}};
	 \node[above = -0.8cm of p12] (t12) {$\alpha_1=0.5, \alpha_2=0$}; &
	 \node (p13) {\includegraphics[scale=0.35]{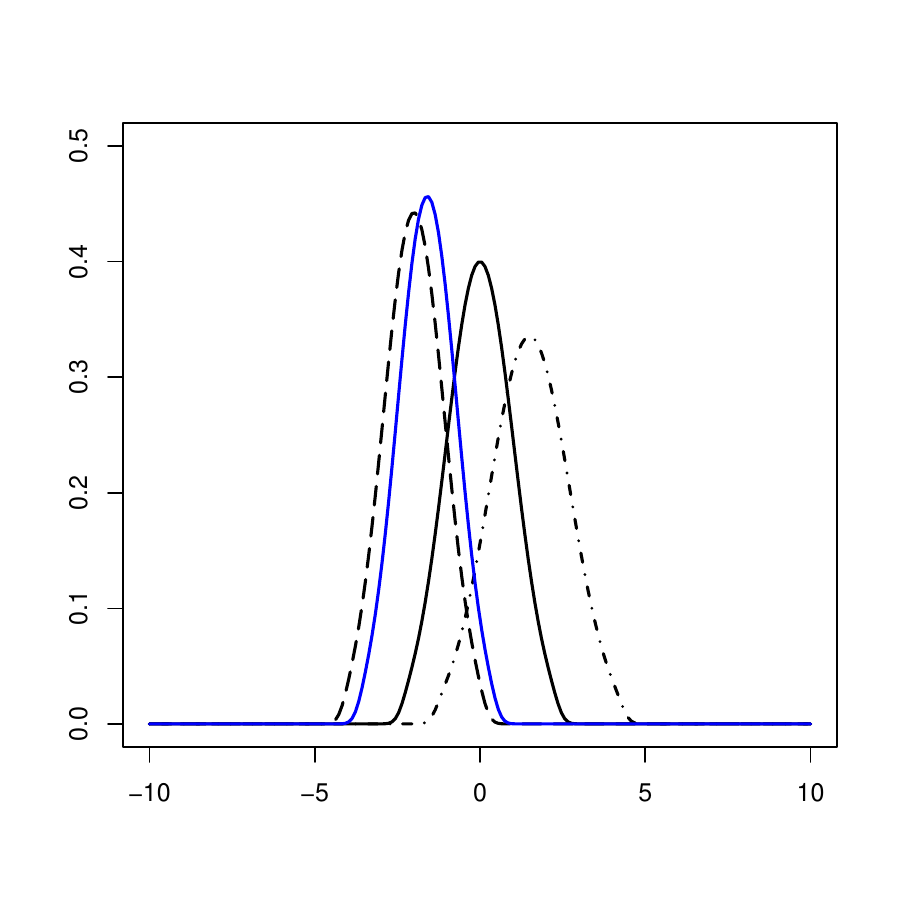}};
	 \node[above = -0.8cm of p13] (t13) {$\alpha_1=0.5, \alpha_2=-0.5$}; 
	 \\ 
	 \node (p21) {\includegraphics[scale=0.35]{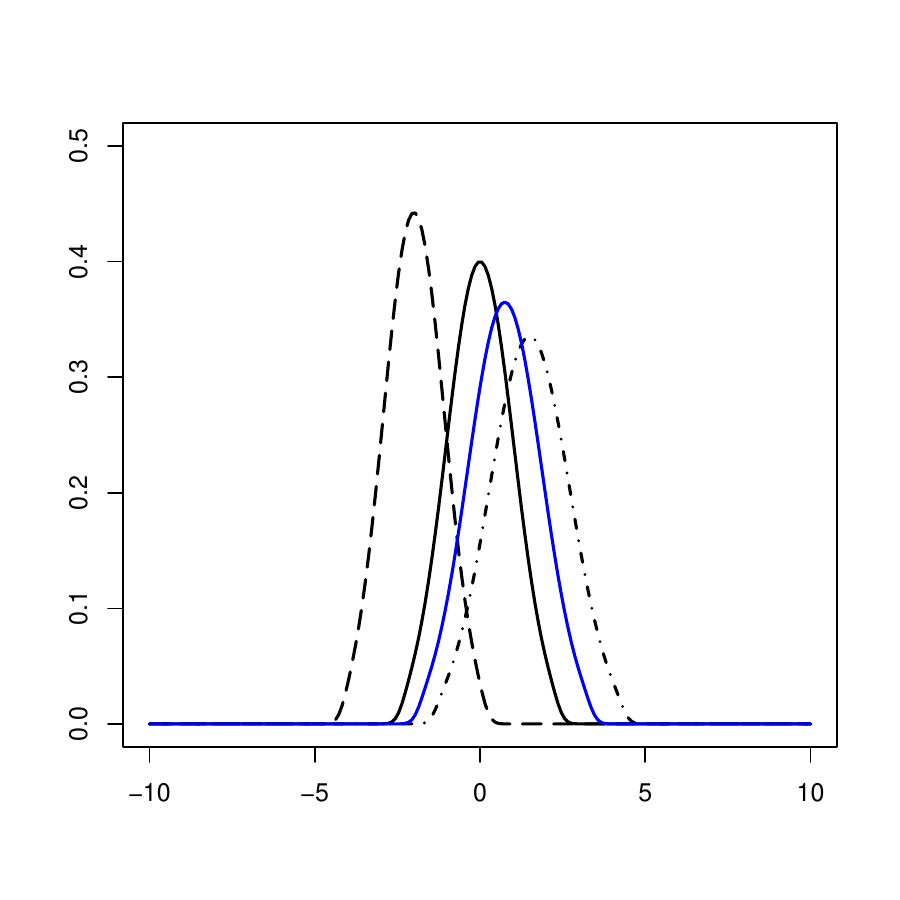}};
	 \node[above = -0.8cm of p21] (t21) {$\alpha_1=0, \alpha_2=0.5$}; &
	 \node (p22) {\includegraphics[scale=0.35]{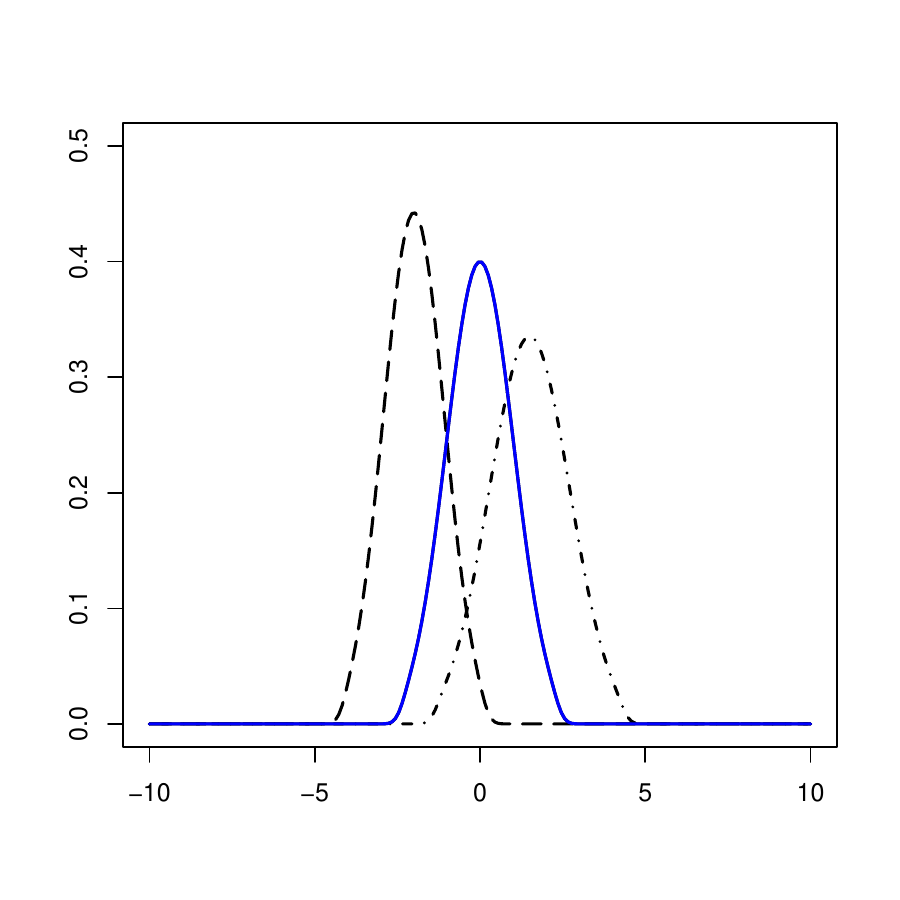}}; 
	 \node[above = -0.8cm of p22] (t22) {$\alpha_1=0, \alpha_2=0$}; &
	 \node (p23) {\includegraphics[scale=0.35]{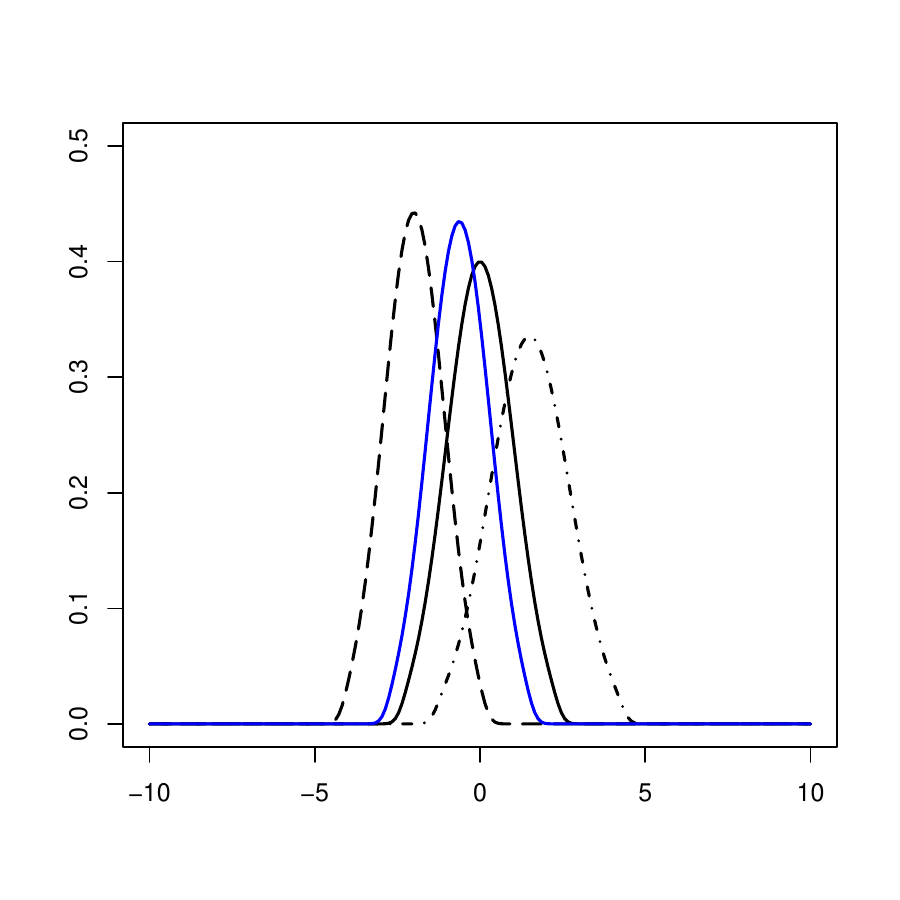}}; 
	 \node[above = -0.8cm of p23] (t23) {$\alpha_1=0, \alpha_2=-0.5$}; 
	 \\
	 \node (p31) {\includegraphics[scale=0.35]{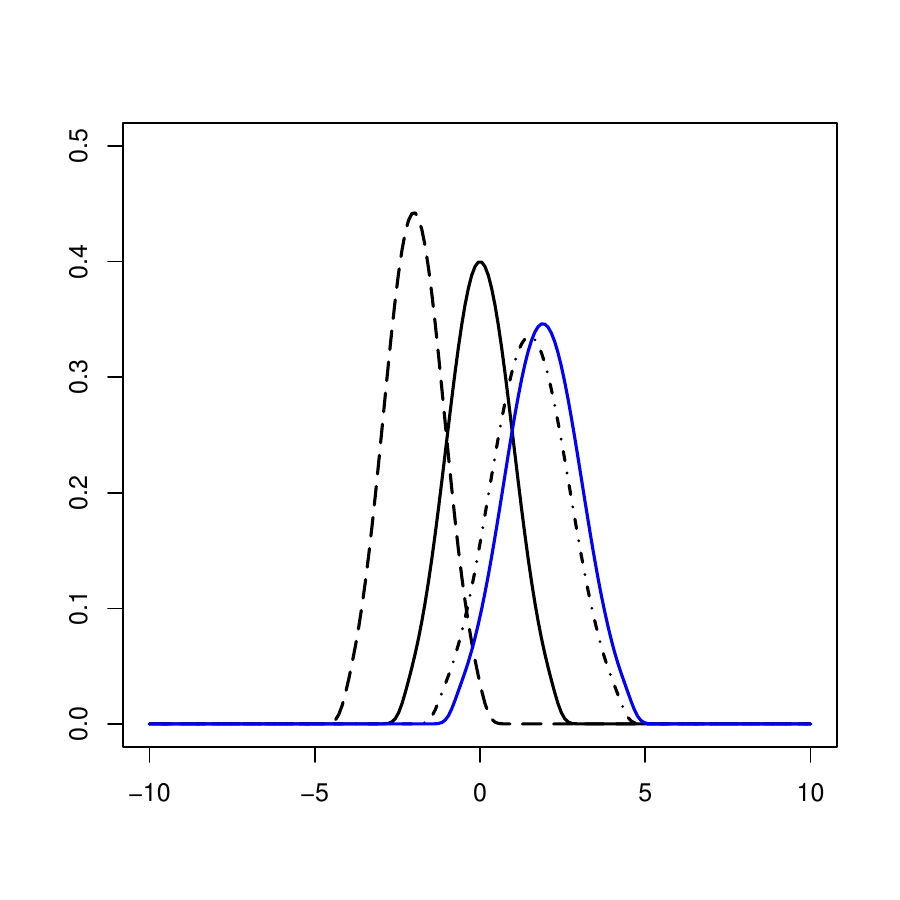}};
	 \node[above = -0.8cm of p31] (t31) {$\alpha_1=-0.5, \alpha_2=0.5$}; &
	 \node (p32) {\includegraphics[scale=0.35]{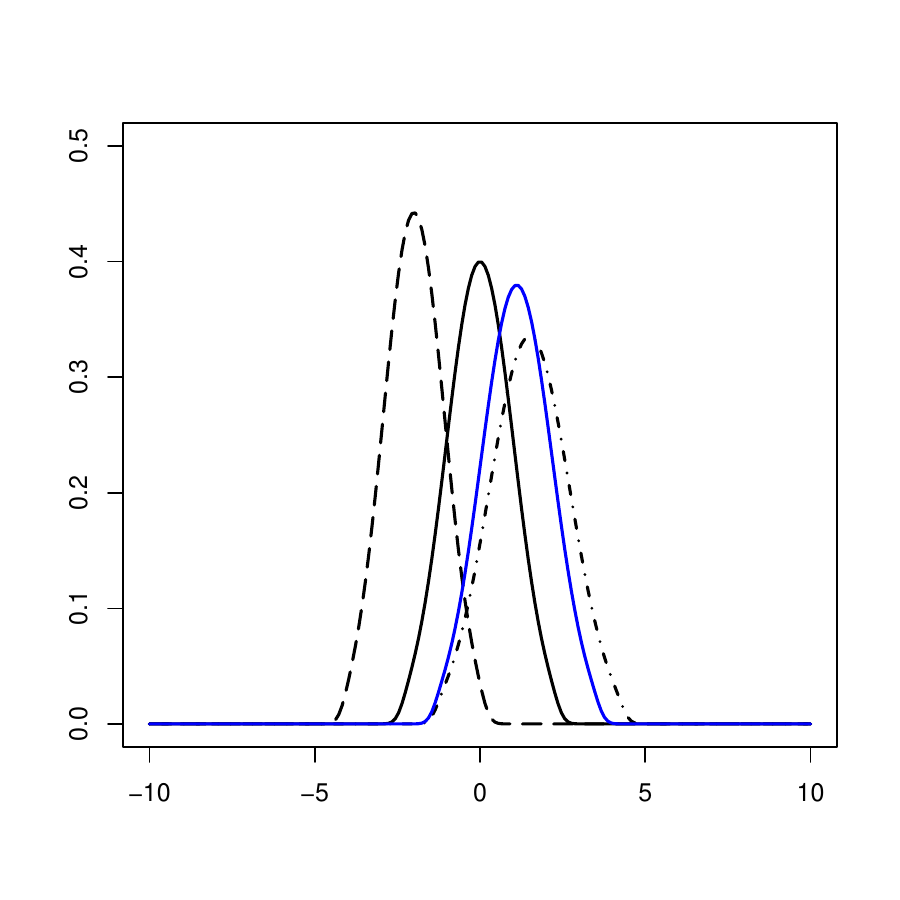}};
	 \node[above= -0.8cm of p32] (t32) {$\alpha_1=-0.5, \alpha_2=0$}; &
	 \node (p33) {\includegraphics[scale=0.35]{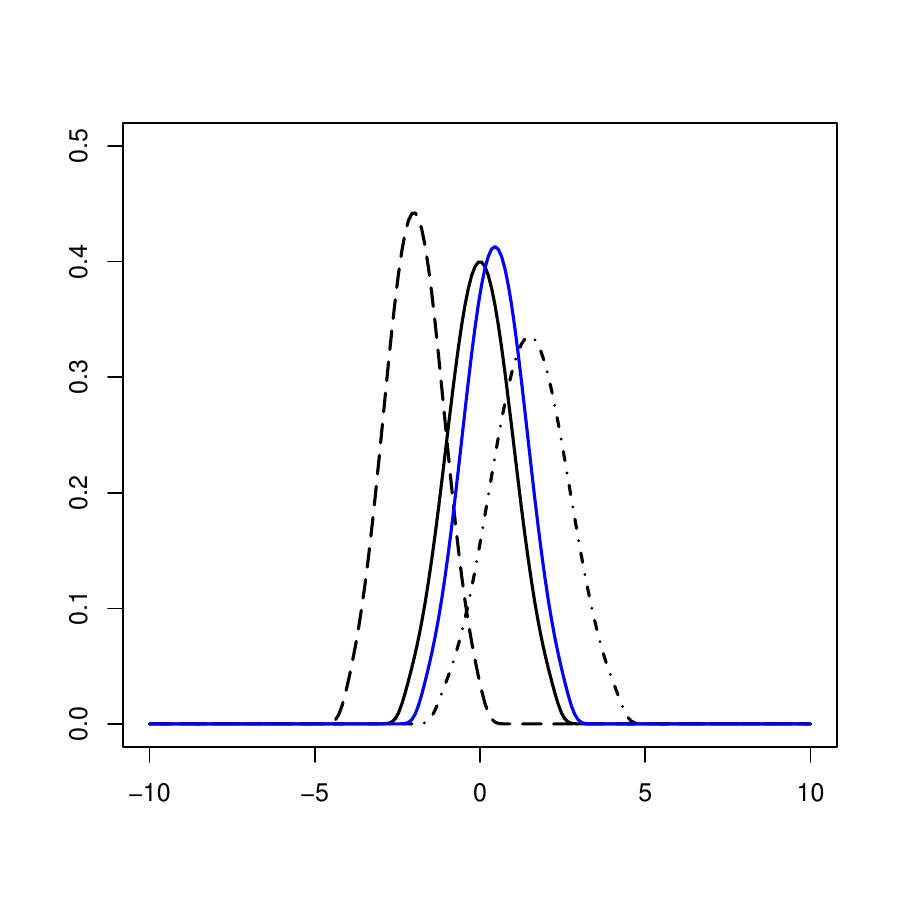}}; 
	 \node[above= -0.8cm of p33] (t33) {$\alpha_1=-0.5, \alpha_2=-0.5$}; 
	 \\
};
%\node[below right=  -0.6cm and -6cm of p32] (legend) {\includegraphics[scale=0.25]{ATM_2021DEC/NonLinear/mlegend.png}} ;
\end{tikzpicture}
    \caption{Illustrating  ATM$_m(2)$ and  ATM$_m(1)$  models for distributional time series. Each panel depicts  the  density functions for distributions $\mu_2$ (dot-dashed), $\mu_3$ (dashed)
   (these are the same across all panels)  and the density of distribution $\mu_4$ generated by $\text{ATM}_m$ (blue), which varies across panels.  For all panels  the density of $\mu_0$ (standard normal) is also included (solid black).}
    \label{fig:atmmcoef}
\end{figure}

\begin{figure}
    \centering
    \begin{tikzpicture}\dots
\matrix (m) [row sep = -1em, column sep = - 1.2em]{ 
	 \node (p11) {\includegraphics[scale=0.35]{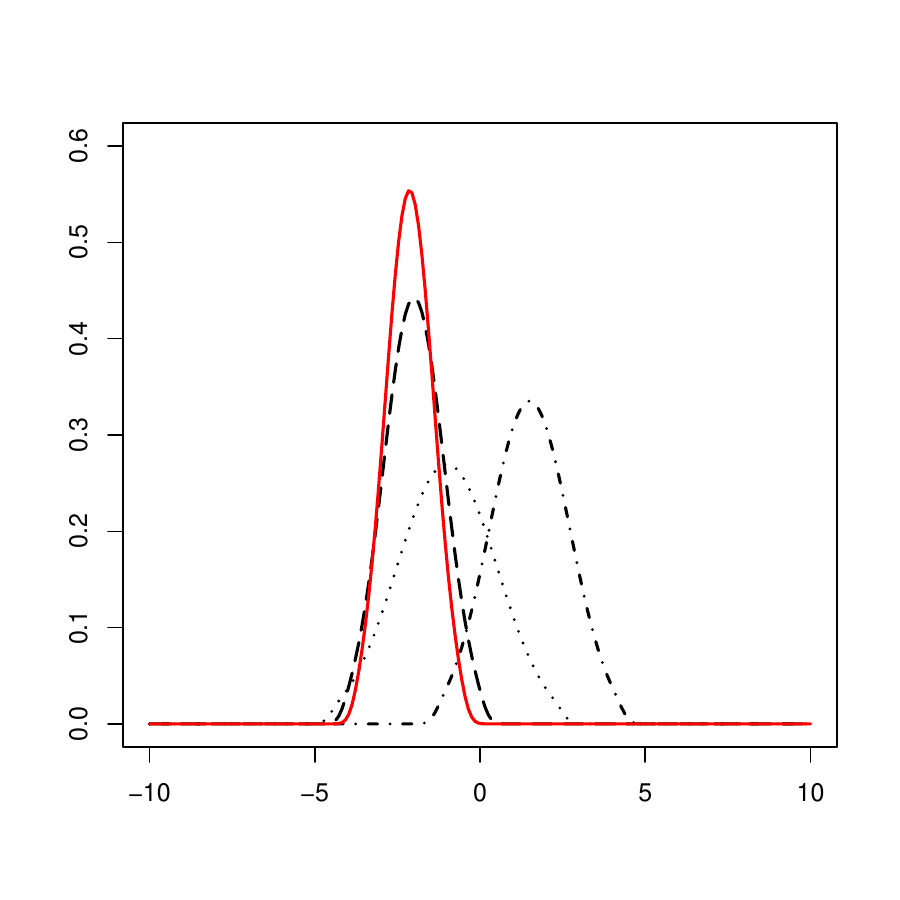}}; 
	 \node[above = -0.8cm of p11] (t11) {$\alpha_1=0.5, \alpha_2=0.5$}; &
	 \node (p12) {\includegraphics[scale=0.35]{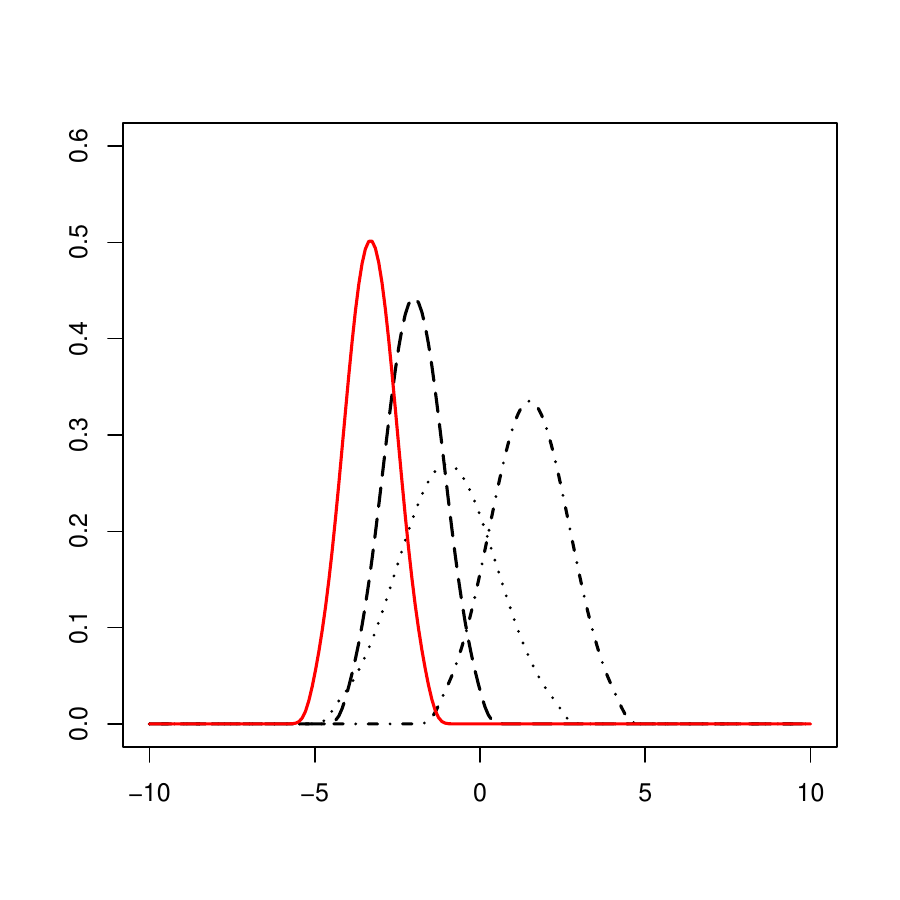}};
	 \node[above = -0.8cm of p12] (t12) {$\alpha_1=0.5, \alpha_2=0$}; &
	 \node (p13) {\includegraphics[scale=0.35]{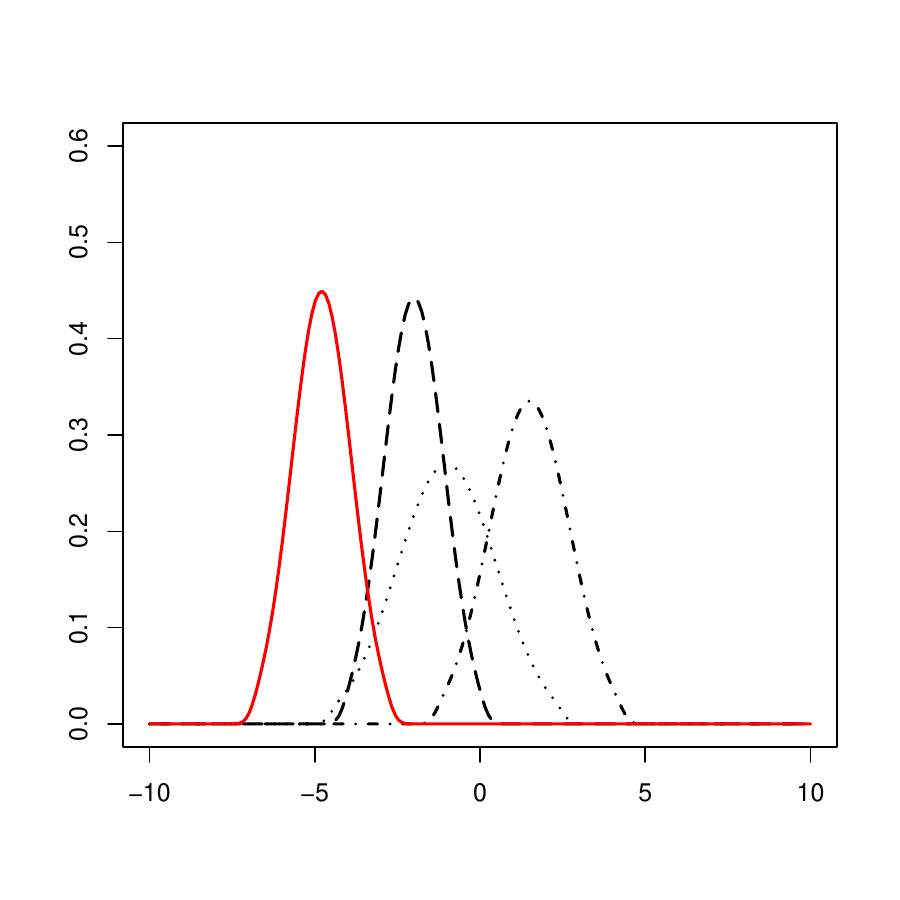}};
	 \node[above = -0.8cm of p13] (t13) {$\alpha_1=0.5, \alpha_2=-0.5$}; 
	 \\ 
	 \node (p21) {\includegraphics[scale=0.35]{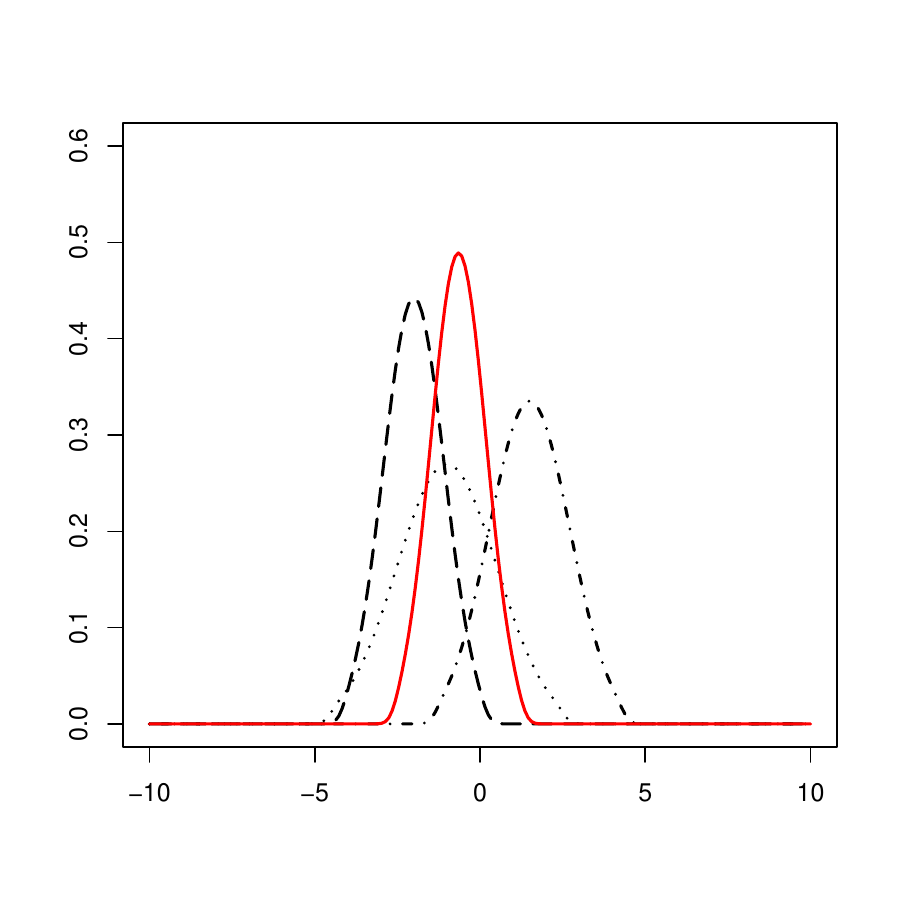}};
	 \node[above = -0.8cm of p21] (t21) {$\alpha_1=0, \alpha_2=0.5$}; &
	 \node (p22) {\includegraphics[scale=0.35]{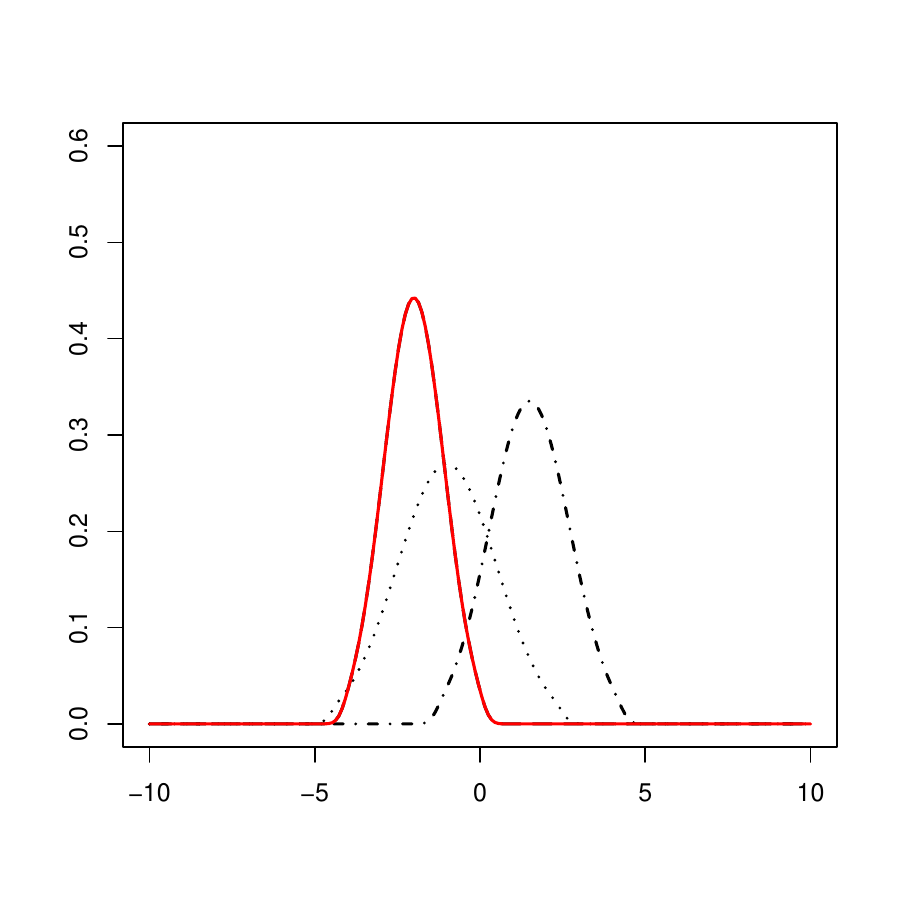}}; 
	 \node[above = -0.8cm of p22] (t22) {$\alpha_1=0, \alpha_2=0$}; &
	 \node (p23) {\includegraphics[scale=0.35]{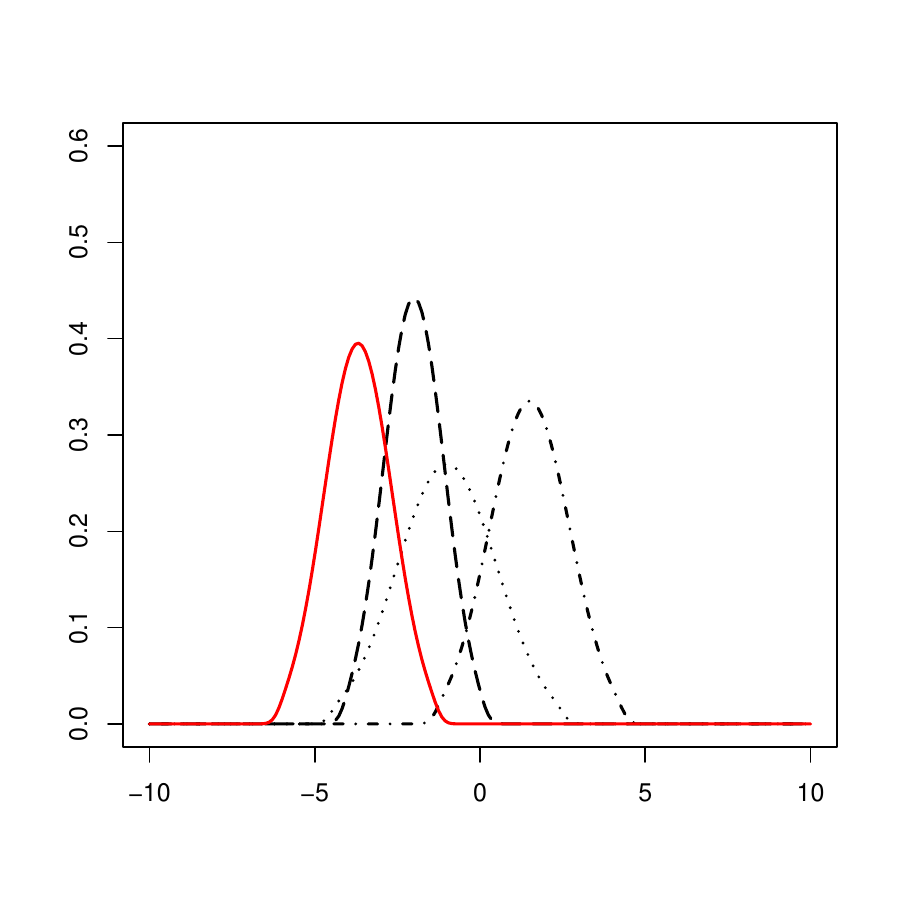}}; 
	 \node[above = -0.8cm of p23] (t23) {$\alpha_1=0, \alpha_2=-0.5$}; 
	 \\
	 \node (p31) {\includegraphics[scale=0.35]{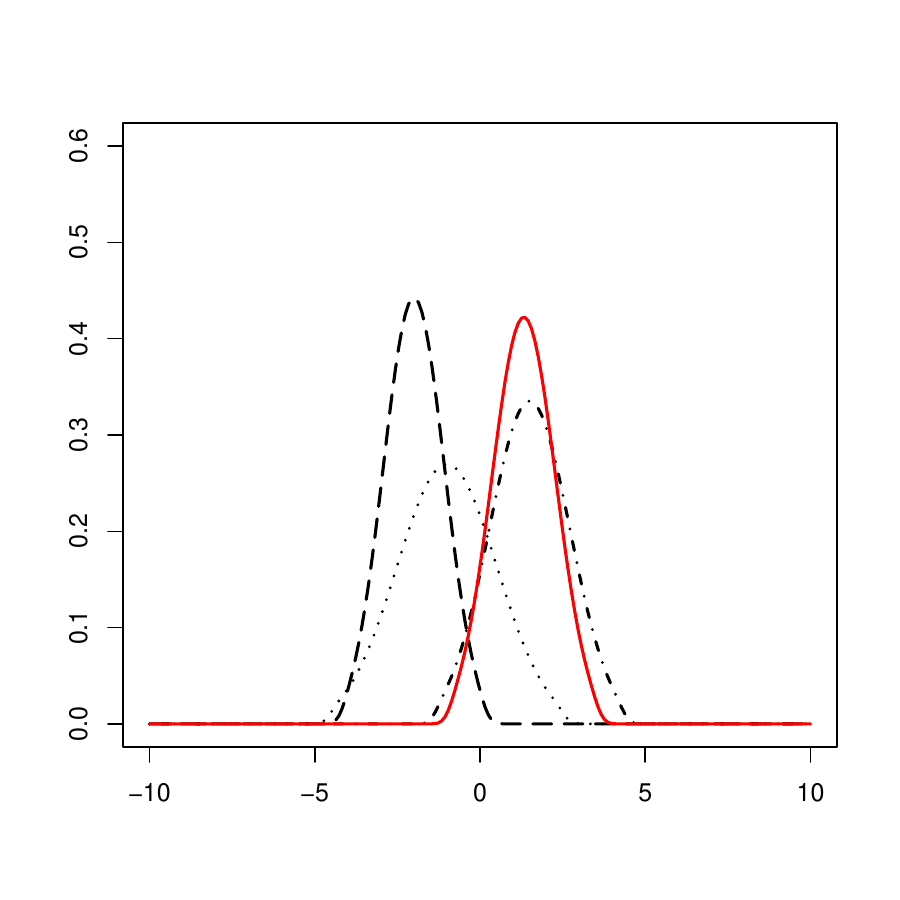}};
	 \node[above = -0.8cm of p31] (t31) {$\alpha_1=-0.5, \alpha_2=0.5$}; &
	 \node (p32) {\includegraphics[scale=0.35]{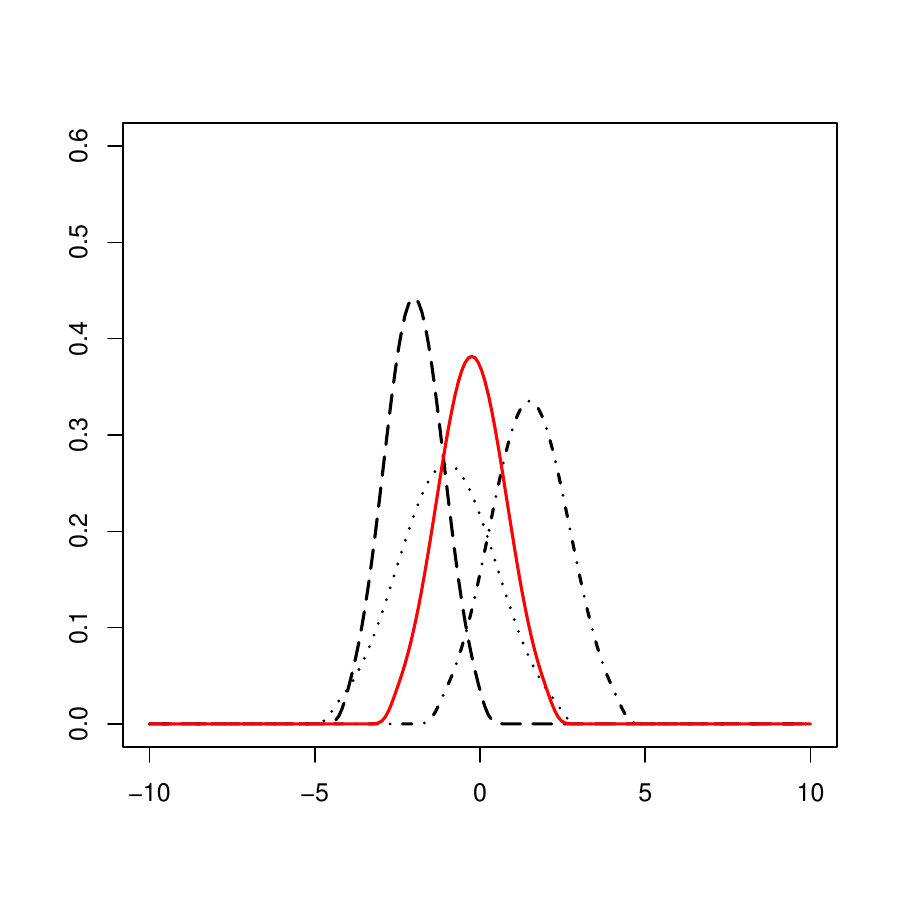}};
	 \node[above= -0.8cm of p32] (t32) {$\alpha_1=-0.5, \alpha_2=0$}; &
	 \node (p33) {\includegraphics[scale=0.35]{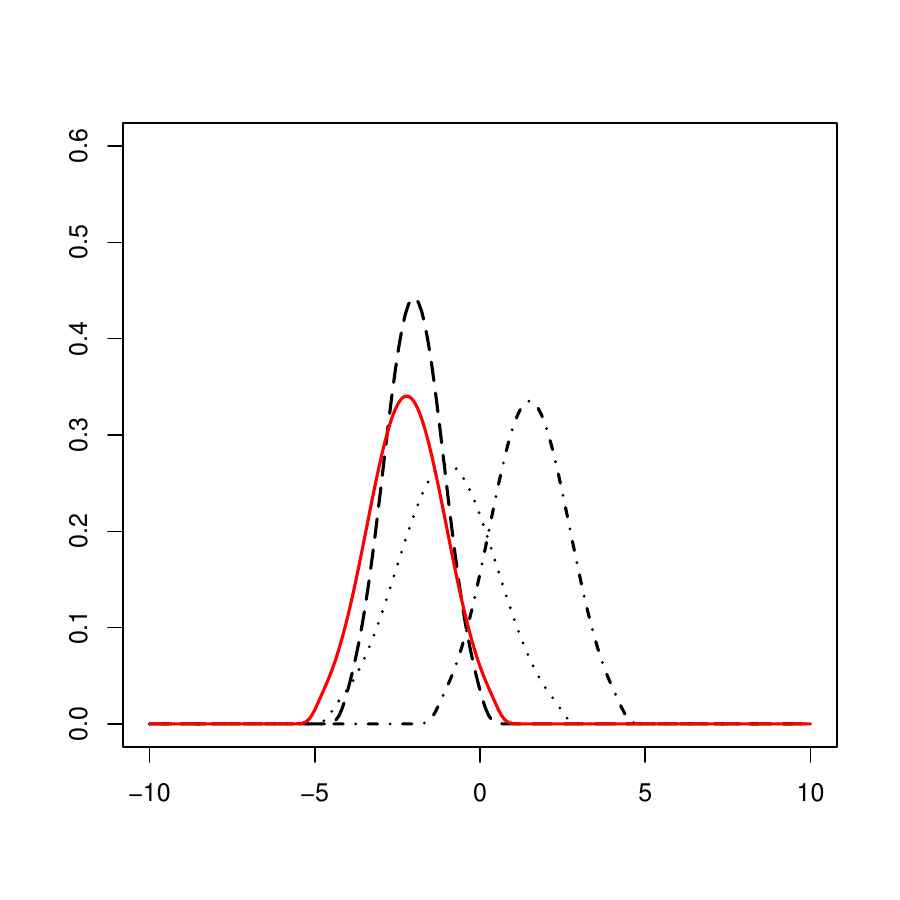}}; 
	 \node[above= -0.8cm of p33] (t33) {$\alpha_1=-0.5, \alpha_2=-0.5$}; 
	 \\
};
\end{tikzpicture}

\caption{Illustrating  ATM$_d(2)$ and  ATM$_d(1)$  models for distributional time series. Each panel depicts  the  density functions for distributions $\mu_1$ (dotted), $\mu_2$ (dot-dashed) and $\mu_3$ (dashed)    (these are the same across all panels) and the density of distribution $\mu_4$ generated by $\text{ATM}_d$ (red), which varies across panels.}
    \label{fig:atmdcoef}
\end{figure}

%\no { \sf 6.2 \quad Reducing Non-stationarity} \sm  
\subsection{Reducing Non-stationarity}

\no The following example  illustrates that the difference-based models  $\text{ATM}_d$ and $\text{CAT}_d$ are advantageous compared to  $\text{ATM}_m$, $\text{CAT}_m$ and $\text{WR}$ if the assumption that $\{\mu_1, \dots, \mu_n\}$ is a stationary sequence does not hold. Stationarity  of the sequence $\{\mu_1, \dots, \mu_n\}$ is a basic assumption for models  $\text{ATM}_m$, $\text{CAT}_m$ and  $\text{WR}$, whereas models  $\text{ATM}_d$ and $\text{CAT}_d$ only require stationarity for differences, i.e. the sequence of optimal transport maps constructed by taking transports between consecutive distributions $\{\mu_1, \dots, \mu_n\}$ as predictors.  

Consider  Gaussian distributions  $\{ \mu_1, \dots, \mu_6 \}$   with mean 0 and decreasing standard deviations $ 4.8, 4, 3, 1.6, 1.15, 1,  $ respectively. This sequence of distributions is non-stationary.  We use $\{ \mu_t: t=1,2, \dots, 6 \}$ as training data and aim to predict the distribution $\mu_7$ with models $\text{ATM}_m(1)$, $\text{ATM}_d(1)$, $\text{CAT}_m$, $\text{CAT}_d$, WR and LQD.
The densities of the training data are visualized  in the left panel of Figure \ref{fig:nonlinear}. One  would  expect $\mu_7$ to follow this trend, i.e. to also have mean 0 with even smaller variance than $\mu_6$. %Figure \ref{fig:nonlinear} (right) 
The right panel shows the predicted densities obtained with the different  methods, and we find  that only $\text{ATM}_d$ and $\text{CAT}_d$  capture the underlying trend and provide  reasonable predictions for next element $\mu_7$ in the sequence. 

\begin{figure}
\centering
\begin{tikzpicture}
\matrix (m) [row sep = -3em, column sep = - 1.2em]{ 
	 \node (p11) {\includegraphics[scale=0.45]{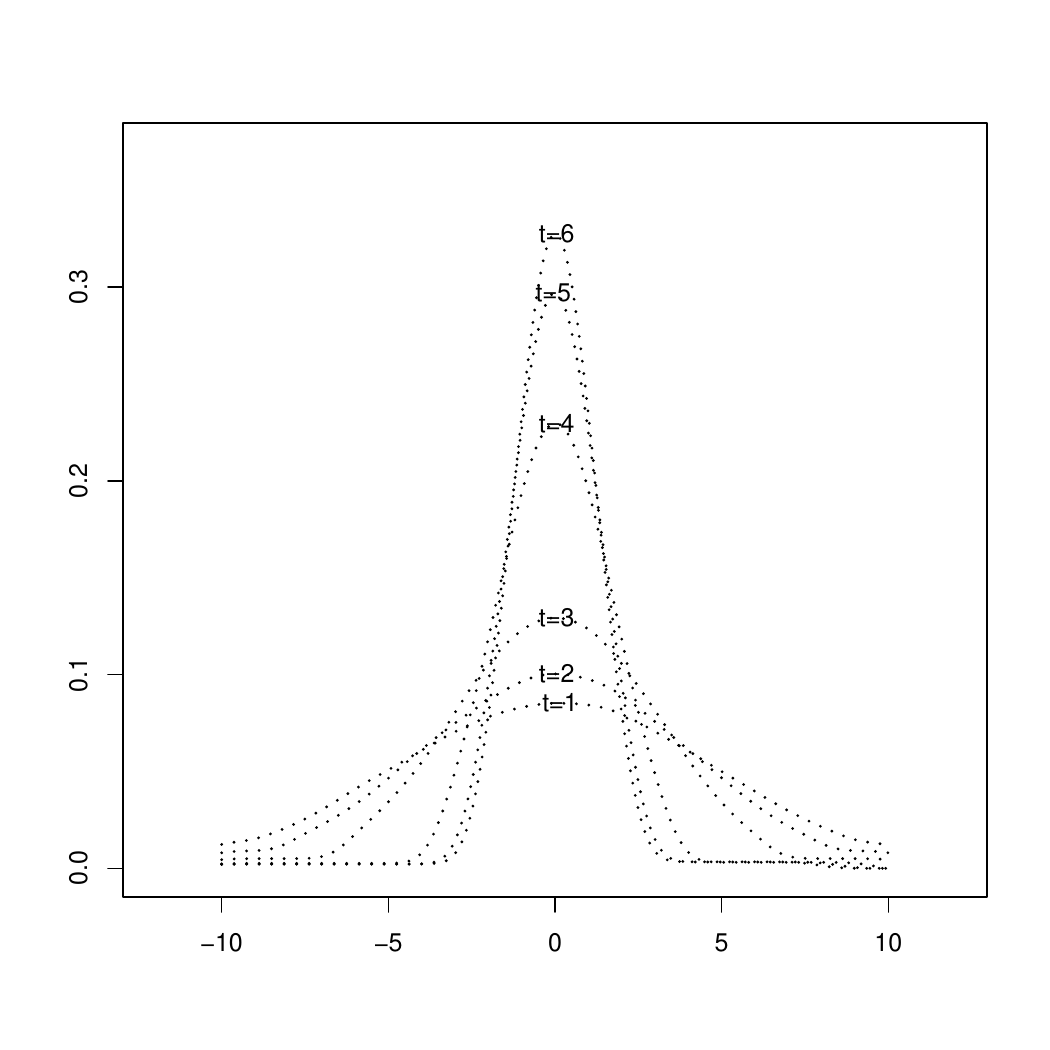}}; &
	 %\node[above left = -1.2cm and -6.5cm of p11] (t11) {Training Set}; & 
	 \node (p12) {\includegraphics[scale=0.45]{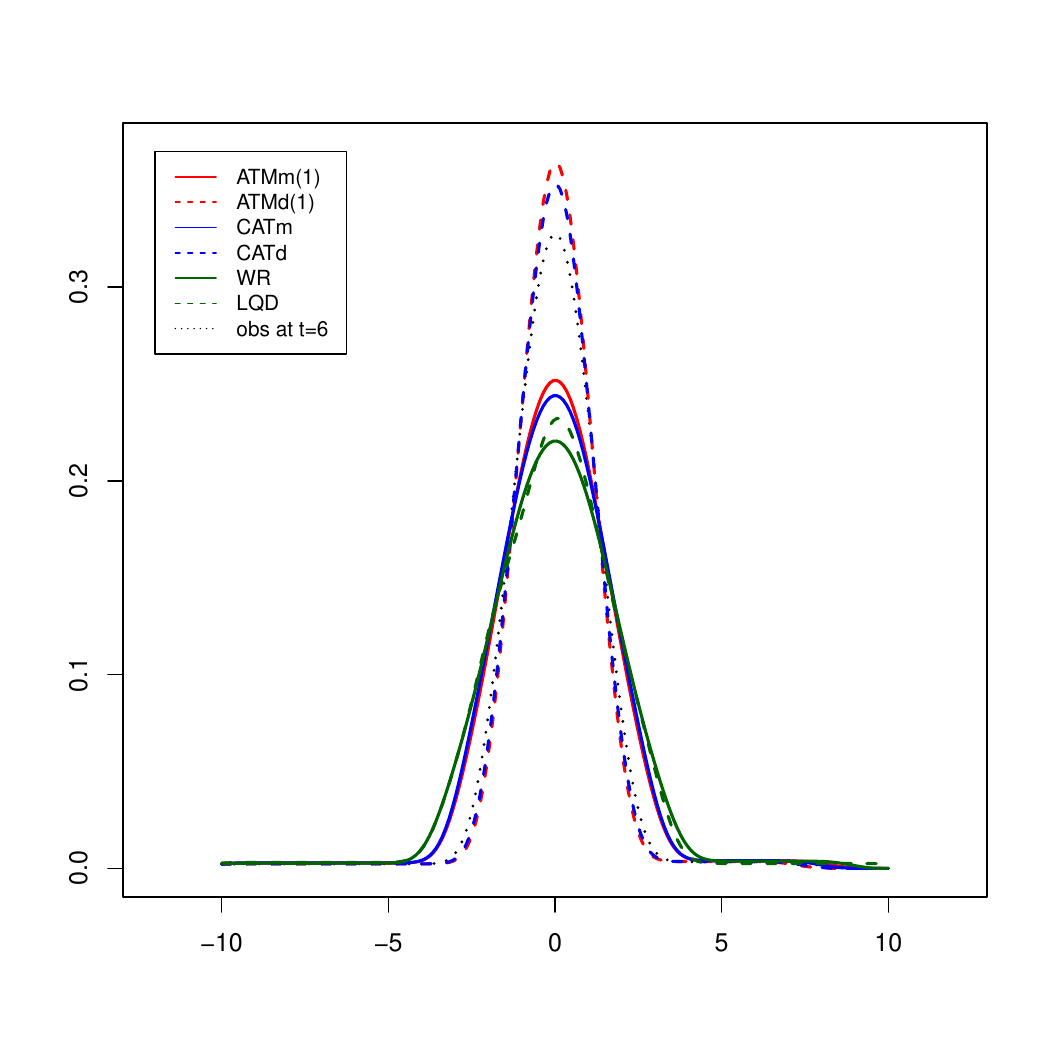}};
	 %\node[above left = -1.2cm and -6.5cm of p12] (t12) {Predictions}; 
	 \\ 
};
\end{tikzpicture}
    \caption{Left panel: The  training sample introduced in Section 6.2. Right panel:  The one-step forecasts obtained for different methods at $t=7$, where only predictions obtained from  $\text{ATM}_d(1)$ and $\text{CAT}_d$ reflect the declining trend in variances, as only these two predictions have smaller variance compared  to the last observed density at $t=6$, which is also plotted on the right panel.}
    \label{fig:nonlinear}
\end{figure}

\vspace{.4cm}

%\no { \sf 6.3 \quad Simulations} \sm  
\subsection{Simulations}

\no We generate random transports according to 
\begin{align}
\label{eq:exp1}
T_i = \alpha_4 \odot T_{i-4} \oplus \alpha_3 \odot T_{i-3} \oplus \alpha_2 \odot T_{i-2} \oplus \alpha_1 \odot T_{i-1} \oplus \varepsilon_{i}, \; i \in \mathbb{Z},
\end{align}
where $\varepsilon_i(x) = \frac{1}{2} ( (1+\xi_i)g (h^{-1}(x)) + (1+\xi_i)h^{-1}(x) ) $, $h(x) = \frac{1}{2} ( (1-\xi_i)g(x) + (1+\xi_i)x )$, $x \in \mathcal{S} = [0,1]$ and $\{ \xi_i \} \sim^{i.i.d} \text{Uniform}(-1, 1)$. Here  $g(x)$ is the natural cubic spline passing through points $(0,0), (0.33, 0.7), (0.66, 0.8), (1,1)$. We note that 
this construction ensures that the $\varepsilon_i$ are transports. When  representing these transports as quantile  functions, for   
$0< \xi_i < 1 $ the function $g(x)$ is shifted along the direction perpendicular to the diagonal towards the identity map and for $ -1< \xi_i < 0$  this shift is applied to  $g^{-1}$ instead; see Figure \ref{fig:f} for an illustration of $g$ and $\varepsilon_i(x)$. By construction, $E(\varepsilon_i) = id$.

To compare prediction accuracy across  different models,  we generated $\{ T_i \}_{i=1}^{101}$ from the above model, using  $ \{ T_i \}_{i=1}^{100} $ as training set, aiming to  predict $T_{101}$. The Wasserstein distance between $T_{101}$ and its prediction was computed for different combinations of $\alpha_1, \alpha_2$ by treating the transport maps $\{T_{i} \}$ as quantile functions. For these comparisons, we modified  LQD to operate on  transport maps, rather than predictor distributions
(as originally devised). The simulation results for  1000 Monte Carlo replications are in Table \ref{tab:exp1} (numbers multiplied by 100). The order of $\text{ATM}_m$ was obtained by rolling-window validation based on a pre-sample of size 50. When $\alpha_2 = 0 $, model \eqref{eq:exp1} reduces to an autoregressive model of order 1. Overall, ATM was found to outperform  WR and LQD. We also  use this example with $h(x)$ chosen as  natural cubic spline passing through $(0,0), (0.3, 0.5), (0.6, 0.8), (1,1)$ to illustrate the empirical rate of  convergence  of the estimates for the  parameters of $\text{ATM}_m(1)$. Figure \ref{fig:rate} displays  estimation error versus $\sqrt{n}$ based on 200 Monte Carlo repetitions,  demonstrating that  finite sample  performance with increasing sample sizes matches  the root-$n$ convergence rate predicted by Theorem 2. 

It is also of interest to consider a sequence of distributions that are not generated from any of the examined models. Starting with  the sequence of square integrable functions
\begin{align} \label{eq:exp2}
    R_i(x)  =  \text{sin}(\zeta_i x), 
\end{align}
where $x \in [0,1]$ and the $\{ \zeta_i \}$ are generated from the AR(2) model, $ \zeta_i = \alpha_1 \zeta_{i-1} + \alpha_2 \zeta_{i-1}+ \alpha_3 \zeta_{i-3} + \alpha_4 \zeta_{i-4} + \epsilon_i $, $  \epsilon_i \overset{i.i.d.}{\sim}\text{Uniform}(-4 \pi, 4 \pi)$, we  convert the $\{ R_i \}$ to distributions by applying the inverse log quantile density transformation \citep{petersen2016functional}, scaling the resulting distributions to be supported on $[0,1]$; see Figure \ref{fig:f} (right panel) for an illustration. Again, we generate 100 distributions for training  and report the results for 1000 Monte Carlo replications. The simulation results are in the lower part of Table \ref{tab:exp1}. For this case, we find that $\text{ATM}_m(1)$ is the overall preferred model. \vspace{.4cm} 

\begin{figure}
		\centering
		\begin{tikzpicture}[scale = 0.4]
		\node[anchor=south west] (straw) at (0,0) {\includegraphics[width=.3\textwidth, height=.3\textwidth]{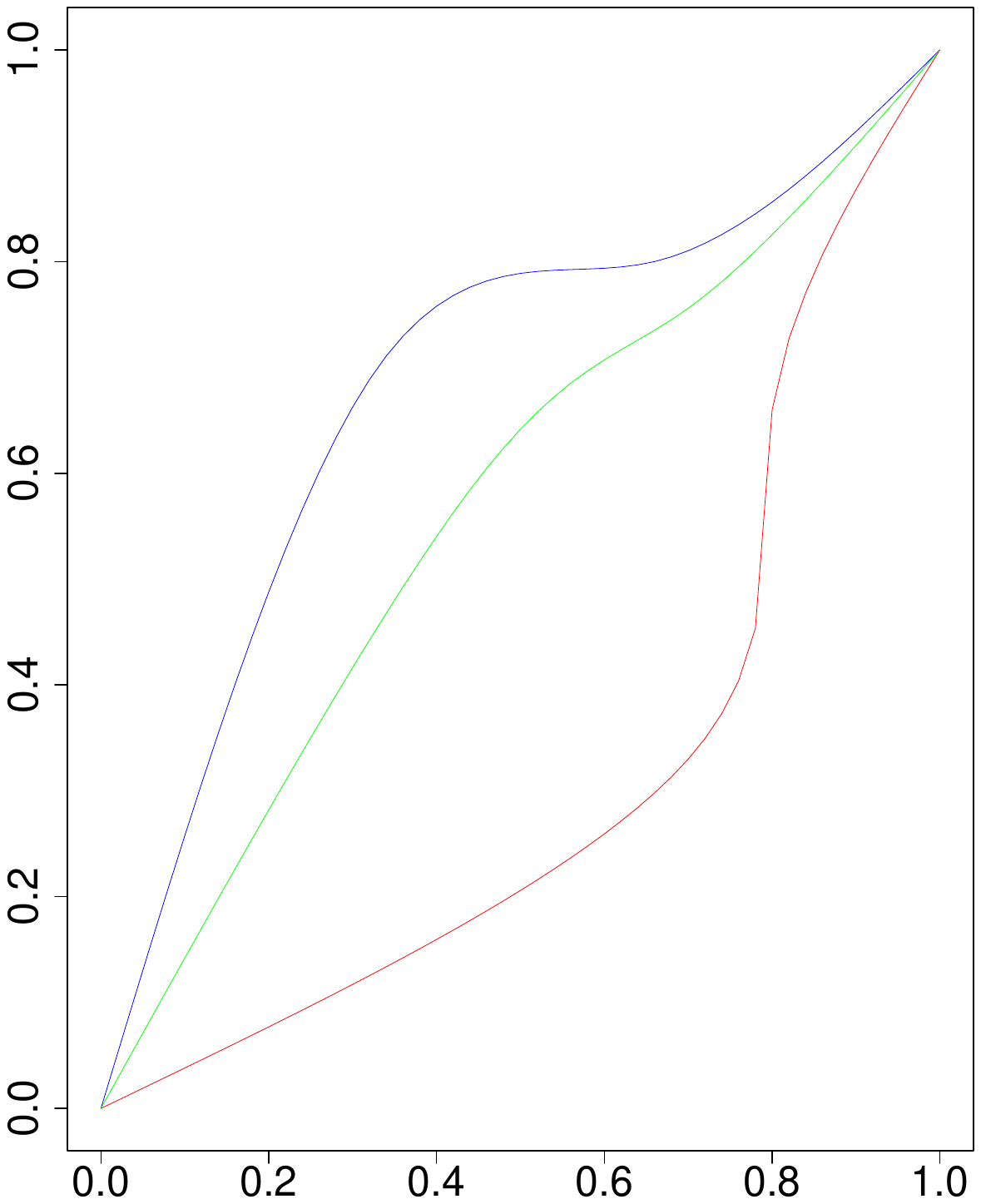}};
		
		\node[anchor=south west] (Tseq) at (straw.south east) {\includegraphics[width=.3\textwidth, height=.3\textwidth]{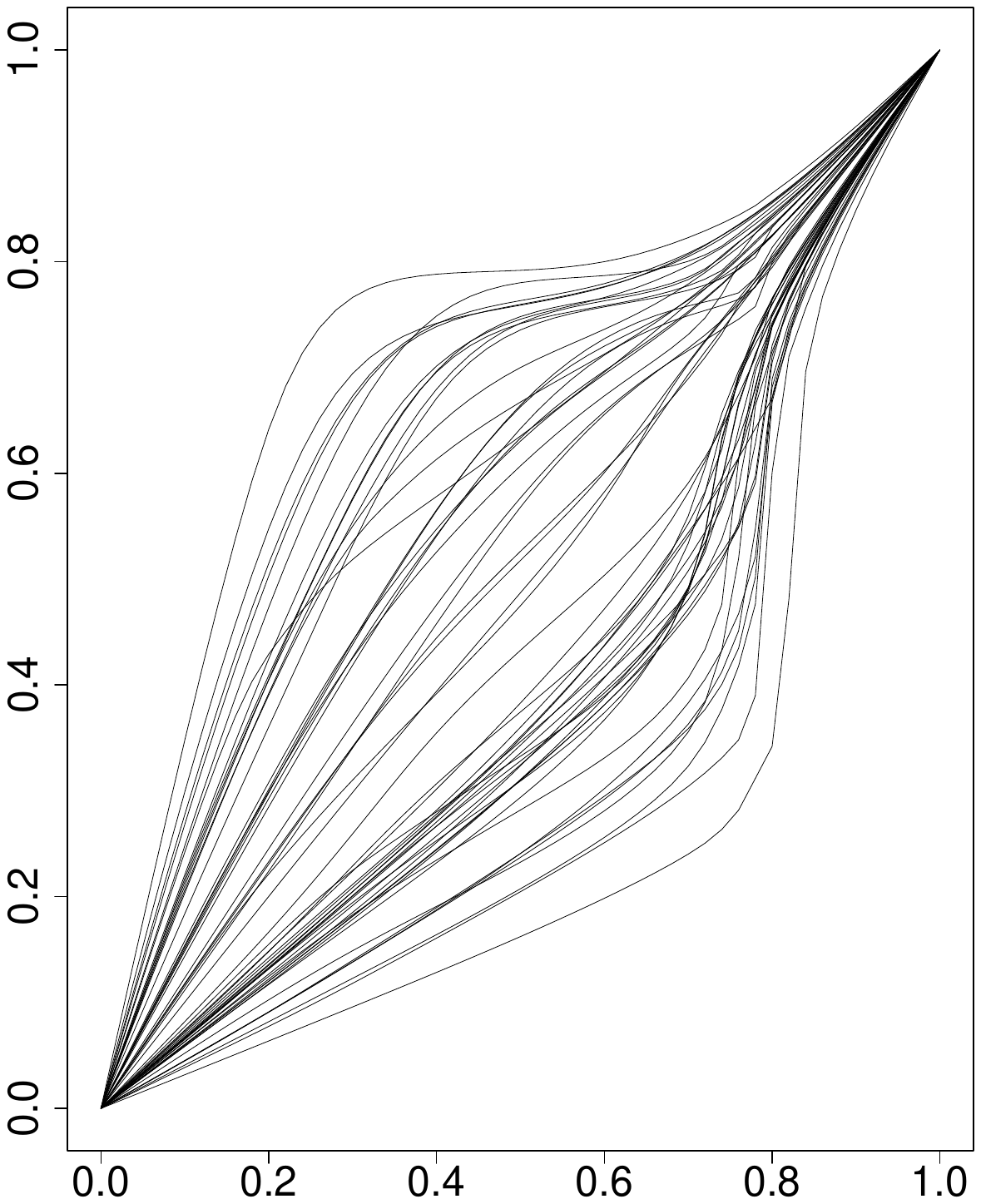}};
		
	   	\node[anchor=south west] at (Tseq.south east) {\includegraphics[width=.3\textwidth, height=.3\textwidth]{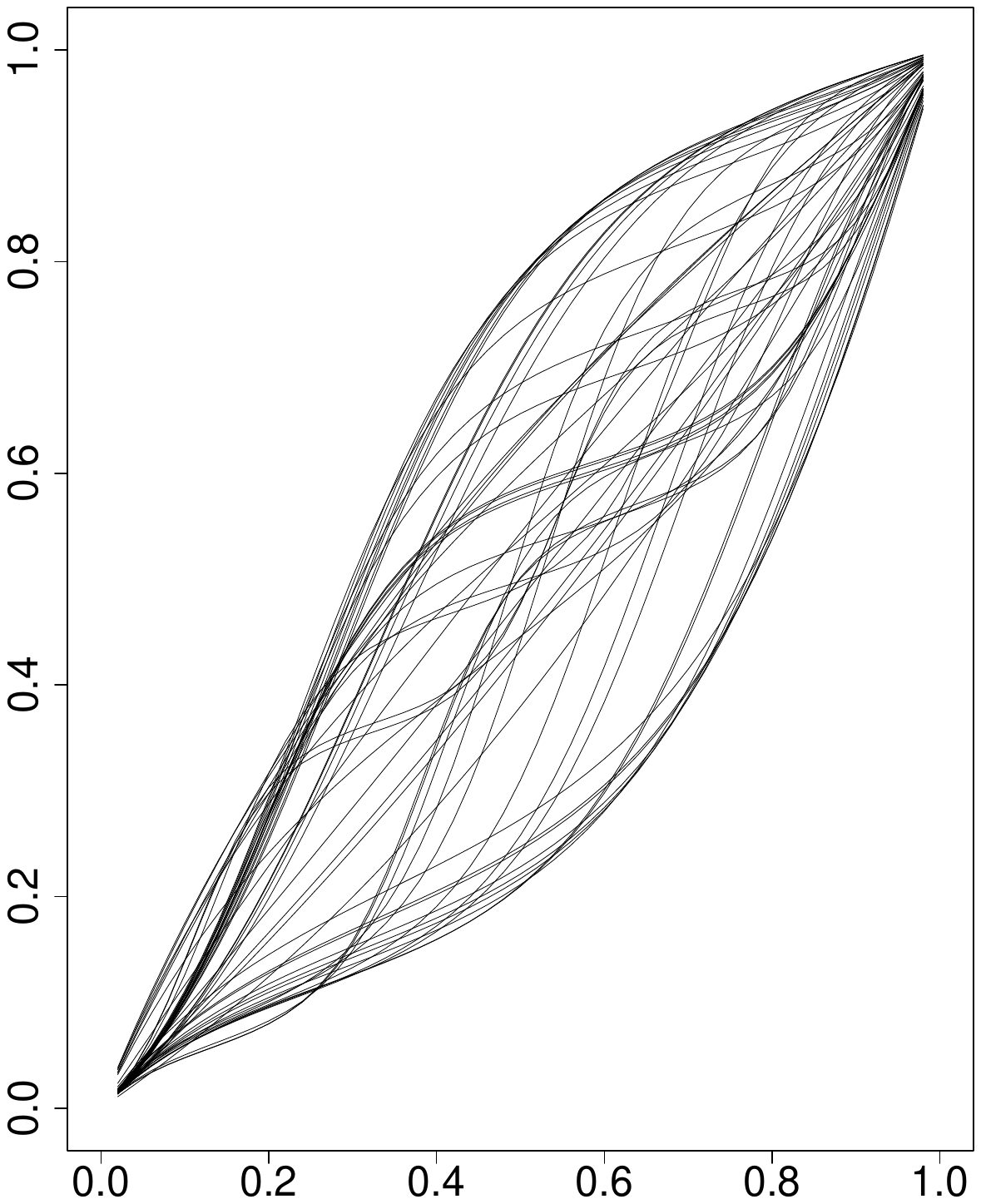}};
		\end{tikzpicture}
		\caption{Auxiliary functions for the simulation. Left panel: The monotone function $g$ (blue), $g^{-1}$ (red) and $\epsilon_i$ with $\xi_i = -0.7$ (green) in the simulation \eqref{eq:exp1}. Middle panel:  Sequence of generated transport maps $T_i$ for simulation  \eqref{eq:exp1} with  $\alpha_1 = -0.3$, $\alpha_2 = 0.2$. Right panel: Quantile functions generated for  simulation  \eqref{eq:exp2} with $\alpha_1 = -0.3$, $\alpha_2 = 0.2$.}
		\label{fig:f}
	\end{figure}

\begin{figure}
    \centering
    \begin{tikzpicture}
\matrix (m) [row sep = -3em, column sep = - 1.2em]{ 
	 \node (p11) {\includegraphics[scale=0.5]{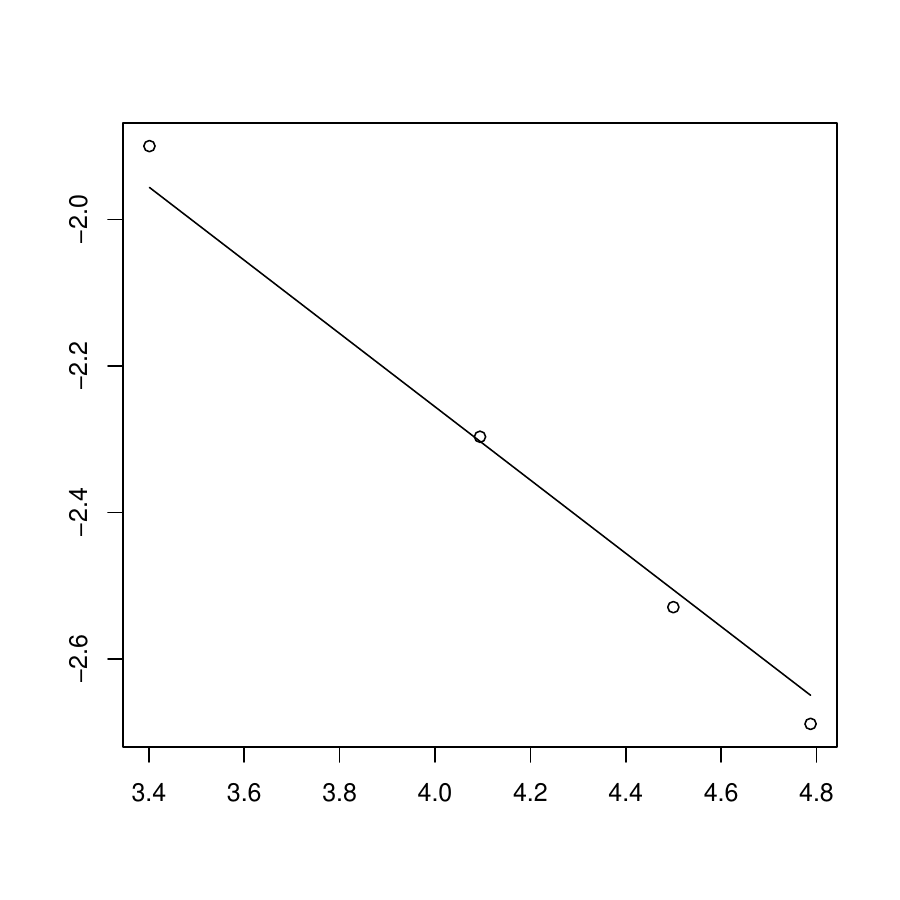}}; 
	 \node[below = -1cm of p11] (t11) {$\text{ln}(n)$};  \node[above left = -3cm and -0.65cm of p11, rotate=90] (t11) {$\text{ln}(|\widehat{\alpha} - \alpha|)$}; & 
	 \node (p12) {\includegraphics[scale=0.5]{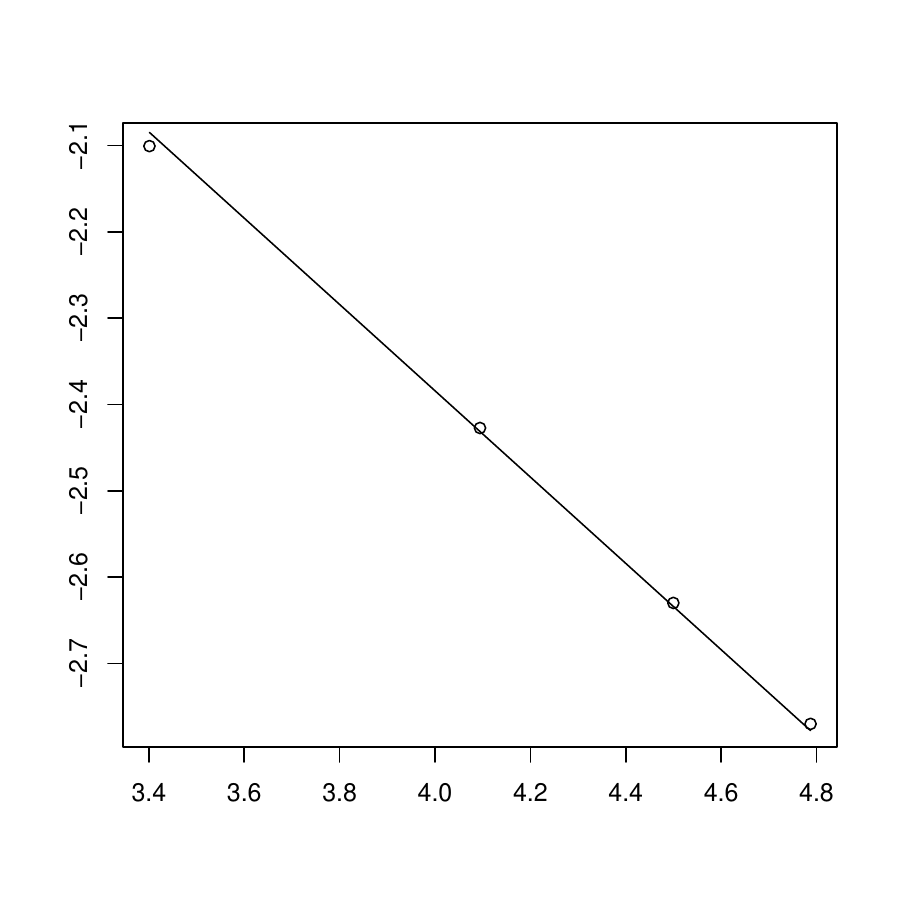}};
	 \node[below  = -1cm of p12] (t12) {$\text{ln}(n)$}; 
	 \\ 
};
\end{tikzpicture} \vspace{-.5cm}
    \caption{Log-estimation error of $\hat\alpha$ versus log sample size $n$ for $\text{ATM}(1)$, for  $\alpha = 0.5$ (left) and $\alpha=-0.5$(right). The solid black line in each panel is a line with slope -0.5 that is predicted by theory.}
    \label{fig:rate}
\end{figure}

\begin{table}[t] \centering 
\begin{tabular}{@{\extracolsep{5pt}} cccccc} 
\\[-1.8ex]\hline 
\hline \\[-1.8ex] 
& ($\alpha_1$, $\alpha_2$, $\alpha_3$, $\alpha_4$) & (0.2, -0.5, 0.1, -0.3) & (0.5, 0, 0, 0)   \\ \hline
\multirow{3}{*}{Example \eqref{eq:exp1}}  & $\text{ATM}_m$ & \bf{12.264}  & \bf{11.586}   \\
  & LQD & 13.891   & 13.282    \\
  & WR & 12.535  & 11.765    \\
\hline \\[-1.8ex] 
\multirow{3}{*}{Example \eqref{eq:exp2}}  & $\text{ATM}_m$ & \bf{9.841} & \bf{9.644}  \\
  & LQD & 10.079 & 9.836 \\
  & WR & 10.082 & 9.838  \\
\hline \\[-1.8ex] 
\end{tabular} 
\caption{Forecasting accuracy comparison for simulations \eqref{eq:exp1} and \eqref{eq:exp2}.}
\label{tab:exp1} 
\end{table} 

\subsection{Temperature Data}

\no One consequence of global warming may be an increasing frequency  of warm summer nights in the Northern hemisphere. Inspired by the article of \cite{bhat:21},
we studied this with 
%and it is of  interest to model the temperature distributional time series. The 
temperature data that  were recorded at O'Hare international airport (available at  \url{https://www.ncdc.noaa.gov/cdo-web/search?datasetid=GHCND}).  The annual distributions of daily minimum temperatures,  aggregating these temperatures over the period June 1 to September 30 over the summer months of each year,  % was constructed from the sample that pulls the daily observations together. 
are illustrated in Figure 
\ref{fig:ohare} for the years from  1990 to 2019, where we use the distributions prior to 2019 as training data to predict the distribution for the  year 2019. 

For the  ATM models we  varied  $p$ from 1 to 3 and found  that $p=3$ yielded the best prediction. The observed and predicted  densities for 2019 are shown in Figure \ref{fig:ohare}. %where $\text{ATM}_d(3)$ gives the best prediction results. 
The Wasserstein distances  between observed  and predicted distributions  were found to be  0.334  for  $\text{ATM}_d(3)$, 1.01  for $\text{ATM}_m(3)$, 0.462 for $\text{CAT}_d$, 1.477  for $\text{CAT}_m$, 1.134  for WR, and  1.255  for LQD. The fitted  model coefficients for the best model, i.e. $\text{ATM}_d(3)$, are $ \alpha_1 = -0.724, \alpha_2 = -0.5, \alpha_3 = -0.268$. Denote by  $\mu_{2018}$, $\mu_{2017}$, $\mu_{2016}$, $\mu_{2015}$  the observed distributions for the years 2018, 2017, 2016, 2015,  respectively,  and by $T_3$ be the optimal transport from  $\mu_{2015}$ to $\mu_{2016}$, by  $T_2$ the optimal transport from $\mu_{2016}$ to $\mu_{2017}$ and by  $T_1$ the optimal transport from  $\mu_{2017}$ to $\mu_{2018}$. The training set of distributions, i.e. the observed data,  is illustrated in the form of densities in the left panel of 
Figure \ref{fig:ohare},  predicted densities are in the middle panel and the 
densities of $\mu_{2015}, \dots, \mu_{2018}$ in the right panel. 

Comparing the densities of $\mu_{2017}$ and $\mu_{2018}$, $ \mu_{2016} $ and those of  $\mu_{2017}$, $\mu_{2015}$ and $\mu_{2016}$, respectively, we find
that  $T_1$ corresponds to a shift to the right and a sharpening of the distribution, $T_{2}$ corresponds to a shift to the left and a smoothing of the distribution and $T_3$ corresponds to a shift to the right and a sharpening of the distribution. The proposed model applies deformations  $ \alpha_3\odot T_3 $, $\alpha_2 \odot T_2$ and  $\alpha_1 \odot T_1$ sequentially to $\mu_{2018}$.
%%as visualized the d. The predicted distribution by $\text{ATM}_d(3)$ at year 2019 is given by $(\alpha_1 \odot T_{1} \oplus \alpha_2 \odot T_2 \oplus \alpha_3 \odot T_3)_{\#} \mu_{2019}$, which is the same as first applying $\alpha_3 \odot T_3$, i.e, to generate $\nu_3 := (\alpha_3 \odot T_3)_{\#} \mu_{2019}$, then applying $ \alpha_2 \odot T_2 $ on $\nu_3$ to generate $\nu_2$, i.e. $\nu_2 = (\alpha_2 \odot T_2)_{\#}(\nu_3)$, and finally apply $\alpha_1 \odot T_1$ on $\nu_2$ to generate the predicted distribution, which is denoted as $\nu_1$. The densities of $\nu_3, \nu_2, \nu_1$ are plotted in the bottom right part of Figure \ref{fig:ohare}. 
it is likely that $\text{ATM}_d$ and $\text{CAT}_d$ yield the best results because of   the non-stationarity of this sequence, as  the distributions shift to the right over the years, reflecting  a warming
trend.  \vspace{.4cm}

\begin{figure}
    \centering
    \begin{tikzpicture}
\matrix (m) [row sep = -1em, column sep = - 1.3em]{ 
	 \node (p11) {\includegraphics[width=.34\textwidth, height=.34\textwidth]{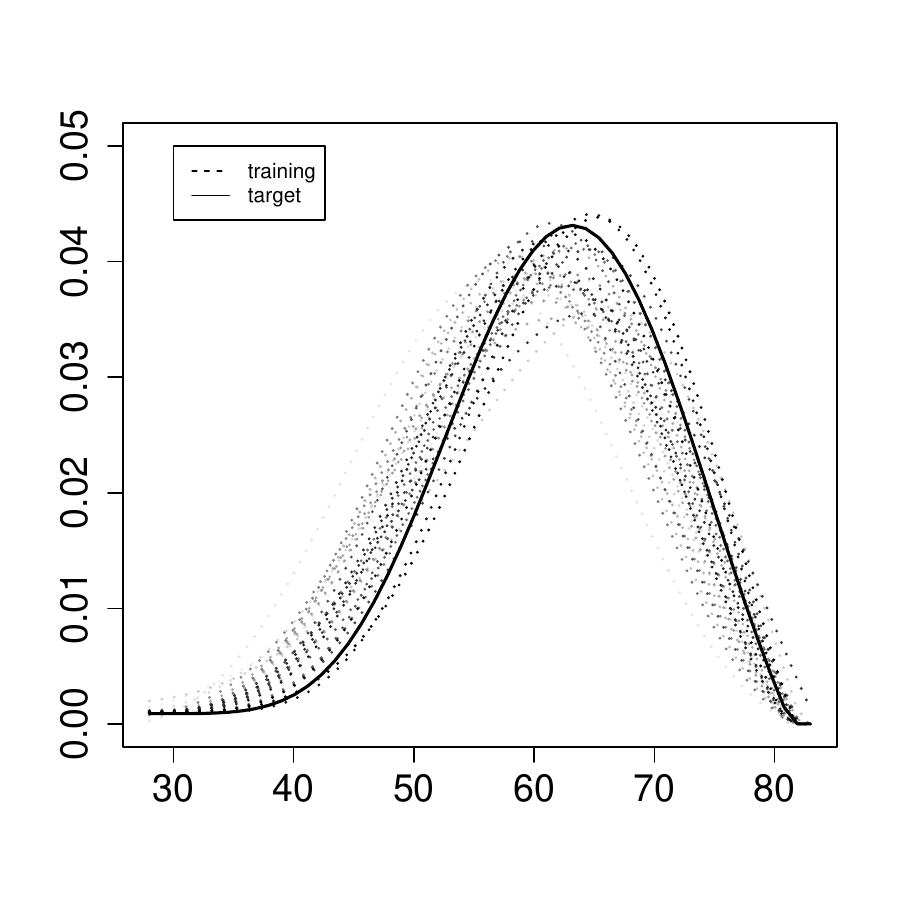}}; &
	 %\node[above left = -1.2cm and -6.5cm of p11] (t11) {Training Set}; & 
	 \node (p12) {\includegraphics[width=.34\textwidth, height=.34\textwidth]{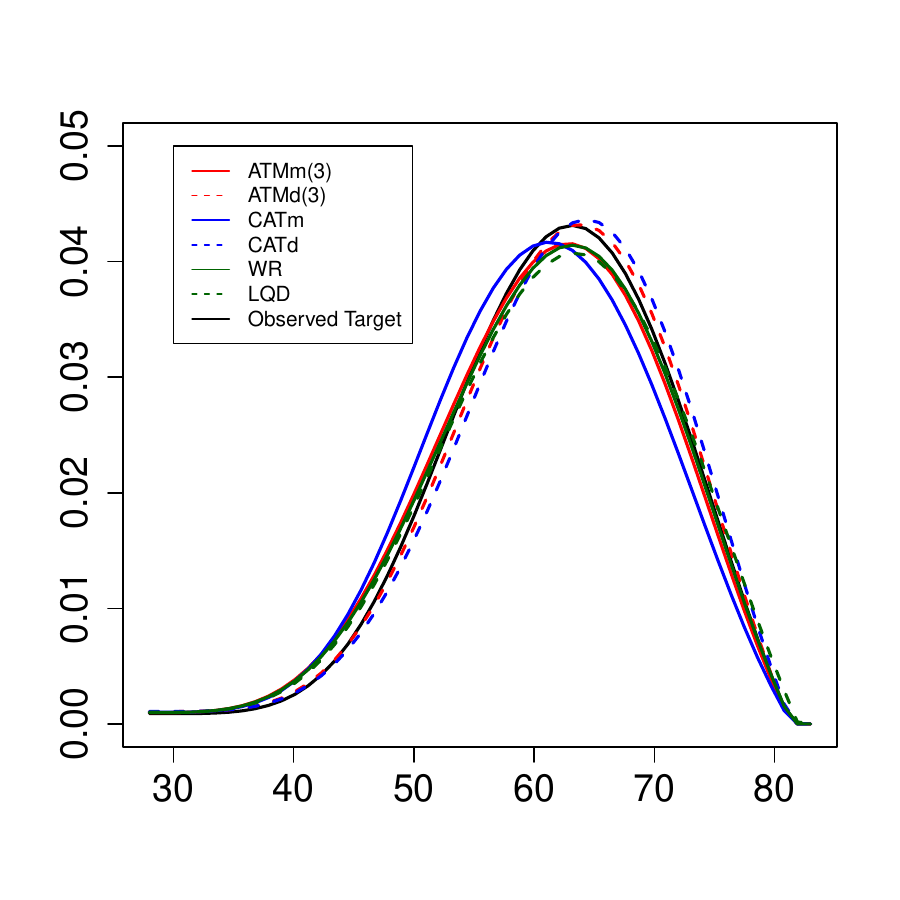}};
	 %\node[above left = -1.2cm and -6.5cm of p12] (t12) {Predictions}; 
	 &
	  \node (p21) {\includegraphics[width=.34\textwidth, height=.34\textwidth]{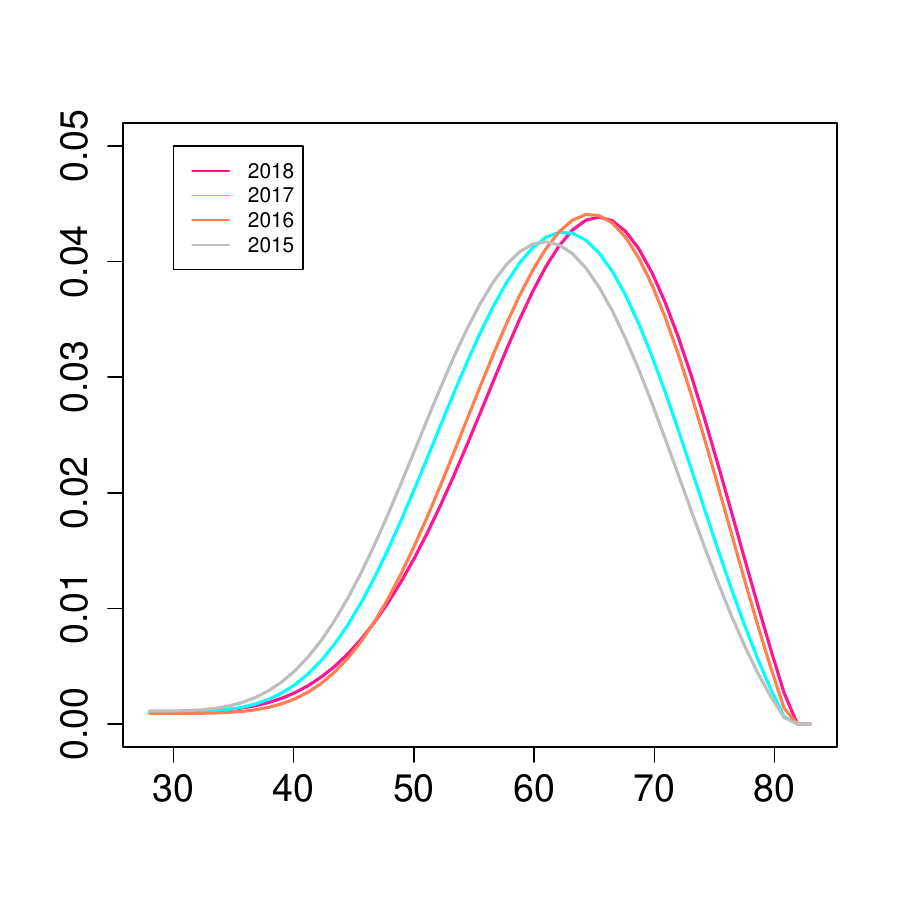}}; 
	 \\
};
\end{tikzpicture}\vspace{-1cm}
    \caption{Left panel: Densities of the annual distributions of minimum summer night temperatures at O'Hare International Airport from 1990-2018.  Middle panel: Observed density and predicted densities obtained from various models for the year 2019. Right figure shows the densities of $\mu_{2018}$, $\mu_{2017}$, $\mu_{2016}$, $\mu_{2015}$ that are the observed distributions for years 2018, 2017, 2016, 2015 respectively.}
    \label{fig:ohare}
\end{figure}

\vspace{-.8cm} 

%\no  { \sf 6.5 \quad U.S. House Price Data} \sm  
\subsection{U.S. House Price Data}
 
\no Given the sequence of distributions $\{ \mu_1, \mu_2,  \dots, \mu_n \}$, for a starting time $s_r \in \{k + 1, k+2, \dots, n-k \} $, we used the subset $\{ \mu_{s_r}, \mu_{s_r + 1}, \dots, \mu_{s_r + k-1} \} $ to train models and to produce the prediction $\widehat{\mu}_{s_r+k}$ at time $s_r+k$. The autoregressive order $p$ was  selected so as to minimize $\sum_{i=s_r}^{s_r+k-1} d_{\mathcal{W}}(\mu_i, \widehat{\mu}_i)$, where $ \widehat{\mu}_i $ is the predicted distribution at time $i$ by $\text{ATM}(p)$ trained on the sample $\{ \mu_{i-k}, \dots, \mu_{i-1}\}$. The candidate set for $p$ was  $\{ 1,2,3,4,5 \} $ when $ k=8 $ and $ \{ 1,2,3,4,6,8 \} $ when $k >8$. We adopted the rolling window approach  \citep{zivot2007modeling} and used the prediction loss $\sum_{s_r=k+1}^{n-k} d_{\mathcal{W}}(\mu_{s_r + k}, \widehat{\mu}_{s_r+k})/(n-2k).$ % in the comparisons.  %, computing the Wasserstein distances $ \{ d_{\mathcal{W}}(\mu_{s_r + k} \}$ always on the same grid.

%\noindent { \sf 6.2.1 \quad US House Price}

The  US house price data  contain bimonthly median house prices for 306 U.S. cities and counties  from June 1996 to August 2015 (available at http://www.zillow.com).  We adjusted the data to account for  inflation by a monthly adjustment factor (deflator) and constructed the bimonthly house price distributions over the 306 cities/counties. The preprocessed  distributions (equivalently density or quantile functions) were  then scaled to be supported on $[0,1]$. Figure \ref{fig:real} 
%\red{the figure needs to come after it is mentioned, please check this for all figures} 
presents the  house price densities over time. Setting   the learning rate $\eta=1$ in  algorithm \ref{alg2}, the prediction results are summarized in Table \ref{tab:house}. In general, $\text{ATM}_d$ emerged as  the best performing model for these data, which is not surprising due to the non-stationarity of these data.  % and the lowest prediction error is achieved by $\text{ATM}_d$ when the training size is $12$. 

	\begin{figure}[t]
		\centering
		\begin{tikzpicture}[scale = 0.6]
		\node[anchor=south west] (straw) at (0,0) {\includegraphics[width=.6\textwidth, height=.4\textwidth]{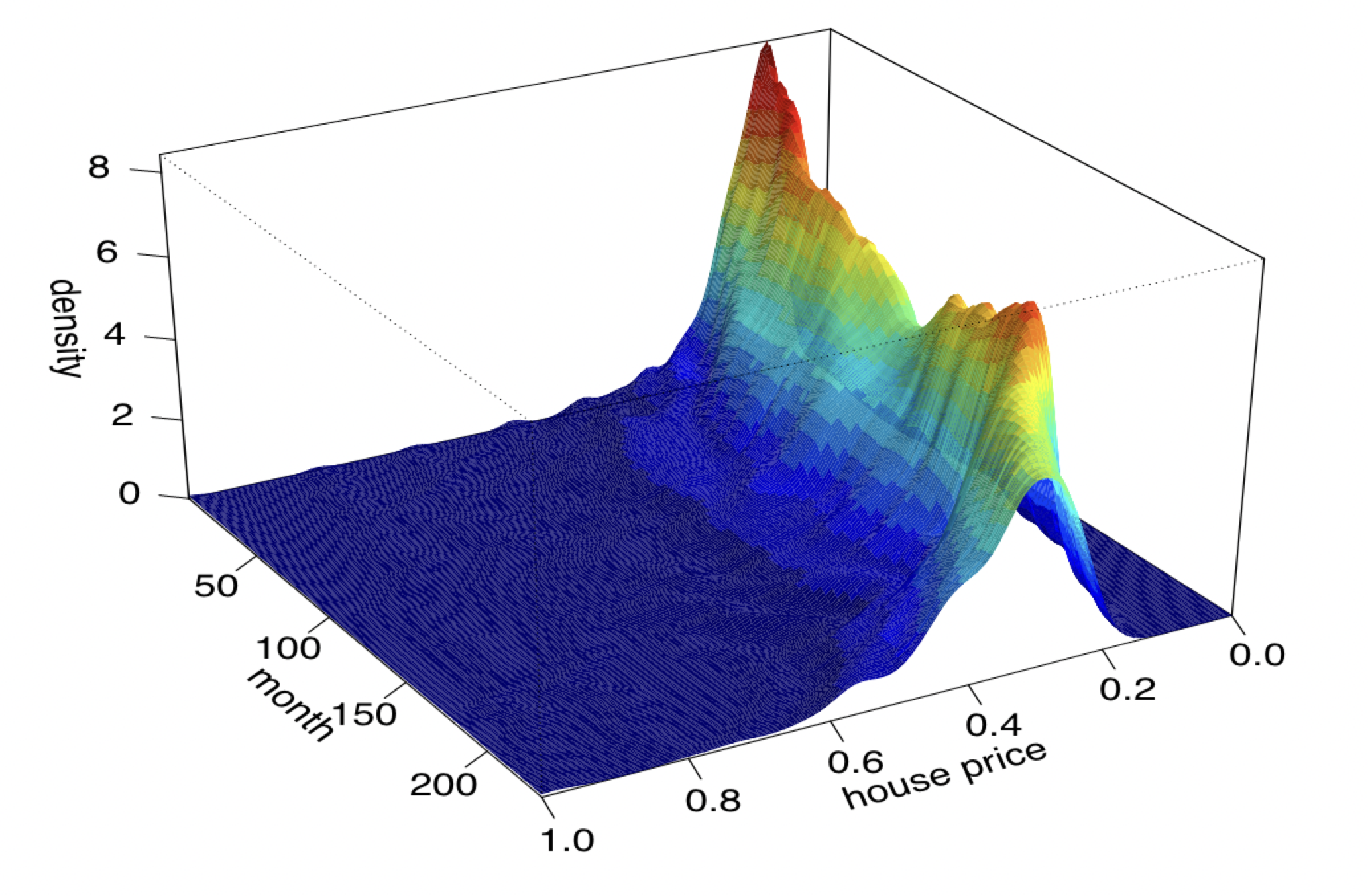}};
		
		%\node[anchor=south west] at (straw.south east) {\includegraphics[width=.48\textwidth, height=.48\textwidth]{StockNormalize.pdf}};
		\end{tikzpicture}
	 \caption{Distributions of US house prices across counties for 240 months between 1996 and 2015, shown as densities. %\red{where the unit on the $X$ axis is xxx USD.}
		}
		\label{fig:real}
	\end{figure}

\section{CONCLUDING REMARKS}

\no Distributional data analysis is challenged by the fact  that distributions do not form a vector space and  basic operations such as addition and multiplication are not available. This especially affects regression models, including  distributional autoregressive models for time series analysis.  At the same time, many time series data can be viewed as sequences of distributional data that are indexed by time and there is a need for more advanced 
statistical tools to model such time series. A key innovation  of this paper is that it provides a novel class of regression models for distributional data that  are intrinsic and enjoy geometric interpretations. These models result  from adopting the point of view that predictors and responses are elements of a space of optimal transports that is equipped with basic algebraic operations.  The existence of stationary solutions of the associated ATM models can be guaranteed if a geometric moment-contraction condition is satisfied. 

\begin{table}[t] \centering 
\begin{tabular}{@{\extracolsep{5pt}} cccccc} 
\\[-1.8ex]\hline 
\hline \\[-1.8ex] 
$k$ &8 & 12 & 18 & 26 & 36 \\ \hline
$\text{ATM}_m$ & $1.878$ & $1.754$ & $1.771$ & $2.715$ & $2.952$   \\ 
$\text{CAT}_m$ &$2.660$ & $2.473$ & $2.345$ & $2.327$ & $2.363$ \\ 
$\text{ATM}_d$ & $\bf{1.647}$ & $\bf{1.611}$ & $\bf{1.652}$ & $\bf{1.708}$ & $\bf{1.778}$ \\
$\text{CAT}_d$ & $1.797$ & $1.787$ & $1.802$ & $1.845$ & $1.924$ \\ 
WR & 4.052 & 3.986 & 4.074 & 4.045 & 4.322 \\
$\text{LQD}$ & $3.405$   & $3.079$   & $2.927$  & $2.730$  & $2.860$  \\
\hline \\[-1.8ex] 
\end{tabular}  
\caption{Comparison of prediction errors  for the US house price distributional time series,  where $k$ is the length of the training set. Actual prediction errors to be multiplied by $10^{-3}$.} 
 \label{tab:house} 
\end{table} 

The proposed models not only provide new ways of modeling distributional time series, but also shed light on the possibility of developing models for time series that take values in other geodesic metric spaces. The  proposed approach  is not limited to optimal transport and other transports that correspond to geodesics with respect to relevant metrics in distribution spaces could similarly be considered, for example Fisher-Rao transports \citep{dai2021statistical}.   %, srivastava2007riemannian, rao1992information}. 
Modeling time series that take values in the space of multivariate distributions will be a challenging future problem; see also the  discussion of this case in \cite{mull:20:7}.

%\vspace{.3cm} 
%\bc {\bf \sf \quad ACKNOWLEDGMENTS} \sm \ec \rs
\section*{ACKNOWLEDGMENTS}

We express our gratitude to the Associate Editor and several referees for very helpful remarks that led to numerous improvements in the paper.\\

%\vspace{.3cm}

%\bc {\bf \sf \quad APPENDIX} \sm \ec \rs
%\section*{APPENDIX}
\appendix
\section{APPENDIX}

%\noindent { \sf A.1 \quad Proof of Theorem 1}
\subsection{Proof of Theorem \ref{thm1}}

\begin{proof} The proof proceeds analogously to that of  Theorem 2 in \cite{wu2004limit}. Set $T$ in equation \eqref{eq:moment} as  $\phi_{\varepsilon_{i-m}}(S_0)$, 
\begin{align*}
E \left[ d_q^{\eta} ( \widetilde{ \phi}_{i, m+1}(S_0), \widetilde{\phi}_{i, m}(S_0)) \right]  \lesssim r^m d_q^{\eta} (  \phi_{\varepsilon_{i-m}}(S_0), S_0)   \lesssim r^m.
\end{align*}
Then $P( d_q ( \widetilde{ \phi}_{i, m+1}(S_0), \widetilde{\phi}_{i, m}(S_0) ) \geq r^{m/(2\eta)} ) \leq r^{m/2}.$ Applying the Borel-Cantelli lemma, we have $ P( d_q ( \widetilde{ \phi}_{i, m+1}(S_0), \widetilde{\phi}_{i, m}(S_0) ) \geq r^{m/(2\eta)}  \text{ infinitely often} ) = 0 $, which entails that $ \widetilde{\phi}_{i,m} (S_0) \rightarrow  \widetilde{T}_{i} $ almost surely due to the completeness of $\mathcal{T}$. Also,
\begin{align*}
E \left[ d_q^{\eta} ( \widetilde{ \phi}_{i, m}(S_0), \widetilde{T}_{i} ) \right] \lesssim E\left[ \sum_{j=0}^{\infty} d_p^{\eta}( \widetilde{ \phi}_{i, m+j+1}(S_0), \widetilde{ \phi}_{i, m+j}(S_0) ) \right] \lesssim r^{m}.
\end{align*}
For any $S \in \mathcal{T}$,
\begin{align*}
E \left[ d_q^{\eta} ( \widetilde{ \phi}_{i, m}(S), \widetilde{T}_{i} ) \right] \lesssim E \left[ d_q^{\eta} ( \widetilde{ \phi}_{i, m}(S_0), \widetilde{T}_{i} ) \right] + E \left[ d_q^{\eta} ( \widetilde{ \phi}_{i, m}(S_0), \widetilde{\phi}_{i, m}(S) ) \right] \lesssim r^m,
\end{align*}
which entails that $ \widetilde{ \phi}_{i, m}(S) \rightarrow  \widetilde{T}_{i} $ almost surely. Clearly, $\{T_i\}_{i \in \mathbb{Z}}$ is stationary and satisfies $T_{i} = \alpha \odot T_{i-1}\oplus \varepsilon_i $. Suppose $\widetilde{S}_i \in L^q(\mathcal{S})$ is another stationary solution, it holds that $\widetilde{S}_i = \widetilde{\phi}_{i, m} (\widetilde{S}_{i-m})$ and $ E[d_q^{\eta}( \widetilde{T}_i, \widetilde{S}_i )] \lesssim r^m E[d_q^{\eta}(\widetilde{T}_{i-m}, \widetilde{S}_{i-m}  )]  $. Since $m$ is arbitrary and $ \widetilde{T}_{i-m}, \widetilde{S}_{i-m} $ are bounded functions, we can conclude that $\widetilde{T}_i = \widetilde{S}_i  $ almost surely in $\mathcal{L}^q(\mathcal{S})$.
\end{proof}%\newpage

%\noindent { \sf A.2 \quad Proof of Theorem 2}
\subsection{Proof of Theorem \ref{thm2}}

\begin{proof}
The estimator $\widehat{\alpha}$ depends on the following quantities 
\begin{align*}
\begin{array}{l}
 \widehat{\rho}^1_{+}  =  \frac{1}{n-1} \sum_{i=2}^n \int_{\mathcal{S}} ( \widehat{T}_i(x) -x )( \widehat{T}_{i-1}(x) -x ) \, d x,  \\
 \widehat{\rho}^1_{-}  =  \frac{1}{n-1} \sum_{i=2}^n \int_{\mathcal{S}} ( \widehat{T}_i(x) -x )(x- \widehat{T}_{i-1}^{-1}(x) )\, d x  \\
 \widetilde{\rho}^1_{+}  = \frac{1}{n-1} \sum_{i=2}^n \int_{\mathcal{S}} ( T_i(x) -x )( T_{i-1}(x) -x ) \,dx, \\
 \widetilde{\rho}^1_{-}  = \frac{1}{n-1} \sum_{i=2}^n \int_{\mathcal{S}} ( T_i(x) -x )(x-  T_{i-1}^{-1}(x)  ) \,dx, \\
 \rho^1_{+}  =  \int_{\mathcal{S}} E[( T_2(x) -x )( T_{1}(x) -x )] \,dx, \\
 \rho^1_{-}  =  \int_{\mathcal{S}} E[( T_2(x) -x )( x- T_{1}^{-1}(x)  )] \,dx, 
 \end{array}
 \end{align*}
 and
 \begin{align*}
 \begin{array}{l}
 \widehat{\rho}^0_{+}  =  \frac{1}{n-1} \sum_{i=2}^n \int_{\mathcal{S}} ( \widehat{T}_{i-1}(x) -x )^2 \, d x,   \\
 \widehat{\rho}^0_{-}  =  \frac{1}{n-1} \sum_{i=2}^n \int_{\mathcal{S}} ( x- \widehat{T}_{i-1}^{-1}(x) )^2 \, dx,   \\
  \widetilde{\rho}^0_{+}  = \frac{1}{n-1} \sum_{i=2}^n \int_{\mathcal{S}} ( T_{i-1}(x) -x )^2 \, dx,  \\
  \widetilde{\rho}^0_{-}  = \frac{1}{n-1} \sum_{i=2}^n \int_{\mathcal{S}} ( x- T_{i-1}^{-1}(x))^2 \, dx,  \\
  \rho^0_{+}  =  \int_{\mathcal{S}} E[( T_{1}(x) -x )^2] \, dx, \\
  \rho^0_{-}  =  \int_{\mathcal{S}} E[( x - T_{1}^{-1}(x)  )^2] \, d x,
\end{array}
\end{align*}
Some algebra shows that $ \widehat{\alpha}_{+} = \widehat{\rho}^1_{+}/\widehat{\rho}^0_{+} $, $ \widehat{\alpha}_{-} = \widehat{\rho}^1_{-}/ \widehat{\rho}^0_{-} $ and $\widehat{\alpha}$ can be written as
\begin{align*}
\widehat{\alpha} =   \widehat{\alpha}_{+}\mathbb{I}_{\{ l_{+}(\widehat{\alpha}_{+}) \leq l_{-}(\widehat{\alpha}_{-}) \}} + \widehat{\alpha}_{-}\mathbb{I}_{\{ l_{+}(\widehat{\alpha}_{+}) > l_{-}(\widehat{\alpha}_{-})\}}  .
\end{align*}
It can be further shown that
\begin{align*}
l_+(\widehat{\alpha}_{+})  & =  \frac{1}{n-1} \sum_{i=2}^n \int_{\mathcal{S}} ( \widehat{T}_{i}(x) -x )^2 dx -  \frac{(\widehat{\rho}^1_{+})^2}{\widehat{\rho}^0_{+}}, \\
l_-(\widehat{\alpha}_{-})  & =  \frac{1}{n-1} \sum_{i=2}^n \int_{\mathcal{S}} ( \widehat{T}_{i}(x) -x )^2 dx -  \frac{(\widehat{\rho}^1_{-})^2}{\widehat{\rho}^0_{-}}, 
\end{align*}
which entails that $ l_+(\widehat{\alpha}_{+}) > l_-(\widehat{\alpha}_{-}) $ iff $ (\widehat{\rho}^1_{+})^2/ \widehat{\rho}^0_{+} <  (\widehat{\rho}^1_{-})^2/\widehat{\rho}^0_{-}$. Since $T_i$ is increasing, the signs of $T_{i}(x) - T_{i}^{-1}(x), x- T_{i}^{-1}(x)$ and $T_i(x) - x$ are always the same and $ |T_{i}(x) - T_{i}^{-1}(x)| \geq |  T_i(x) - x| $, which implies that
\begin{align*}
& 0 < \int_{\mathcal{S}} E[(T_{i}(x)-x)^2] \,dx\leq  \int_{\mathcal{S}} E[(T_{i}(x)-x)(T_{i}(x) - T_i^{-1}(x))] \,dx= \rho_{+}^1+\rho_{-}^1.
\end{align*}
Applying the triangle inequality, 
$$\left| \widehat{\alpha} - \alpha \right| \leq \left| \widehat{\alpha}_{+} - \alpha \right| \mathbb{I}_{\{(\widehat{\rho}^1_{+})^2/ \widehat{\rho}^0_{+} \geq  (\widehat{\rho}^1_{-})^2/\widehat{\rho}^0_{-} \}}   + \left| \widehat{\alpha}_{-} - \alpha \right| \mathbb{I}_{\{ (\widehat{\rho}^1_{+})^2/ \widehat{\rho}^0_{+} <  (\widehat{\rho}^1_{-})^2/\widehat{\rho}^0_{-} \}}.$$ 
Considering  the case when $\alpha \geq 0$, where the case $\alpha < 0$ is analogous, 
\begin{align*}
 \left| \widehat{\alpha}_{+} - \alpha \right|  =  \left| \frac{\widehat{\rho}_{+}^1}{\widehat{\rho}^0_{+}}- \frac{\rho^1_{+}}{\rho^0_{+}} \right|  
 =\left| \frac{\widehat{\rho}_{+}^1 \rho^0_{+} - \rho^1_{+}\widehat{\rho}^0_{+} }{\widehat{\rho}^0_{+} \rho^0_{+}} \right|   \leq   \left| \frac{\widehat{\rho}_{+}^1 \rho^0_{+} - \rho^0_{+}\rho^1_{+}  }{\widehat{\rho}^0_{+} \rho^0_{+}} \right|  +  \left| \frac{ \rho^0_{+}\rho^1_{+} - \rho^1_{+}\widehat{\rho}^0_{+} }{\widehat{\rho}^0_{+} \rho^0_{+}} \right|.
\end{align*}
Since $ \widehat{\rho}^0_{+} \rightarrow \rho^0_{+}$ in probability, $ 1/ \widehat{\rho}^0_{+} \rho^0_{+} = O_p(1)$. Lemma 1 implies that $\left| \widehat{\alpha}_{+} - \alpha \right| = O_p(\tau +  1 /\sqrt{n} )$. We next show that $\left| \widehat{\alpha}_{-} - \alpha \right| \mathbb{I}_{\{ (\widehat{\rho}^1_{+})^2/ \widehat{\rho}^0_{+} <  (\widehat{\rho}^1_{-})^2/\widehat{\rho}^0_{-}\}}$ is asymptotically negligible if the true parameter $\alpha \geq 0$. By Cauchy's inequality, 
\begin{align*}
 (\rho_{+}^1)^2 \rho_{-}^0  =
\alpha^2 (\rho_{+}^0)^2 \rho_{-}^0 \geq \alpha^2 \left(\int_{\mathcal{S}} E[(T_{1}(x)-x)(x - T_{1}^{-1}(x))] \,dx\right)^2 \rho_{+}^0 = (\rho_{-}^1)^2 \rho_{+}^0.
\end{align*}
In addition, due to the assumption $\int_{\mathcal{S}}E[ ( T_{1}(x)-x)^2 ] \,dx\geq C$,  $ (\rho_{+}^1)^2 \rho_{-}^0 - (\rho_{-}^1)^2 \rho_{+}^0$ is bounded below by a constant $c'>0$. From Lemma 1
\begin{align*}
 P\left( (\widehat{\rho}_{-}^1)^2\widehat{\rho}_{+}^0 - (\widehat{\rho}_{+}^1)^2\widehat{\rho}_{-}^0 >  0 \right) & \leq P\left( \left| (\widehat{\rho}_{-}^1)^2\widehat{\rho}_{+}^0- (\widehat{\rho}_{+}^1)^2\widehat{\rho}_{-}^0 -  (\rho_{-}^1)^2\rho_{+}^1 + (\rho_{+}^1)^2\rho_{-}^1 \right| > c' \right) \\ 
 & \lesssim \frac{\tau + 1/\sqrt{n} }{ c'}.
\end{align*}
Since $ |\widehat{\alpha}_{-} - \alpha| $ is bounded from above,  $ \left| \widehat{\alpha}_{-} - \alpha \right|  \mathbb{I}_{\{ (\widehat{\rho}^1_{+})^2/ \widehat{\rho}^0_{+} <  (\widehat{\rho}^1_{-})^2/\widehat{\rho}^0_{-} \} } = O_p( \tau + 1/\sqrt{n} ) $.
\end{proof}

\begin{lemma}\label{lem1}
Under the assumptions of Theorem 1 with $q=1$,
\begin{align*}
 E |\widehat{\rho}_{+}^1 - \rho_{+}^1 | \lesssim \tau + \frac{1}{\sqrt{n}} , \,
 E |\widehat{\rho}_{-}^1 - \rho_{-}^1 | \lesssim \tau + \frac{1}{\sqrt{n}} , \\
 E |\widehat{\rho}^0_{+} - \rho^0_{+} | \lesssim  \tau + \frac{1}{\sqrt{n}}, \, E |\widehat{\rho}^0_{-} - \rho^0_{-} | \lesssim  \tau + \frac{1}{\sqrt{n}}.
\end{align*}
\end{lemma}

\begin{proof}
To study the asymptotic properties of $\widetilde{\rho}_{+}^1, \widetilde{\rho}_{-}^1, \widetilde{\rho}^0_{+}, \widetilde{\rho}^0_{-} $, given random transports $S, T \in \mathcal{T}$, define $g_{+}, g_{-}:\mathcal{T} \rightarrow \mathbb{R}$ as 
\begin{align*}
&g_{+}^0(T) =  \int_{\mathcal{S}} ( T(x) -x )^2 \, dx - E\left[\int_{\mathcal{S}} ( T(x) -x )^2 \, dx \right], \\
& g_{-}^0(T) =  \int_{\mathcal{S}} ( x -T^{-1}(x) )^2 \, dx - E\left[ \int_{\mathcal{S}} ( x -T^{-1}(x) )^2 \, dx \right]
\end{align*}
and $g_{+}^1, g_{-}^1:\mathcal{T} \times \mathcal{T} \rightarrow \mathbb{R}$ as
\begin{align*}
g_{+}^1(S, T)& = \int_{\mathcal{S}} (S(x) - x)( T(x) - x) \,dx- E\left[ \int_{\mathcal{S}} (S(x) - x)( T(x) - x) \,dx\right],\\
g_{-}^1(S, T) &= \int_{\mathcal{S}} (S(x) - x)( x - T^{-1}(x) ) \,dx- E\left[ \int_{\mathcal{S}} (S(x) - x)( x - T^{-1}(x) ) \,dx \right].
\end{align*}
For identically distributed pairs $(S_1, T_1)$ and $(S_2, T_2)$,  
\begin{align*}
    &  | g_{+}^1(S_1, T_1) - g_{+}^1(S_2, T_2) | \mathbb{I}_{\{d_1((S_1, T_1),(S_2, T_2)) \leq \delta \}}   \\
    \lesssim  & \left(  | g_{+}^1(S_1, T_1) - g_{+}^1(S_1, T_2) |     + |g_{+}^1(S_1, T_2) - g_{+}^1(S_2, T_2)| \right)  \mathbb{I}_{\{ d_1((S_1, T_1),(S_2, T_2)) \leq \delta \} }  \\
    \lesssim & \delta,
\end{align*}
where $d_1((S_1, T_1),(S_2, T_2)) = \sqrt{d_1^2(S_1, S_2) + d_1^2(T_1, T_2)}$ is the product metric. This shows that $g_{+}^1$ is stochastic Dini-continuous in view of condition (9) in \cite{wu2004limit}. Similarly, it can be shown that  $g_{+}^0$, $ g_{-}^0$ and $g^1_{-}$ are stochastic Dini-continuous. Note that
\begin{align*}
    \widetilde{\rho}_{+}^1 - \rho_{+}^1 = \frac{1}{n} \sum_{i=2}^{n} g_{+}^1 (T_i, T_{i-1}),
\end{align*}
where $\{g_{+}^1(T_{i}, T_{i-1})\}$ are bounded random variables. Thus, applying Theorem 3 \citep{wu2004limit}, we have
%This may be useful: https://math.stackexchange.com/questions/238031/applying-central-limit-theorem-to-show-that-e-left-fracs-n-sqrtn-right?noredirect=1&lq=1
$ E|\widetilde{\rho}_{+}^1 - \rho_{+}^1 | \lesssim  \frac{1}{\sqrt{n}} $, and, similarly,  
\begin{align*}
E|\widetilde{\rho}_{-}^1 - \rho_{-}^1 | \lesssim  \frac{1}{\sqrt{n}} , \, |\widetilde{\rho}^0_{+} - \rho^0_{+} | \lesssim  \frac{1}{\sqrt{n}} \text{ and } |\widetilde{\rho}^0_{-} - \rho^0_{-} | \lesssim  \frac{1}{\sqrt{n}} .
\end{align*}
Then,  $\left| \widehat{\rho}_{+}^1 - \rho^1_{+} \right| \leq \left| \widehat{\rho}_{+}^1 - \widetilde{\rho}^1_{+} \right| + \left| \rho^1_{+} - \widetilde{\rho}^1_{+} \right|,$ and  %to bound the difference between $\widehat{\rho}_{+}^1$ and $\widetilde{\rho}^1$, 
\begin{multline*}
\left| \int_{\mathcal{S}} ( \widehat{T}_i(x) -x )( \widehat{T}_{i-1}(x) -x )dx - \int_{\mathcal{S}} ( T_i(x) -x )( T_{i-1}(x) -x ) dx \right| \\
\leq  \left| \int_{\mathcal{S}} ( \widehat{T}_i(x) -  T_i(x) ) ( \widehat{T}_{i-1}(x) - T_{i-1}(x) )   dx   \right| 
 +  \left| \int_{\mathcal{S}} ( \widehat{T}_i(x) -  T_i(x) ) ( T_{i-1}(x)-x )   dx   \right| \\
 + \left| \int_{\mathcal{S}} (  T_i(x) -x ) ( \widehat{T}_{i-1}(x) - T_{i-1}(x) )   dx   \right|,
\end{multline*}
where the first term is bounded by the second and third term on the r.h.s. . Since $\mathcal{S}$ is a bounded closed interval, $T_i(x) - x$ is bounded by a constant for all $i$, and 
\begin{align*}
 E \left| \int_{\mathcal{S}} ( \widehat{T}_i(x) -  T_i(x) ) ( T_{i-1}(x)-x )   \,dx  \right| 
\lesssim  E \left[ \int_{\mathcal{S}}  | \widehat{T}_i(x) -  T_i(x) |  \,dx\right] \lesssim  \tau,  
\end{align*}
which entails that $ E \left| \widehat{\rho}_{+}^1 - \widetilde{\rho}_{+}^1 \right| \lesssim \tau $, whence  $ E \left| \widehat{\rho}_{+}^1 - \rho^1_{+} \right| \lesssim (\tau + 1 / \sqrt{n}) $. The results for $\widehat{\rho}_{-}^{1}$, $\widehat{\rho}^0_{+}$ and $\widehat{\rho}^0_{-}$ can be shown similarly.
\end{proof}

%\noindent { \sf A.3 \quad Proof of Theorem 4} 
\subsection{Proof of Theorem \ref{thm4}}

\begin{proof}
Let $L(\alpha_1, \dots, \alpha_p)  = E \left[ \int_{\mathcal{S}} \left( T_{i}(x) - \alpha_p \odot T_{i-p} \oplus  \dots \oplus  \alpha_1 \odot T_{i-1}(x) \right)^2 dx \right]$ , we first show that 
\begin{align}
\label{eq:uniform}
  \sup_{-c \leq \alpha_1, \dots, \alpha_p \leq c}\left|L_n(\alpha_1, \dots, \alpha_p)- L(\alpha_1, \dots, \alpha_p)  \right|  = o_p(1).
\end{align}
To this end, we apply Corollary 3.1 of \cite{whitney1991}. Since $\{ T_i \}$ are bi-Lipschitz with a common Lipschitz constant $K$, it can be easily seen that 
\begin{align*}
    \left| L_n(\alpha_1, \dots, \alpha_p) - L_n(\alpha_1', \dots, \alpha_p') \right| \lesssim |\alpha_1 - \alpha_1'| + \cdots + |\alpha_p - \alpha_p'|.
\end{align*} 

Next, we obtain the pointwise convergence of $L_n$. Fix $(\alpha_1, \dots, \alpha_p)$, set $\mathbf{T}_i = (T_i,  \dots, T_{i-p})$ and $g(\mathbf{T}_i) = \widetilde{g}(\mathbf{T}_i) - E[\widetilde{g}(\mathbf{T}_i)]$, where
\begin{align*}
    \widetilde{g}(\mathbf{T}_i) = \int_{\mathcal{S}} \left( T_{i}(x) - \alpha_p \odot T_{i-p} \oplus  \dots \oplus  \alpha_1 \odot T_{i-1}(x) \right)^2 dx.
\end{align*}
Using the fact that $ \{T_i\}$ are bi-Lipschitz continuous and have bounded derivatives both below and above, it can be shown that
\begin{align*}
    & d_1(\alpha_2 \odot T_{i-2} \oplus \alpha_1 \odot T_{i-1}, \alpha_2 \odot T'_{i-2} \oplus \alpha_1 \odot T'_{i-1} ) \\
    \lesssim & d_1(T_{i-2}, T_{i-2}') + d_1(T_{i-1} \circ ( \alpha_2 \odot T_{i-2}),  T'_{i-1} \circ ( \alpha_2 \odot T'_{i-2})) \\
    \lesssim & d_1(T_{i-2}, T_{i-2}') + d_1(T_{i-1} \circ ( \alpha_2 \odot T_{i-2}),  T_{i-1} \circ ( \alpha_2 \odot T'_{i-2})) \\ 
    &  + d_1(T_{i-1} \circ ( \alpha_2 \odot T'_{i-2}),  T'_{i-1} \circ ( \alpha_2 \odot T'_{i-2})) \\
    \lesssim & d_1(T_{i-1}, T_{i-1}') + d_1(T_{i-2}, T_{i-2}').
\end{align*}
Applying  the above logic recursively, we can show that for $\mathbf{T}_i' \sim^{i.i.d}\mathbf{T}_i$, $ |\widetilde{g}(\mathbf{T}_i) - \widetilde{g}(\mathbf{T}_i')| \lesssim d_{1}(\mathbf{T}_i, \mathbf{T}_i')$, which  implies that $|g(\mathbf{T}_i) - g(\mathbf{T}_i')| \lesssim d_{1}(\mathbf{T}_i, \mathbf{T}_i')$ and then  equation (9) in Theorem 3 of \cite{wu2004limit}. Thus, 
\begin{align*}
    L_n(\alpha_1, \dots, \alpha_p) - L(\alpha_1, \dots, \alpha_p) = \frac{1}{n-p} \sum_{i=p+1}^n g(\mathbf{T}_i) \overset{p}{\rightarrow} 0.
\end{align*}
This proves the pointwise convergence and Corollary 3.1 of \cite{whitney1991} implies \eqref{eq:uniform}. Next, we show that $\Delta(\widetilde{\bm{\alpha}}, \bm{\alpha}^{*})$ converges to 0 in probability. For  $\bm{\alpha} = (\alpha_1, \dots, \alpha_p)^T$, observe that $L(\alpha_1, \dots, \alpha_p) = \Delta(\bm{\alpha}, \bm{\alpha}^{*})+ c_{\Delta}$, where $c_{\Delta} = \int_{\mathcal{S}} [ E(T_{i}^2(x)) - (E (T_{i}(x)))^2 ] dx$. Then 
\begin{align*}
 \Delta(\widetilde{\bm{\alpha}}, \bm{\alpha}^{*}) & =  L(\widetilde{\bm{\alpha}}) - c_{\Delta} \\ 
 & \leq L_n(\widetilde{\bm{\alpha}}) + \sup_{-c \leq \alpha_1, \dots, \alpha_p \leq c} \left|L_n(\alpha_1, \dots, \alpha_p)- L(\alpha_1, \dots, \alpha_p)  \right| - c_{\Delta}, 
   %& \leq \inf_{-c \leq \alpha_1, \dots, \alpha_p \leq c} L_n(\bm{\alpha}) + \sup_{-c \leq \alpha_1, \dots, \alpha_p \leq c} \left|L_n(\alpha_1, \dots, \alpha_p)- L(\alpha_1, \dots, \alpha_p)  \right|.
\end{align*}
and since $ \inf_{\bm{\alpha}} L(\bm{\alpha}) =c_{\Delta}$,  \eqref{eq:uniform} implies that 
$ \inf_{-c \leq \alpha_1, \dots, \alpha_p \leq c} L_n(\bm{\alpha}) \rightarrow^p \inf_{\bm{\alpha}} L(\bm{\alpha}) =c_{\Delta} $ and $\sup_{-c \leq \alpha_1, \dots, \alpha_p \leq c} \left|L_n(\alpha_1, \dots, \alpha_p)- L(\alpha_1, \dots, \alpha_p)  \right| \rightarrow^p 0$. This concludes the proof.
\end{proof}

%\noindent { \sf A.4 \quad Proof of Theorem 6}  
\subsection{Proof of Theorem \ref{thm6}}

\begin{proof}
For convenience, define the  quantities 
\begin{align*}
\begin{array}{l}
 \widehat{\rho}^1_{+}(x)  =  \frac{1}{n-1} \sum_{i=2}^n ( \widehat{T}_i(x) -x )( \widehat{T}_{i-1}(x) -x ) ,  \\
 \widehat{\rho}^1_{-}(x)  =  \frac{1}{n-1} \sum_{i=2}^n  ( \widehat{T}_i(x) -x )(x- \widehat{T}_{i-1}^{-1}(x) )  \\
 \widetilde{\rho}^1_{+}(x)  = \frac{1}{n-1} \sum_{i=2}^n  ( T_i(x) -x )( T_{i-1}(x) -x )  , \\
 \widetilde{\rho}^1_{-}(x)  = \frac{1}{n-1} \sum_{i=2}^n  ( T_i(x) -x )(x-  T_{i-1}^{-1}(x)  )  , \\
 \rho^1_{+}(x)  =   E[( T_2(x) -x )( T_{1}(x) -x )]  , \\
 \rho^1_{-}(x)  =   E[( T_2(x) -x )( x- T_{1}^{-1}(x)  )] , 
 \end{array}
 \end{align*}
 and also
 \begin{align*}
\begin{array}{l}
 \widehat{\rho}^0_{+}(x)  =  \frac{1}{n-1} \sum_{i=2}^n ( \widehat{T}_{i-1}(x) -x )^2 ,   \\
 \widehat{\rho}^0_{-}(x)  =  \frac{1}{n-1} \sum_{i=2}^n ( x-\widehat{T}_{i-1}^{-1}(x) )^2 ,   \\
  \widetilde{\rho}^0_{+}(x)  = \frac{1}{n-1} \sum_{i=2}^n  ( T_{i-1}(x) -x )^2 ,  \\
   \widetilde{\rho}^0_{-}(x)  = \frac{1}{n-1} \sum_{i=2}^n  ( x-T_{i-1}^{-1}(x) )^2 ,  \\
  \rho^0_{+}(x)  =   E[( T_{1}(x) -x )^2] , \\
  \rho^0_{-}(x)  =   E[( x- T_{1}^{-1}(x))^2].
\end{array}
\end{align*}
Then  $ \widehat{\alpha}_{+}(x) = \widehat{\rho}^1_{+}(x)/ \widehat{\rho}^0_{+}(x) $, $ \widehat{\alpha}_{-}(x) = \widehat{\rho}^1_{-}(x)/ \widehat{\rho}^0_{-}(x) $ and  $\widehat{\alpha}(x)$ can be written as
\begin{align*}
\widehat{\alpha}(x) =   \widehat{\alpha}_{+}(x) \mathbb{I}_{\{ l_{+}(\widehat{\alpha}_{+}(x)|x) \leq l_{-}(\widehat{\alpha}_{-}(x)|x) \}} + \widehat{\alpha}_{-}(x)\mathbb{I}_{\{ l_{+}(\widehat{\alpha}_{+}(x)|x) > l_{-}(\widehat{\alpha}_{-}(x)|x)\}}  .
\end{align*}
Similarly, we can show that $ l_+(\widehat{\alpha}_{+}(x)|x) > l_-(\widehat{\alpha}_{-}(x)|x) $ if and only if  $ (\widehat{\rho}^1_{+}(x))^2/\widehat{\rho}^0_{+}(x) <  (\widehat{\rho}^1_{-}(x))^2/\widehat{\rho}^0_{-}(x).$ Then
\begin{multline*}
 \left| \widehat{\alpha}(x) - \alpha(x) \right| \leq  \left| \widehat{\alpha}_{+}(x) - \alpha(x) \right| \mathbb{I}_{\{ (\widehat{\rho}^1_{+}(x))^2/\widehat{\rho}^0_{+}(x) \geq  (\widehat{\rho}^1_{-}(x))^2/\widehat{\rho}^0_{-}(x) \}} \\  +  \left| \widehat{\alpha}_{-}(x) - \alpha(x) \right| \mathbb{I}_{\{ (\widehat{\rho}^1_{+}(x))^2/\widehat{\rho}^0_{+}(x) <  (\widehat{\rho}^1_{-}(x))^2/\widehat{\rho}^0_{-}(x) \}} .
\end{multline*}
When $\alpha \geq 0$, 
\begin{multline*}
 \left| \widehat{\alpha}_{+}(x) - \alpha(x) \right|  =  \left| \frac{\widehat{\rho}_{+}^1(x) }{\widehat{\rho}^0_{+}(x) }- \frac{\rho^1_{+}(x) }{\rho^0_{+}(x) } \right|  
 =\left| \frac{\widehat{\rho}_{+}^1(x) \rho^0_{+}(x) - \rho^1_{+}(x) \widehat{\rho}^0_{+}(x) }{\widehat{\rho}^0_{+}(x) \rho^0_{+}(x) } \right| \\  \leq   \left| \frac{\widehat{\rho}_{+}^1(x) \rho^0_{+}(x) - \rho^0_{+}(x) \rho^1_{+}(x)  }{\widehat{\rho}^0_{+}(x) \rho^0_{+}(x) } \right|  +  \left| \frac{ \rho^0_{+}(x) \rho^1_{+}(x) - \rho^1_{+}(x)\widehat{\rho}^0_{+}(x) }{\widehat{\rho}^0_{+}(x) \rho^0_{+}(x) } \right|.
\end{multline*}
Since $ |\widehat{\rho}^0_{+}(x) - \rho^0_{+}(x)| = o_p(1)$, it holds that $ |1/ \widehat{\rho}^0_{+}(x) \rho^0_{+}(x)| = O_p(1)$. Lemma 2 implies that
\begin{align*}
\left| \widehat{\alpha}_{+}(x) - \alpha(x) \right| = O_p(\tau(x) + 1/\sqrt{n} ).
\end{align*}
Next, we show  $ \left| \widehat{\alpha}_{-}(x) - \alpha(x) \right| \mathbb{I}_{\{ (\widehat{\rho}^1_{+}(x))^2/\widehat{\rho}^0_{+}(x) <  (\widehat{\rho}^1_{-}(x))^2/\widehat{\rho}^0_{-}(x)\}}$  is asymptotically negligible. Note that 
\begin{align*}
(\rho^1_{+}(x))^2\rho_{-}^0(x) = \beta^2(x) (\rho_{+}^0(x))^2 \rho_{-}^0(x)  \geq (\rho_{-}^1(x))^2 \rho_{+}^0(x)
\end{align*}
and that the assumption $ E[ ( T_{1}(x)-x)^2 ] \geq C $ implies that $(\rho_{+}^1(x))^2 \rho_{-}^0 (x)- (\rho_{-}^1 (x))^2\rho_{+}^0(x) \geq c' $ for some constant $c'>0$. Then, from Lemma 2,
\begin{align*}
 & P\left([(\widehat{\rho}_{-}^1(x))^2\widehat{\rho}_{+}^0(x) - (\widehat{\rho}_{+}^1(x))^2\widehat{\rho}_{-}^0(x)] >  0 \right) \\ \leq & P\left( \left| (\widehat{\rho}_{-}^1(x) )^2\widehat{\rho}_{+}^0 (x)- (\widehat{\rho}_{+}^1(x))^2\widehat{\rho}_{-}^0(x) -  (\rho_{-}^1 (x))^2\rho_{+}^0(x) + (\rho_{+}^1(x))^2 \rho_{-}^0 (x) \right| > c' \right) \\
  \lesssim & \frac{\tau(x) + 1/\sqrt{n} }{ c'},
\end{align*}
which entails $ \left| \widehat{\alpha}_{-}(x) - \alpha(x) \right| \mathbb{I}_{\{ (\widehat{\rho}^1_{+}(x))^2/\widehat{\rho}^0_{+}(x) <  (\widehat{\rho}^1_{-}(x))^2/\widehat{\rho}^0_{-}(x)\}}  = O_p( \tau(x) + 1/\sqrt{n} ) $.
\end{proof}

\begin{lemma}\label{lem2}
    Under the  assumptions of Theorem 3,
\begin{align*}
& E[  |\widehat{\rho}_{+}^1(x) - \rho_{+}^1(x) | ] \lesssim \tau(x) + \frac{1}{\sqrt{n}}, \; E[  |\widehat{\rho}_{-}^1(x) - \rho_{-}^1(x) | ] \lesssim \tau(x) + \frac{1}{\sqrt{n}}, \\
& E[  |\widehat{\rho}^0_{+}(x) - \rho^0_{+}(x) | ] \lesssim \tau(x) + \frac{1}{\sqrt{n}}, \; E[  |\widehat{\rho}^0_{-}(x) - \rho^0_{-}(x) | ] \lesssim \tau(x) + \frac{1}{\sqrt{n}}.
\end{align*}
\end{lemma}

\begin{proof}
Here, we only show that $E[| \widehat{\rho}^1_{+}(x) -  \rho_{+}^1(x)|] \lesssim \tau(x) + 1/\sqrt{n}$, the other part is analogous. Observe
$
\left| \widehat{\rho}_{+}^1(x) - \rho^1_{+}(x) \right| \leq  \left| \widehat{\rho}_{+}^1(x) - \widetilde{\rho}^1_{+}(x) \right| +  \left| \rho^1_{+}(x) - \widetilde{\rho}^1_{+}(x) \right|.
$
%to bound $|\widehat{\rho}_{+}^1 - \widetilde{\rho}^1 |$, 
Since  $\mathcal{S}$ is a bounded interval,
\begin{align*}
&  \left| \widehat{\rho}_{+}^1(x) - \widetilde{\rho}^1_{+}(x)  \right| \\
\leq & \frac{1}{n-1} \sum_{i=2}^n  \left|  ( \widehat{T}_i(x) -x )( \widehat{T}_{i-1}(x) -x ) - ( T_i(x) -x )( T_{i-1}(x) -x )  \right| \\
\lesssim & \frac{1}{n-1} \sum_{i=2}^n   \left|   \widehat{T}_i(x) -  T_i(x)      \right|  + \frac{1}{n-1} \sum_{i=2}^n  \left|   \widehat{T}_{i-1}(x) - T_{i-1}(x)      \right|,
\end{align*}
which entails that $E \left[ \left| \widehat{\rho}_{+}^1(x) - \widetilde{\rho}_{+}^1(x) \right| \right] \lesssim   \tau(x)$. We proceed to demonstrate  that $E\left[ \left|\widetilde{\rho}_{+}^1(x) - \rho_{+}^1(x) \right| \right] \lesssim 1/\sqrt{n}$. Condition \eqref{eq:momentinfty} implies that for any fixed $x \in \mathcal{S}$, 
\begin{align*}
E \left[ |  \widetilde{ \varphi}_{i, i}(T_0(x))- \widetilde{\varphi}_{i, i}(T_0'(x)) | \right] \lesssim  r^{i},
\end{align*}
where $T_0'$ is an independent copy of $T_0$. Given random transports $S,T \in \tilde{\mathcal{T}}$, define $g_{+}^1:\mathcal{T} \times \mathcal{T} \rightarrow \mathbb{R}$ as 
\begin{align*}
g_{+}^1 (S,T) = (S(x)-x)(T(x)-x) - E[(S(x)-x)(T(x)-x)].
\end{align*}
Since 
\begin{align*}
\widetilde{\rho}_{+}^1(x) - \rho_{+}^1(x) = \frac{1}{n-1} \sum_{i=2}^n g_{+}^1 (T_i,T_{i-1})
\end{align*}
and $g_{+}^1$ is stochastic Dini-continuous with respect to $d_{\infty}$, it  follows from Theorem 2  \citep{wu2004limit} that $E[| \widetilde{\rho}^1_{+}(x) -  \rho_{+}^1(x)|] \lesssim 1/\sqrt{n}$. 

\end{proof}

%\references

\bibliographystyle{rss}
%\bibliography{ATM2023-1}

\begin{thebibliography}{31}
\expandafter\ifx\csname natexlab\endcsname\relax\def\natexlab#1{#1}\fi
\expandafter\ifx\csname url\endcsname\relax
  \def\url#1{\texttt{#1}}\fi
\expandafter\ifx\csname urlprefix\endcsname\relax\def\urlprefix{URL }\fi

\bibitem[{Bekierman and Gribisch(2021)}]{bekierman2016mixed}
Bekierman, J. and Gribisch, B. (2021) A mixed frequency stochastic volatility
  model for intraday stock market returns.
\newblock \textit{Journal of Financial Econometrics}, \textbf{19}, 496--530.

\bibitem[{Bhatia and Katz(2021)}]{bhat:21}
Bhatia, A. and Katz, J. (2021) Why we are experiencing so many unusually hot
  summer nights.
\newblock \textit{The New York Times}, \textbf{September 16}, A12.

\bibitem[{Bigot \textit{et~al.}(2017)Bigot, Gouet, Klein and
  L{\'o}pez}]{bigo:17}
Bigot, J., Gouet, R., Klein, T. and L{\'o}pez, A. (2017) Geodesic {PCA} in the
  {W}asserstein space by convex {PCA}.
\newblock \textit{Annales de l{'}Institut Henri Poincar{\'e} B: Probability and
  Statistics}, \textbf{53}, 1--26.

\bibitem[{Bogin \textit{et~al.}(2019)Bogin, Doerner and
  Larson}]{bogin2019local}
Bogin, A., Doerner, W. and Larson, W. (2019) Local house price dynamics: new
  indices and stylized facts.
\newblock \textit{Real Estate Economics}, \textbf{47}, 365--398.

\bibitem[{Bosq(2000)}]{bosq:00}
Bosq, D. (2000) \textit{Linear Processes in Function Spaces: Theory and
  Applications}.
\newblock New York: Springer-Verlag.

\bibitem[{Chen \textit{et~al.}(2022)Chen, Lin and M{\"u}ller}]{mull:20:7}
Chen, Y., Lin, Z. and M{\"u}ller, H.-G. (2022) Wasserstein regression.
\newblock \textit{Journal of the American Statistical Association}, to appear.

\bibitem[{Cheng and Parzen(1997)}]{cheng1997unified}
Cheng, C. and Parzen, E. (1997) Unified estimators of smooth quantile and
  quantile density functions.
\newblock \textit{Journal of Statistical Planning and Inference}, \textbf{59},
  291--307.

\bibitem[{Dai(2022)}]{dai2021statistical}
Dai, X. (2022) Statistical inference on the {H}ilbert sphere with application
  to random densities.
\newblock \textit{Electronic Journal of Statistics}, \textbf{16}, 700--736.

\bibitem[{Diaconis and Freedman(1999)}]{diaconis1999iterated}
Diaconis, P. and Freedman, D. (1999) Iterated random functions.
\newblock \textit{SIAM Review}, \textbf{41}, 45--76.

\bibitem[{Falk(1983)}]{falk1983relative}
Falk, M. (1983) Relative efficiency and deficiency of kernel type estimators of
  smooth distribution functions.
\newblock \textit{Statistica Neerlandica}, \textbf{37}, 73--83.

\bibitem[{Ghodrati and Panaretos(2022)}]{ghod:21}
Ghodrati, L. and Panaretos, V.~M. (2022) {Distribution-on-distribution
  regression via optimal transport maps}.
\newblock \textit{Biometrika}, \textbf{109}, 957--974.

\bibitem[{Kloeckner(2010)}]{kloe:10}
Kloeckner, B. (2010) A geometric study of {W}asserstein spaces: {E}uclidean
  spaces.
\newblock \textit{Ann. Scuola Norm. Sup. Pisa Cl. Sci}, \textbf{IX}, 297--323.

\bibitem[{Kokoszka \textit{et~al.}(2019)Kokoszka, Miao, Petersen and
  Shang}]{kokoszka2019forecasting}
Kokoszka, P., Miao, H., Petersen, A. and Shang, H.~L. (2019) Forecasting of
  density functions with an application to cross-sectional and intraday
  returns.
\newblock \textit{International Journal of Forecasting}, \textbf{35},
  1304--1317.

\bibitem[{Leblanc(2012)}]{leblanc2012estimating}
Leblanc, A. (2012) On estimating distribution functions using {B}ernstein
  polynomials.
\newblock \textit{Annals of the Institute of Statistical Mathematics},
  \textbf{64}, 919--943.

\bibitem[{Matabuena and Petersen(2021)}]{mata:21}
Matabuena, M. and Petersen, A. (2021) Distributional data analysis with
  accelerometer data in a nhanes database with nonparametric survey regression
  models.
\newblock \textit{arXiv preprint arXiv:2104.01165}.

\bibitem[{Mazzuco and Scarpa(2015)}]{mazzuco2015fitting}
Mazzuco, S. and Scarpa, B. (2015) {Fitting age-specific fertility rates by a
  flexible generalized skew normal probability density function}.
\newblock \textit{Journal of the Royal Statistical Society Series A},
  \textbf{178}, 187--203.

\bibitem[{McCann(1997)}]{mcca:97}
McCann, R.~J. (1997) A convexity principle for interacting gases.
\newblock \textit{Advances in Mathematics}, \textbf{128}, 153--179.

\bibitem[{Menafoglio \textit{et~al.}(2018)Menafoglio, Grasso, Secchi and
  Colosimo}]{mena:18}
Menafoglio, A., Grasso, M., Secchi, P. and Colosimo, B.~M. (2018) Profile
  monitoring of probability density functions via simplicial functional {PCA}
  with application to image data.
\newblock \textit{Technometrics}, \textbf{60}, 497--510.

\bibitem[{Newey(1991)}]{whitney1991}
Newey, W.~K. (1991) Uniform convergence in probability and stochastic
  equicontinuity.
\newblock \textit{Econometrica}, \textbf{59}, 1161--1167.

\bibitem[{Oikarinen \textit{et~al.}(2018)Oikarinen, Bourassa, Hoesli and
  Engblom}]{oikarinen2018us}
Oikarinen, E., Bourassa, S.~C., Hoesli, M. and Engblom, J. (2018) {US}
  metropolitan house price dynamics.
\newblock \textit{Journal of Urban Economics}, \textbf{105}, 54--69.

\bibitem[{Ouellette and Bourbeau(2011)}]{ouellette2011changes}
Ouellette, N. and Bourbeau, R. (2011) Changes in the age-at-death distribution
  in four low mortality countries: a nonparametric approach.
\newblock \textit{Demographic Research}, \textbf{25}, 595--628.

\bibitem[{Panaretos and Zemel(2016)}]{panaretos2016amplitude}
Panaretos, V.~M. and Zemel, Y. (2016) Amplitude and phase variation of point
  processes.
\newblock \textit{Annals of Statistics}, \textbf{44}, 771--812.

\bibitem[{Pegoraro and Beraha(2021)}]{pego:21}
Pegoraro, M. and Beraha, M. (2021) Projected statistical methods for
  distributional data on the real line with the {W}asserstein metric.
\newblock \textit{arXiv preprint arXiv:2101.09039}.

\bibitem[{Petersen and M{\"u}ller(2016)}]{petersen2016functional}
Petersen, A. and M{\"u}ller, H.-G. (2016) Functional data analysis for density
  functions by transformation to a {H}ilbert space.
\newblock \textit{The Annals of Statistics}, \textbf{44}, 183--218.

\bibitem[{Petersen and M\"{u}ller(2019)}]{mull:19:3}
Petersen, A. and M\"{u}ller, H.-G. (2019) Fr\'echet regression for random
  objects with {E}uclidean predictors.
\newblock \textit{Annals of Statistics}, \textbf{47}, 691--719.

\bibitem[{Shang and Hyndman(2017)}]{shang2017grouped}
Shang, H.~L. and Hyndman, R.~J. (2017) Grouped functional time series
  forecasting: an application to age-specific mortality rates.
\newblock \textit{Journal of Computational and Graphical Statistics},
  \textbf{26}, 330--343.

\bibitem[{Shorack and Wellner(2009)}]{shor:09}
Shorack, G.~R. and Wellner, J.~A. (2009) \textit{Empirical Processes with
  Applications to Statistics}.
\newblock SIAM.

\bibitem[{Villani(2003)}]{vill:03}
Villani, C. (2003) \textit{Topics in Optimal Transportation}.
\newblock American Mathematical Society.

\bibitem[{Wu and Shao(2004)}]{wu2004limit}
Wu, W.~B. and Shao, X. (2004) Limit theorems for iterated random functions.
\newblock \textit{Journal of Applied Probability}, \textbf{41}, 425--436.

\bibitem[{Zhang \textit{et~al.}(2022)Zhang, Kokoszka and
  Petersen}]{zhang2020wasserstein}
Zhang, C., Kokoszka, P. and Petersen, A. (2022) Wasserstein autoregressive
  models for density time series.
\newblock \textit{Journal of Time Series Analysis}, \textbf{43}, 30--52.

\bibitem[{Zivot and Wang(2007)}]{zivot2007modeling}
Zivot, E. and Wang, J. (2007) \textit{Modeling financial time series with
  S-Plus{\textregistered}}, vol. 191.
\newblock Springer Science \& Business Media.

\end{thebibliography}

\end{document}